\documentclass[10pt,a4paper]{article}
\pdfoutput=1
\usepackage{jheppub}
\usepackage{arydshln}
\usepackage{tikz}
\usetikzlibrary{positioning}
\usepackage{lipsum}
\usepackage{parskip}
\usepackage[all]{xy}
\usepackage{amsfonts}
\usepackage{amscd}
\usepackage{amssymb}
\usepackage{amsmath,bbm}
\usepackage{graphicx}
\usepackage{epsfig}
\usepackage{latexsym}
\usepackage{mathtools}
\usepackage{hyperref}
\usepackage{amsthm}
\usepackage[vcentermath]{youngtab}
\usepackage{accents}
    
\usepackage{calligra}

\DeclareMathAlphabet{\mathcalligra}{T1}{calligra}{m}{n}
\DeclareFontShape{T1}{calligra}{m}{n}{<->s*[2.2]callig15}{}
\newcommand{\scriptr}{\mathcalligra{r}\,}

\newtheorem{conjecture}{Conjecture}

\def \be  {\begin{equation}}
\def \ee  {\end{equation}}
\def \bea {\begin{equation}\begin{aligned}}
\def \eea {\end{aligned}\end{equation}}
\def \ba  {\begin{eqnarray}}
\def \ea  {\end{eqnarray}}
\def \bb  {\begin{thebibliography}}
\def \eb  {\end{thebibliography}}
\def \lab #1 {\label{#1}}

\newcommand\cD{\mathcal{D}}

\newcommand\cG{\mathcal{G}}

\newcommand\cJ{\mathcal{J}}

\newcommand\cM{\mathcal{M}}
\newcommand\cN{\mathcal{N}}

\newcommand\cU{\mathcal{U}}
\newcommand\cV{\mathcal{V}}

\newcommand\gf{\mathfrak{g}}
\newcommand\hf{\mathfrak{h}}
\newcommand\glf{\mathfrak{gl}}
\newcommand\slf{\mathfrak{sl}}
\newcommand\pslf{\mathfrak{psl}}

\newcommand{\ph}[1]{\phantom{-}#1}
\newcommand{\hhat}{\widehat{\mathfrak h}}

\newcommand{\wt}[1]{\widetilde{#1}}

\newcommand\tr{\mathrm{Tr}}

\newcommand{\cf}{\emph{cf}.}
\newcommand{\ie}{\emph{i.e.}}

\definecolor{cardinal}{rgb}{0.6,0,0}
\definecolor{darkgreen}{rgb}{0,0.5,0}
\definecolor{golden}{rgb}{0.92, 0.7, 0}
\definecolor{midnight}{rgb}{0, 0, 0.5}
\definecolor{darkblue}{rgb}{0.2, 0, 0.8}

\newtheorem{theorem}{Theorem}

\theoremstyle{definition}
\newtheorem{definition}{Definition}[section]

\newif\ifniklas\niklastrue
\RequirePackage{verbatim}

\makeatletter
\newwrite\bibinl@out
{\immediate\closeout\bibinl@out\@esphack}
\makeatother

\newtheorem{lemma}{Lemma}

\title{Free field realisation of boundary vertex algebras\\ for Abelian gauge theories in three dimensions}
\author[1]{Christopher Beem}
\affiliation[1]{Mathematical Institute, University of Oxford\\
Radcliffe Observatory Quarter, Woodstock Road, Oxford OX2 6GG, United Kingdom}
\emailAdd{christopher.beem@maths.ox.ac.uk}
\author[2]{and Andrea E. V. Ferrari}
\affiliation[2]{Department of Mathematical Sciences, Durham University\\
Upper Mountjoy, Stockton Road, Durham DH1 3LE, United Kingdom}
\emailAdd{andrea.ferrari@durham.ac.uk}

\abstract{We study the boundary vertex algebras of $A$-twisted $\cN=4$ Abelian gauge theories in three dimensions. These are identified with the BRST quotient (semi-infinite cohomology) of collections of symplectic bosons and free fermions that reflect the matter content of the corresponding gauge theory. We develop various free field realisations for these vertex algebras which we propose to interpret in terms of their localisation on their associated varieties. We derive the free field realisations by bosonising the elementary symplectic bosons and free fermions and then calculating the relevant semi-infinite cohomology, which can be done systematically. An interesting feature of our construction is that for certain preferred free field realisations, the outer automorphism symmetry of the vertex algebras in question (which are identified with the symmetries of the Coulomb branch in the infrared) are made manifest.}

\begin{document}
\today
\maketitle

\newpage

\section{\label{sec:intro}Introduction}

Topologically twisted $\cN=4$ gauge theories in three dimensions admit holomorphic boundary conditions that support boundary vertex operator algebras (VOAs). Constructions of these VOAs were announced in~\cite{Gaiotto:2016wcv,Gaiotto:2017euk} and subsequently elaborated upon in~\cite{Costello:2018fnz,Costello:2018swh} (see also \cite{Creutzig:2021ext,Ballin:2022rto,Garner:2022vds,Garner:2023pmt} for related developments related to these vertex algebras, and in particular the recent work \cite{Yoshida:2023wyt} which studies some of the same examples that are discussed here).

Three-dimensional $\cN=4$ theories admit two topological twists, the $A$- (also H-) and the $B$- (also C-)twists. In the $A$-twist, standard (perturbative) boundary conditions that are compatible with the VOA construction are deformations of $(0,4)$ Neumann boundary conditions. Alternatively in the $B$-twist there are standard perturbative boundary conditions given by deformations of Dirichlet boundary conditions. The two choices of twist/holomorphic boundary condition lead to very different descriptions for the (distinct) associated VOAs: in the $A$-twist they are BRST reductions of symplectic bosons and boundary fermions, while in the $B$-twist they are presented as (essentially) WZW models based on Lie superalgebras. We will refer to the former as ``$A$-twisted VOAs'' and to the latter (though we rarely mention them) as ``$B$-twisted VOAs''. (For brevity, we sometimes use VOA also in the case of vertex operator superalgebras.)

The main aim of the present paper is to elucidate the $A$-twisted VOAs for a large class of Abelian gauge theories by providing them with free field realisations. Our initial approach was inspired by the constructions of free field realisations for VOAs associated to four-dimensional $\cN=2$ SCFTs in, \emph{e.g.}, \cite{Beem:2019tfp,Beem:2021jnm,Beem:2022vfz}. Those free field realisations have a distinctly geometric flavour---they exploit the geometry of the Higgs branch of the theories in question, namely the existence of certain (Zariski) open subsets that can be identified (as holomorphic symplectic varieties) with simple cotangent bundles of the form $(T^\ast \mathbb{C}^\times)^m\times (T^\ast \mathbb{C}^\times)^n$ for some non-negative integers $m$ and $n$. A motivation for those constructions is the (speculative) expectation that the vertex algebras $\cV$ associated to four-dimensional SCFTs should be realisable as the algebras of global sections of some kind of sheaf of vertex algebras on the \emph{associated variety} $X_\cV$. (See, \emph{e.g.}, \cite{Arakawa:2011vpo} for rigorous work in a similar direction but in a somewhat distinct setting; we do not pursue any such careful analysis here.)

We recall (briefly) the notion of the associated variety of a VOA $\mathcal{V}$ \cite{Arakawa:2010ni}. Let $C_2(\cV)$ denote the vector subspace of $\cV$ spanned by all states that can be expressed as a normally ordered product of operators at least one of which has been differentiated. This is an ideal with respect to the normally ordered product, and the quotient,
\begin{equation}
    R_\cV \coloneqq \mathcal{V}/C_2(\mathcal{V})~,
\end{equation}
inherits the structure of a (super)commutative $\mathbb{C}$-algebra dubbed \emph{Zhu's $C_2$ algebra} \cite{zhu1990vertex}. (In fact, it further inherits from the simple pole in the OPE for $\cV$ a Poisson structure). Discarding nilpotent elements of this algebra, one defines the associated variety,
\begin{equation}
    X_\cV \coloneqq {\rm Spec}\,\left(R_\cV\right)_{\rm red}~,
\end{equation}
which is in general a complex, affine Poisson variety.

In the context of VOAs arising in four dimensions, the \emph{Higgs branch conjecture} of \cite{Beem:2017ooy} postulates the identification of the associated variety with the Higgs branch of vacua of the corresponding SCFT. This in particular would imply that the vertex algebras arising from four-dimensional SCFTs always have associated varieties that are not only Poisson varieties, but symplectic in the sense of having only a finite number of symplectic leaves---such vertex algebras have been designated \emph{quasi-lisse} \cite{Arakawa:2016hkg}.

Though it is surely expected by many experts (the expectation is briefly articulated in~\cite{Costello:2018swh}), to our knowledge the generalisation of this conjecture to three dimensions has not been well advertised. Indeed, in three dimensions there is potentially more variation possible as there is in principle a choice of (deformable) boundary condition involved in the specification of a boundary vertex algebra. Conservatively, we expect that at least in the general class of theories under consideration here an analogous result will hold:\footnote{The analogous statement for $B$-twisted vertex algebras and the Coulomb branch also seems very plausible.}
\begin{conjecture}[$A$-twisted Higgs Branch Conjecture]\label{conj:ass-var}
    For the boundary vertex algebras of $A$-twisted gauge theories with purely free fermion boundary degrees of freedom, the associated variety and the Higgs branch of the corresponding gauge theory are isomorphic as Poisson varieties.
\end{conjecture}
If this is indeed correct, then these constructions in three dimensions will provide another rich source of quasi-Lisse vertex algebras to complement those originating in four dimensions.

Partly in light of the conjecture above, we develop in this work free field realisations for $A$-twisted boundary vertex algebras for Abelian gauge theories that are motivated by the same philosophy of localisation on the Higgs branch/associated variety. For relatively simple theories (with simple Higgs branches given by, \emph{e.g.}, Kleinian singularities of type $A$), we have found that the procedures employed while studying four-dimensional vertex algebras, though somewhat \emph{ad hoc}, can be successfully applied in the present setting as well. The principle technical difference is now the presence of free fermions in the free field realisations. (These free fermions trace their existence back to the boundary degrees of freedom in the gauge theory description; from the point of view of the free field realisation, they play an analogous role to that of the ``residual degrees of freedom'' on the Higgs branch in the four-dimensional setting).

For general gauge theories such a case-by-case analysis would become burdensome. Fortunately, we find that the realisations in question can alternatively be \emph{derived from first principles} by exploiting the relative simplicity of semi-infinite cohomology for Abelian gauge groups. The trick, so to speak, is to bosonise the elementary symplectic bosons and free fermions and embed the entire calculation into the context of a lattice vertex superalgebra. An appropriate change of basis then allows the action of the gauge symmetry to be sequestered from the remaining degrees of freedom, and the corresponding semi-infinite cohomology computation trivialises.

\emph{A posteriori}, we can understand our bosonisation approach to constructing these vertex algebras as a chiral version of a similar method for studying Higgs branch. For Abelian $\cN=4$ gauge theories in three dimensions, the Higgs branch is a hypertoric variety (\emph{cf.} \cite{bielawski2000geometry,proudfoot2004hyperkahler,proudfoot2008survey}). These are simple examples of hyperk\"ahler quotients (in particular, holomorphic symplectic quotients) and their geometry has been studied extensively and is well understood. Bosonisation of a symplectic boson vertex algebra amounts to the embedding
\begin{equation}
    \cD^{ch}(\mathbb{C})\hookrightarrow \cD^{(ch)}(\mathbb{C}^\times)~,
\end{equation}
where $\cD^{(ch)}(X)$ denotes global sections of the sheaf of chiral differential operators on $X$ \cite{MelnikovCDO}, and $\cD^{ch}(\mathbb{C})$ can be thought of as the chiralisation of the algebra of functions on $T^\ast\mathbb{C}$. In this identification the base $\mathbb{C}$ is the space parameterised by one of two scalars in the hypermultiplet. In terms of Higgs branches, the analogous replacement is therefore
\begin{equation}
    \mathbb{C}[V]\hookrightarrow \mathbb{C}[V^\ast]~,
\end{equation}
where $V\cong T^\ast\mathbb{C}^N$ is the representation space for the hypermultiplet scalars in the theory, and $V^\ast\cong(T^\ast\mathbb{C}^\times)^N$. Performing symplectic reduction on $V^\ast$ in general produces an open subset of the very simple form,
\begin{equation}\label{eq:general_resolved_embedding}
    (T^\ast\mathbb{C}^\times)^{\dim_{\mathbb{H}}\cM_H}\hookrightarrow \cM_{H,\zeta}~,
\end{equation}
where $\cM_{H,\zeta}$ is a \emph{resolution of} the Higgs branch $\cM_H$.

There are a multiplicity of inequivalent ways of identifying $V$ as a cotangent bundle of $\mathbb{C}^N$, and correspondingly there are multiple bosonisations of the free fields which lead to inequivalent free field realisations of the same $A$-twisted vertex algebras. An interesting phenomenon---which we can already observe in our simplest example---is that for certain choices of free field realisation, the outer automorphism symmetry of the $A$-twisted vertex algebra (corresponding to the Coulomb branch flavour symmetry of the same $\cN=4$ theory) lifts to a symmetry of the free field vertex algebra and manifestly acts on the relevant subalgebra. On the other hand, for the general realisation this may not be the case. We show that outer automorphism covariance is a consequence of the free field realisation in question being associated to an embedding \eqref{eq:general_resolved_embedding} that instead targets the singular (unresolved) Higgs branch. This aligns nicely with the physical fact that general Higgs branch resolutions break the Coulomb branch symmetry group to its maximal torus.

The structure of the remainder of the paper is as follows. In Section~\ref{sec:tsu2} we introduce our main example, which is the $A$-twisted vertex algebra for the theory $T[{\rm SU}(2)]$. We develop this example in a narrative fashion, starting out with an \emph{ad hoc} construction of the associated vertex algebra $V_1(\pslf(2|2))$ in the spirit of previous four-dimensional work, before re-deriving the same result in a systematic fashion by directly computing the relevant semi-infinite cohomology in terms of bosonised degrees of freedom. We also study a second free field realisation associated with an alternative localisation of the Higgs branch, but which does not exhibit manifest outer automorphism symmetry. In Section~\ref{sec:general_abelian}, we generalise the methods used in the main example to apply to all Abelian gauge theories with unimodular charge matrices, thus giving a completely systematic procedure for producing (multiple) free field realisations for all of the corresponding $A$-twisted vertex algebras. In Section \ref{sec:fin-geom} we develop the relationship of our various free field realisations to an analogous construction purely in terms of Higgs branches, and identify a connection between the preferred localisations that give rise to outer automorphism-invariant free field realisations and certain good open charts on the \emph{singular} Higgs branch. In Section~\ref{sec:voa_examples}, we illustrate our general procedure in a large class of concrete examples, whose Higgs branches are $A$-type Kleinian singularities and $A$-type minimal nilpotent orbit closures, respectively.

\medskip

\paragraph{Remark.} A closely related set of vertex algebras has been considered previously by Kuwabara in \cite{kuwabara2021vertex}. In that work, starting with the same set of symplectic bosons, the author includes auxiliary, positive-level Heisenberg current algebras rather than (physically motivated) free fermions to cancel the anomaly in the BRST operator of the na\"ive gauge theory degrees of freedom. Amongst other things, this has the affect of changing the resulting vertex algebras so they are not quasi-lisse. Furthermore, in that work the author calculates the relevant semi-infinite cohomology by employing a sophisticated localisation construction involving sheaves of $\hbar$-adic vertex algebras tailored to a particular choice of resolution of the associated hypertoric variety. We expect that a similar approach in the setting involving free fermions would reproduce a subset of our free field realisations, though perhaps not the preferred realisations that manifest the Coulomb branch outer automorphisms.

\section{\label{sec:tsu2}Main example}

To develop our approach, we will focus in the first instance on a simple but instructive example: the self-mirror theory $T[\mathrm{SU}(2)]$. Our treatment will be somewhat idiosyncratic, so as to lend itself to generalisation in Section~\ref{sec:general_abelian}. 

The theory $T[\mathrm{SU}(2)]$ is realised by the rank-one Abelian gauge theory (so with gauge group $T_G=\mathrm{U}(1)$) with matter hypermultiplets transforming in the quaternionic representation $V=T^\ast \mathbb{C}^2$, with $\mathbb{C}^2$ representing to two copies of the standard (charge-one) representation. We will encode this choice of matter representation in a charge vector $Q=(1,1)$. We denote by $T_V\cong \mathrm{U}(1)^2$ the maximal torus of the symplectic automorphism group of $V$, which we take to act with charges $(1,0)$ and $(0,1)$ on the base directions of $V$ and with the opposite charges on the cotangent fibres. The flavour torus is given by $T_H=T_V/T_G\cong \mathrm{U}(1)$. We choose representatives for this coset to act with charges $q=(0,1)$ on the base directions of $V$ (and correspondingly $(0,-1)$ on the fibre directions). 

Collectively, the symmetry data described above can be encoded in a square matrix as follows,
\begin{equation}
    \mathbf{Q} \coloneqq 
    \begin{pmatrix} Q \\ q \end{pmatrix} 
    =
    \begin{pmatrix} 1 & 1 \\ 0 & 1 \end{pmatrix}~.
\end{equation}
Here we have implicitly chosen bases for the weight lattices $\mathfrak{t}^\ast_{G,\mathbb{Z}}$, $\mathfrak{t}^\ast_{V,\mathbb{Z}}$, $\mathfrak{t}^\ast_{H,\mathbb{Z}}$ in the (dual) Lie algebras $\mathfrak{t}^\ast_{G,\mathbb{R}}$, $\mathfrak{t}^\ast_{V,\mathbb{R}}$, $\mathfrak{t}^\ast_{H,\mathbb{R}}$ of $T_G$, $T_V$, $T_H$, as well as a choice of splitting (determined by $q^T$, \ie, the choice of representatives for the action of the flavour torus) of the short exact sequence,
\begin{equation}\label{eq:TSU2-SES}
    0 \rightarrow \mathfrak{t}_{G,\mathbb{Z}} \cong \mathbb{Z} \xrightarrow{Q^T}  \mathfrak{t}_{V,\mathbb{Z}} \cong \mathbb{Z}^2 \xrightarrow{\wt{Q}}  \mathfrak{t}_{H,\mathbb{Z}} \cong \mathbb{Z} \rightarrow 0~.
\end{equation}
The map $\wt{Q} = (1,-1)$ in turn encodes the weights of the $\wt{T}_G\cong \mathrm{U}(1)$ gauge action in a mirror Abelian gauge theory, while the mirror flavour weights are encoded in the projection matrix $\mathfrak{t}_{V,\mathbb{Z}} \xrightarrow{\tilde q} \mathfrak{t}_{G,\mathbb{Z}}$. Combining again, one reads off the symmetry data of the mirror theory as follows,
\begin{equation}\label{eq:mirror-sqed2-def}
    \wt{\mathbf{Q}} \coloneqq 
    \begin{pmatrix}
        \wt{q}\\ \wt{Q}
    \end{pmatrix} 
    = \left(\mathbf{Q}^{-1}\right)^{T} =
    \begin{pmatrix}
        \ph{1} & 0 \\ -1 & 1
    \end{pmatrix}~.
\end{equation}
By a symplectic automorphism of $T^\ast \mathbb{C}^2$ and a change of basis for $\mathfrak{t}^\ast_{G,\mathbb{Z}}$, this is equivalent to the data appearing in \eqref{eq:TSU2-SES} and we observe the self-mirror duality of this theory.

\subsection{\label{subsec:tsu2_higgs_coulomb}\texorpdfstring{$T[\mathrm{SU}(2)]$}{T[\mathrm{SU}(2)]} Higgs and Coulomb branches}

We briefly recall the description of the Higgs and Coulomb branches of the moduli space of vacua of $T[\mathrm{SU}(2)]$. As a consequence of $\mathcal{N}=4$ supersymmetry, these are both hyperk\"ahler and in particular holomorphic symplectic. (The hyperk\"ahler structure will play no significant role in this paper.) As the theory is self-mirror, they are in fact isomorphic as symplectic varieties.

\subsubsection{\label{subsubsec:TSU2_Higgs}\texorpdfstring{$T[\mathrm{SU}(2)]$}{T[\mathrm{SU}(2)]} Higgs branch}

Viewing $V\cong T^\ast \mathbb{C}^2$ as a hyperk\"ahler vector space, the action of $T_G$ on $V$ is tri-Hamiltonian, and the Higgs branch of vacua is precisely the hyperk\"ahler quotient,
\begin{equation}
   \cM_H =  V /\!\!/\!\!/ T_G~.
\end{equation}
Picking the standard complex structure on $T^\ast \mathbb{C}^2$, the hyperk\"ahler moment map splits into a complex part $\mu_{\mathbb{C}}$ and a real part $\mu_{\mathbb{R}}$. Adopting conventional physics notation, we write the basic coordinate functions on $\mathbb{C}^2$ as $X_i$, $i=1,2$, and the coordinate functions on the cotangent directions as $Y^i$. The moment maps are then written in components according to
\begin{equation}
\begin{split}
    \mu_{\mathbb{C}} &= \sum_{a=1}^2 Q_i X_iY^i= X_1Y^1 +X_2Y^2~,\\
    \mu_{\mathbb{R}} &= \sum_{a=1}^2 \frac{Q_a}{2}\left( |X_a|^2-|Y^a|^2\right) = \frac{1}{2}\left( |X_1|^2- |Y^1|^2 + |X_2|^2- |Y^2|^2\right)~.
\end{split}
\end{equation}
The hyperk\"ahler quotient is given by
\begin{equation}\label{eq:hyp-quot}
   \cM_H \cong \mu^{-1}_{\mathbb{C}}(0) \cap \mu_{\mathbb{R}}^{-1}(0)/ T_G~,
\end{equation}
which, as an algebraic symplectic variety, can be realised alternatively as a holomorphic symplectic quotient,
\begin{equation}\label{eq:higgs-quot}
    \cM_H \cong \mu_{\mathbb{C}}^{-1}(0) /\!\!/ T_{G,{\mathbb{C}}}~.
\end{equation}
This is the affine GIT quotient of the zero-locus of the complex moment map by the complexified gauge group $T_{G,\mathbb{C}} =\mathbb{C}^\times$. The Higgs branch chiral ring $R_H$ is then realised by $T_{G,\mathbb{C}}$-invariant holomorphic functions with the complex moment map set to zero,
\begin{equation}\label{eq:TSU2_Higgs_chiral_ring_sympquotient}
    R_H\cong\mathbb{C} \left[ \mu_{\mathbb{C}}^{-1}(0) \right]^{T_{G,\mathbb{C}}}~.
\end{equation}

In the present case, the gauge-invariant algebraic functions on $V$ are generated by the entries in the matrix
\begin{equation}\label{eq:tsu2_higg_branch_matrix}
    M=
    \begin{pmatrix}
        X_1 Y^1 & X_1 Y^2 \\
        X_2 Y^1 & X_2 Y^2 
    \end{pmatrix}~,
\end{equation}
and the vanishing of the complex moment map is equivalent to requiring $\tr\,M=0$. Imposing this tracelessness implies the relation $M^2=0$, and indeed the Higgs branch is the just the closure of the minimal (and in this case, unique) nilpotent co-adjoint orbit of ${\rm SL}(2,\mathbb{C})$, and the traceless matrix \eqref{eq:tsu2_higg_branch_matrix} is the corresponding $\slf_2$ moment map.

For our purposes later, let us introduce alternative notation for the matrix entries of $M$ in terms of a Chevalley basis for $\slf_2$,
\begin{equation}
    M=
    \begin{pmatrix}
        h/2 & e \\
         f & -h/2 
    \end{pmatrix}~.
\end{equation}
The nilpotency condition amounts to the relation $ef = -\frac{1}{4}h^2$, which characterises this nilpotent orbit closure as simplest Kleinian singularity,
\begin{equation}
    \cM_H \cong \overline{\mathbb{O}_{\rm min}(\slf_2)} \cong A_1 \cong \mathbb{C}^2/\mathbb{Z}_2~.
\end{equation}
This Higgs branch admits a natural resolution by introducing a suitable FI parameter $\zeta\in\mathfrak{t}^*_{G}\cong\mathbb{R}$. This is the bottom component of a vector multiplet for the topological symmetry, $T_C$, which we recall in the next subsection. In the presence of such an FI parameter, the Higgs branch is modified according to
\begin{equation}
   \cM_{H,\zeta}\cong\mu^{-1}_{\mathbb{C}}(0)\cap\mu_{\mathbb{R}}^{-1}(\zeta)/T_G~.
\end{equation}
Alternatively, this is the projective GIT quotient,
\begin{equation}
    \cM_{H,\zeta} \cong\mu_{\mathbb{C}}^{-1}(0)/\!\!/_{\zeta}T_{G,{\mathbb{C}}}~.
\end{equation}
This resolution just gives the cotangent bundle of $\mathbb{CP}^1$, 
\begin{equation}
    T^\ast\mathbb{CP}^1\twoheadrightarrow A_1~,
\end{equation}
with the zero section being the exceptional divisor.

\subsubsection{\label{subsubsec:TSU2_coulomb}\texorpdfstring{$T[\mathrm{SU}(2)]$}{T[\mathrm{SU}(2)]} Coulomb branch} 

Though we will primarily focus on the $A$-twist, which foregrounds the Higgs branch from the perspective of the boundary VOA, it will be helpful later to have at hand the description of the Coulomb branch of vacua, $\cM_C$, for this theory as well. As the theory is self mirror, ultimately of course one has $\cM_C\cong\cM_H$. 

The Coulomb branch is parameterised by the expectation value of the complex scalar field $\varphi$ in the vector multiplet as well that of certain monopole operators $v_{A}$ labelled by co-characters of the (mirror) gauge group $\wt{T}_G\cong \mathrm{U}(1)$,
\begin{equation}
    A: \mathrm{U}(1) \rightarrow \wt{T}_G~.
\end{equation}
After fixing conventions, co-characters for $\mathrm{U}(1)$ can be identified with integers---the corresponding monopole operators have the effect of modifying the gauge bundle at a point so that it has $A$ units of flux through a sphere surrounding the point. The Coulomb branch chiral ring will be generated by $\varphi$ and by the ``elementary'' monopole operators with $A=\pm1$. These monopole operators obey the (quantum mechanical) relation
\begin{equation}
    v_{1}v_{-1} = \varphi^2~.
\end{equation}
Up to some changes of normalisation, one can make the identification with generators of the Higgs branch chiral ring,
\begin{equation}
    \begin{pmatrix}
        \,\frac12 h\,\\ e \\ f
    \end{pmatrix}
    ~~\longleftrightarrow~~
    \begin{pmatrix}
        \varphi\\ v_+ \\ v_-
    \end{pmatrix}~,
\end{equation}
and so Higgs and Coulomb branches are seen to coincide as expected. (Here we have only briefly reviewed this identification at the level of algebraic structures; the two varieties further agree as holomorphic symplectic varieties.) Under this mirror map, the flavour torus $T_H$ (with corresponding moment map $h$) is mapped to the (UV) topological symmetry $T_C$ (with moment map $\varphi$), which measures the topological charge of the gauge bundle surrounding an operator.

\subsection{\label{subsec:tsu2_voa_BRST}Boundary VOA as semi-infinite cohomology}

Boundary VOAs in the $A$-twist were identified in~\cite{Costello:2018fnz} as a BRST analogue of the symplectic quotient construction of the Higgs branch, with additional (conventionally free fermion) degrees of freedom arising on the boundary that must be included in order to ensure anomaly cancellation. 

The present example was discussed in \cite{Costello:2018fnz} (and also \cite{Creutzig:2017uxh}), and we review and elaborate upon that discussion here. In this theory, the two hypermultiplets give rise to two pairs of symplectic bosons\footnote{In an abuse of notation, we will use the same letters for symplectic bosons and the corresponding functions on $V$/Higgs branch chiral ring operators.)} ($X_a,Y^a$), $a=1,2$. The boundary degrees of freedom contribute two pairs of complex fermions ($\chi_\alpha,\xi^\alpha$), $\alpha=1,2$. The singular OPEs of these gauge theoretic degrees of freedom are standard,
\begin{equation}
    Y^a(z)X_b(w)\sim \frac{\delta^a_b}{z-w}~,\qquad\chi_\alpha(z)\xi^\beta(w) \sim \frac{\delta_\alpha^\beta}{z-w}~.
\end{equation}
For future convenience, we denote the VOA generated by the $i$'th symplectic boson pair by $\texttt{Sb}_i$, and the one generated by $i$'th complex fermions by $\texttt{Ff}_i$. We denote the total ungauged VOA by $\mathcal{X}$,
\begin{equation}
    \mathcal{X}\coloneqq\texttt{Sb}_1\otimes\texttt{Sb}_2\otimes\texttt{Ff}_1\otimes\texttt{Ff}_2~.
\end{equation}
The $\mathrm{U}(1)$ gauge symmetry in $T[\mathrm{SU}(2)]$ corresponds to a level-zero (commutative) Heisenberg vertex subalgebra of $\mathcal{X}$, generated by the current\footnote{Here and throughout this paper, normal ordering,
\begin{equation}
    ab\cdots yz \equiv (a(b(\cdots( yz)\cdots))))~.
\end{equation}
is implied for composite operators in vertex operators algebras.}
\begin{equation}
    \cJ=X_aY^a+\chi_\alpha\xi^\alpha~.
\end{equation}
The mode algebra for this current is the loop algebra $\wt{\glf}_1\cong\glf_1\otimes\mathbb{C}[t,t^{-1}]$ (without central extension due to the vanishing level). BRST reduction is performed concretely by introducing a $(b,c)$ ghost system (of weight $(1,0)$),
\begin{equation}
    b(z)c(w)\sim\frac{1}{z-w}~,
\end{equation}
and forming the Feigin standard complex $\mathcal{X}^{\bullet} = \mathcal{X} \otimes (b,c)^\bullet$. The cohomological grading is by ghost number (with $c$ assigned ghost number $+1$ and $b$ assigned ghost number $-1$), and the differential is identified with the zero-mode of the BRST current,
\begin{equation}\label{eq:psl22-BRST}
    J^{\text{BRST}}= c\cJ~,\qquad d = J^{\text{BRST}}_0~.
\end{equation}
The VOA for the gauge theory is then identified with the cohomology of this BRST operator acting on the subcomplex of $\mathcal{X}^\bullet$ that is annihilated by the zero modes of the $b$-ghost and the gauge current,
\begin{equation}
    \mathbb{V}_{T[\mathrm{SU}(2)]} = H^\bullet_{d}\left(\psi\in\mathcal{X}^\bullet\,|\,\cJ_0\psi=b_0\psi=0\right)~.
\end{equation}
This is the standard model for the $(\gf,\gf_0)=(\wt{\glf}_1,\glf_1)$ \emph{relative  semi-infinite cohomology} of $\mathcal{X}$ \cite{voronov1994semi} (\cf\ also \cite{Arakawa:2018egx}), $H^{\frac{\infty}{2}+\bullet}(\gf,\gf_0,\mathcal{X})$.

Na\"ive candidates for strong generators of the VOA defined in this manner are the various affine currents realised as quadratic gauge invariants of the elementary bosons and fermions (with no additional derivatives), 
\begin{equation}\label{eq:tsu2_naive_gens}
    X_aY^b~,\qquad Y^a\chi_\alpha~,\qquad X_a\xi^\alpha~,\qquad \chi_\alpha\xi^\beta~.
\end{equation}
In the ungauged theory, these generate an affine current algebra which is (a quotient of) the universal level-one $V^1(\glf(2|2))$.\footnote{We adopt the convention that $V^k(\gf)$ denotes the universal affine current algebra of type $\gf$ at level $k$, and $V_k(\gf)$ denotes the unique simple quotient thereof.} In the gauged theory, the super-trace combination $X_aY^a-\chi_\alpha\xi^\alpha$ is not $d$-closed, so one should restrict to (a quotient of) the level-one subalgebra $V^1(\slf(2|2))$.

By Theorem 5.5 of \cite{Creutzig:2016ehb}, the $\slf(2|2)$ current algebra generated by the super-traceless combinations of \eqref{eq:tsu2_naive_gens} is precisely the Heisenberg coset of $\mathcal X$ by the Heisenberg algebra generated by $\mathcal J$. Furthermore, $\cJ$ itself, which is $d$-exact (it is equal to $d(b)$), is the ordinary trace combination of these currents and this is therefore removed by the cohomological quotient. What is left behind is a level-one affine $\pslf(2|2)$ current algebra, which we expect to be precisely the simple quotient vertex algebra $V_1(\pslf(2|2))$.\footnote{In vertex algebras arising from four-dimensional superconformal field theories as in \cite{Beem:2013sza}, simplicity is a consequence of four-dimensional unitarity. Though we suspect an analogous statement may hold for boundary vertex algebras arising on $\mathcal{N}=(0,4)$ boundaries in three-dimensional unitarity theories, we do not pursue such a general result here. Where it is unproven, it will be a standing conjecture that all vertex algebras constructed in this paper by free field methods are, as in the present example, in fact simple.} 

Concretely, this vertex algebra is generated by:
\begin{itemize}
    \item Affine currents $L_{ab}$, which generate a $V_{-1}(\slf_2)$ subalgebra.
    \item Affine currents $J^{\alpha\beta}$, which generate a $V_{1}(\slf_2)$ subalgebra.
    \item Fermionic currents $\Theta_{Aa}^\alpha$, which transform in the tri-fundamental representation of the (zero modes algebras of the) previous two $\slf(2)$ symmetries as well as an $\rm{SL}(2)_o$ outer automorphism symmetry (the index $A$ is the outer automorphism fundamental index).
\end{itemize}
These various currents correspond to the elements $\ell_{ab}$, $\scriptr^{\alpha\beta}$, and $(\vartheta_{(A)})_a^\alpha$ of the $\pslf(2|2)$ superalgebra, the structure of which (along with the precise OPEs for this current algebra) we provide in Appendix~\ref{app:psl22-conv}.

\medskip

Before moving on to develop a free field realisation, we would like to determine the associated variety for this vertex algebra. It follows from completely general considerations (due to the nilpotency of Grassmann-odd $\Theta$ generators at the level of Zhu's $C_2$ algebra) that we will have
\begin{equation}
    X_{V_1(\pslf(2|2))}\subseteq X_{V_1(\slf_2)}\times X_{V_{-1}(\slf_2)}\cong \{pt.\}\times \slf_2^\ast~.
\end{equation}
To proceed further requires some information about the null states of this VOA that do not lie in the $V_1(\slf_2)$ subalgebra. To find the state we need, it is sufficient to observe the existence of a conformal embedding \cite{Creutzig:2017uxh,Adamovic:2019age},
\begin{equation}\label{eq:conf_embedding}
    V_1(\slf_2)\otimes V_{-1}(\slf_2)\hookrightarrow V_1(\pslf(2|2))~.
\end{equation}
This is equivalent to the identity (in the simple quotient) between various Sugawara stress tensors,
\begin{equation}\label{eq:conf_embedding_null}
    T_{\rm sug}^{\,\pslf(2|2)_1} = T_{\rm sug}^{\,\slf(2)_1} + T_{\rm sug}^{\,\slf(2)_{-1}}~.
\end{equation}
Let us denote the \emph{un-normalised} Segal--Sugawara vectors for the current algebra of type $\hat{\gf}_k$ by
\begin{equation}
    S_{\gf_k} = \kappa_{AB}(J^AJ^B)~,
\end{equation}
where $\kappa_{AB}$ is the Killing form on $\gf$. The Sugawara stress tensor is given by 
\begin{equation}
    T_{\gf_k} = \frac{k\dim\gf}{k+h^\vee_\gf}S_{\gf}~.
\end{equation}
Then the null relation \eqref{eq:conf_embedding_null} becomes (substituting in the expression for the $\pslf(2|2)$ Segal--Sugawara vector) the following relation,
\begin{equation}
    S_{\pslf(2|2)} = -2\left(S_{\slf(2)_1} + \Theta_{Aa}^\alpha\Theta^{Aa}_\alpha - S_{\slf(2)_{-1}}\right) = S_{\slf(2)_1} -3 S_{\slf(2)_{-1}}~,
\end{equation}
which we rewrite as
\begin{equation}
    S_{\slf(2)_{-1}} = \frac{3}{5}S_{\slf(2)_1} + \frac{2}{5}\Theta_{Aa}^\alpha\Theta^{Aa}_\alpha~. 
\end{equation}
The right hand side is nilpotent in the Zhu $C_2$ algebra, so we see that the Segal--Sugawara vector of the $\slf(2)_{-1}$ subalgebra also vanishes at the level of the associated variety. In particular, this implies the inclusion
\begin{equation}
    X_{V_1(\pslf(2|2))}\subseteq \overline{\mathbb{O}_{\rm min}(\slf_2)}\subset \slf_2^\ast~,
\end{equation}
As the associated variety must admit an $\mathrm{SL}(2)$ action descending from the zero modes of the level $-1$ affine currents, it must be either the entirety of the nilpotent orbit closure above or a point.

One can see easily that the associated variety cannot be zero-dimensional. The state $(L_{11})^n$ is non-null for all $n$; moreover is the unique state of conformal dimension $n$ and charge $n$ with respect to $L_{12}$ so cannot participate in a null relation with any other states in the VOA. Thus we have the identification
\begin{equation}
    X_{V_1(\pslf(2|2))}\cong \overline{\mathbb{O}_{\rm min}(\slf_2)} \cong \cM_H~.
\end{equation}
In the remainder of this section we develop by two complementary approaches to the free field realisation of this vertex algebra.

\medskip

\paragraph{Remark.} This vertex algebra is the same as the one appearing in the description of string theory on $AdS_3\times S^3\times \mathbb{T}^4$ with minimal NS flux in the Berkowitz--Vafa--Witten hybrid formalism \cite{Berkovits:1999im}, which has recently been investigated in great detail in connection with the ${\rm AdS}_3/{\rm CFT}_2$ correspondence in \cite{Eberhardt:2018ouy}. In this context, in~\cite{Gaberdiel:2022als} the same vertex algebra was investigated from a similar perspective (in terms of symplectic bosons and free fermions), but without solving the semi-infinite cohomology problem directly and with a slightly different approach to the BRST realisation of the problem. It would be interesting to investigate whether the complete free field realisation presented here has some utility in this string theoretic setting.

\subsection{\label{subsec:free-field-ansatz}Free field realisation of \texorpdfstring{$V_1(\pslf(2|2))$}{V[1](psl)(2,2}}

In our initial treatment, we will develop a free field realisation for following the \emph{ad hoc} methodology employed in analogous situations connected to four-dimensional SCFTs in \cite{Beem:2019tfp,Beem:2021jnm,Beem:2022vfz}. 

The general strategy is that the $V_{-1}(\slf_2)$ subalgebra whose generators parameterise the associated variety/Higgs branch should be localised and realised in terms of a rank-one half-lattice vertex algebra $\cD^{(ch)}(\mathbb{C}^\times)$ along with an auxiliary Virasoro vector that fixes the level to the correct value. Extra degrees of freedom are then introduced that, roughly speaking, correspond to the degrees of freedom remaining after Higgsing the theory. These should be identified with the vertex algebra that remains after performing quantum Drinfel'd--Sokolov reduction. Though we will not carry out such an analysis here, on physical grounds our expectation is that they will again be given by two complex free fermions. These are used in order to construct the remaining strong generators of the original VOA, with the $\slf(2)_{-1}$ highest weight states taking a simple form dictated by conformal weights and charges, and the remaining generators recovered by acting with the $\slf(2)_{-1}$ lowering operator.

\subsubsection{Free field realisation of the bosonic \texorpdfstring{$V_{-1}(\slf_2)$}{V[-1](sl2)}}

Following our outline above, we first consider the ``bosonic'' $V_{-1}(\slf_2)$ vertex subalgebra generated by BRST images of the quadratic invariants built from hypermultiplet scalars. We introduce Chevalley generators $e=L_1^{\phantom{a}2}$, $f=L_2^{\phantom{1}2}$, and $h=L_1^{\phantom{1}1}-L_2^{\phantom{2}2}$, in terms of which the level $k=-1$ current algebra takes the form
\begin{equation}\label{eq:affine_sl2_opes}
\begin{split}
    h(z)h(w)&\sim \frac{-2}{(z-w)^2}~, 		\\
    h(z)e(w)&\sim \frac{\ph{2}e(w)}{z-w}~,  \\
    h(z)f(w)&\sim \frac{-2f(w)}{z-w}~, 		\\
    e(z)f(w)&\sim \frac{-1}{(z-w)^2}+ \frac{h(w)}{z-w}~.\\
\end{split}
\end{equation}
We realise this vertex algebra in the manner of \cite{Adamovic:2004zi} (see \cite{Beem:2019tfp} for the first appearance in the context of VOAs for higher dimensional supersymmetric theories). The realisation is in terms of a lattice extension of a rank-two Heisenberg vertex algebra and an auxiliary Virasoro stress tensor; we recall the details presently.

Introducing chiral bosons $\phi$ and $\delta$ with normalisation $(\phi,\phi) = -(\delta,\delta)=-2$ we set
\begin{equation}
    e \coloneqq e^{\phi + \delta}~,\qquad h \coloneqq-J_\phi\equiv -\partial\phi~.
\end{equation}
The $h \times e$ and $h\times h$ OPEs in \eqref{eq:affine_sl2_opes} are immediately realised. The lowering operator $f$ (with its attendant OPEs) is then realised according to
\begin{equation}
    f \coloneqq \left( -\frac{1}{4}(J_\delta)^2 +  T \right) e^{-\phi-\delta}~.
\end{equation}
where $T$ is the generator of an auxiliary Virasoro algebra with central charge $c_T=1$, so obeying
\begin{equation}
    T(z)T(w)\sim \frac{1/2}{(z-w)^4}+\frac{2T(w)}{(z-w)^2}+\frac{\partial T(w)}{z-w}~.
\end{equation}
The lattice extension (by vertex operators of the form $e^{n(\phi+\delta)}$) of the Heisenberg algebra generated by $J_\phi$ and $J_\delta$ is an example of a half-lattice vertex algebra in the sense of \cite{berman2002representations}, and can also be identified with the algebra of global sections of the sheaf of chiral differential operators (CDOs) \cite{MelnikovCDO} on the rank-one algebraic torus, $\mathcal{D}^{\rm ch}(\mathbb{C}^\times)$. This is the vertex algebra obtained from the symplectic boson $\beta\gamma$ VOA by formally inverting (say) $\gamma$. It was suggested in \cite{Beem:2019tfp} and subsequent work that this sort of a realisation should be thought of as the starting point for writing a localisation of the vertex algebra at hand on the (Zariski) open subset of its associated variety where $e\neq0$.

\subsubsection{Free field realisation of the fermionic \texorpdfstring{$V_{1}(\slf_2)$}{V[1](sl2)}}

We next realise the ``fermionic'' $V_1(\slf_2)$ in terms of additional free fields. This is more straightforward, in the sense that we can exploit a completely textbook realisation of this rational vertex algebra in terms of free fermions; indeed such a realisation is already an ingredient in the \emph{ungauged} vertex algebra $\mathcal{X}$. The motivation here is that we expect that after Higgsing the bulk three-dimensional gauge theory, then the microscopic deformable boundary conditions will continue to carry (the same number of) free fermion excitations, and it should be these residual degrees of freedom that accompany the CDOs of the previous subsection in reconstructing the UV vertex algebra.

To this end, we introduce two pairs of complex free fermions, $\psi_A^\alpha$ where $A = 1,2$, $\alpha = 1,2$, enjoying the singular OPE,\footnote{We use the conventions $\epsilon^{12}=\epsilon_{21}=1$.}
\begin{equation}\label{eq:ferm-OPE}
    \psi_A^\alpha(z) \psi_B^\beta(w) \sim \frac{\epsilon_{AB}\epsilon^{\alpha\beta}}{z-w}~.
\end{equation}
These should thought of as being distinct from the elementary gauge theory fermions $(\chi_\alpha,\xi^\alpha)$, though in due course we will understand their precise relationship. This free fermion VOA includes a $V_1({\mathfrak{so}}_4) \cong V_1({\slf}_2)\otimes V_1({\slf}_2)$ current subalgebra generated by all quadratic combinations of the elementary fermions, with $\alpha$ playing the role of a fundamental index for the first $\slf(2)$, and $A$ doing the same for the second.

For our purposes, we will use (say) the first factor as the required $V_1(\slf(2))$ subalgebra of our VOA, with generators
\begin{equation}
    J^{\alpha\beta} = -\frac{1}{2}\epsilon^{AB} \psi_A^\alpha\psi_B^\beta~.
\end{equation}
The zero modes of the second factor generate an $\mathrm{SL}(2)$ action which we will shortly identify with the outer automorphism group $\mathrm{SL}(2)_o$ of $V_1(\pslf(2|2))$ (see Appendix~\ref{app:psl22-conv}).

Let us define the Cartan generator for that other $V_1(\slf_2)$ current algebra,
\begin{equation}\label{eq:ferm-currents}
    J_\psi=-\frac{1}{2}\epsilon_{\alpha\beta}\psi_1^\alpha \psi_2^\beta~.
\end{equation}
This can be utilised to construct a Virasoro vector which has central charge one and commutes with the $J^{\alpha\beta}$ currents,
\begin{equation}
    T_J=(JJ)~.
\end{equation}
(Note that equivalently this is the Sugawara stress tensor for the second $V_1(\slf_2)$ current algebra, with the equivalence following from the level-two null vectors of that current algebra. Importantly, this means that this Virasoro is non-obviously a singlet of the global $\slf_2$ algebra generated by the zero modes of these currents.)

We can utilise this Virasoro vector as the abstract Virasoro vector appearing in the half-lattice realisation of the bosonic $\slf_2$ above. Altogether at this point we have a realisation
\begin{equation}
    V_{-1}(\slf_2)\otimes V_{1}(\slf_2)\hookrightarrow \mathcal{D}^{ch}(\mathbb{C}^\times)\otimes (\texttt{Ff})^{\otimes 2}\eqqcolon V_{\rm FFR}~.
\end{equation}

\subsubsection{\label{subsubsec:oddcurrents}Free field realisation of odd currents}

Finally we come to the Grassmann-odd generators $\Theta^{\alpha}_{Aa}$. These should have singular OPEs,
\begin{equation}
\begin{split}
\label{eq:oddalg}
    \Theta^{\alpha}_{A1} \Theta^{\beta}_{B1} &~\sim~ \frac{\epsilon_{AB}\epsilon^{\alpha\beta}e(w)}{z-w}~,\\ 
    \Theta^{\alpha}_{A1} \Theta^{\beta}_{B2} &~\sim~ \frac{\epsilon_{AB}\epsilon^{\alpha\beta}}{(z-w)^2}-\frac{\epsilon_{AB}J^{\alpha\beta}(w)+\epsilon_{AB}\epsilon^{\alpha\beta}h(w)/2}{z-w}~,\\ 
    \Theta^{\alpha}_{A2} \Theta^{\beta}_{B2} &~\sim~ -\frac{\epsilon_{AB}\epsilon^{\alpha\beta}f(w)}{z-w}~,\\
\end{split}
\end{equation}
where we have rewritten the relevant OPEs from Appendix~\ref{app:psl22-conv} non-covariantly in terms of the Chevalley basis for the bosonic subalgebra $V_{-1}(\slf_2)$.

Accounting for conformal weights and charges under the Cartan generators of the even current algebra, the highest-weight states with respect to bosonic $\slf(2)$ (the $a=1$ operators) can only take the form (up to scaling),
\begin{equation}\label{eq:tsu2_ffr_highest_weight_ansatz}
    \Theta^{\alpha}_{A1} = \psi_A^\alpha e^{\frac{1}{2}(\phi+\delta)}~,
\end{equation}
and the first OPE in \eqref{eq:oddalg} is reproduced with the given normalisation. Taking the OPE with $f$ and picking out the first order pole we can read off the required expression for the $a=2$ operators,
\begin{equation}\label{eq:tsu2_ffr_lowest_weight_computed}
    \Theta_{A2}^{\alpha} =  \frac{1}{2} \left(\partial \psi_{A}^{\alpha}-\psi_{A}^{\alpha}J_\delta + \frac{1}{3} \epsilon_{\beta\gamma}\epsilon^{BC}\psi_B^\alpha \psi_A^\beta \psi_C^\gamma \right)e^{-\frac{1}{2}(\phi+\delta)}~.
\end{equation}
Direct computation confirms that these operators satisfy the full complement of OPEs amongst themselves and with the even currents, so we have a free field realisation of some quotient of $V^1(\pslf(2|2))$, with the outer automorphism $\mathrm{SL}(2)_o$ generated precisely by the zero modes of the second $\slf(2)$ current algebra made from the free fermions. Whether or not these generate the simple quotient $V_1(\pslf(2|2)$ is not entirely evident at this point. In the remainder of this section we will re-derive the free field realisation so that the simplicity is guaranteed by earlier considerations.

\medskip

\paragraph{Remark.} It is interesting to note the manner in which the reduction to $\pslf(2|2)$ from $\slf(2|2)$ is accomplished in this free field realisation. We chose to identify the abstract Virasoro vector appearing in the realisation of the bosonic $V_{-1}(\slf_2)$ with the fermionic stress tensor $T_{\tilde J}$. Had we not done so, the same Ansatz \eqref{eq:tsu2_ffr_highest_weight_ansatz} would lead to a current algebra where the expression in \eqref{eq:tsu2_ffr_lowest_weight_computed} would have been different---in particular would be missing the third term in the parentheses. In the $\Theta\times\Theta$ OPEs, one would then find the current $J_\psi$ arising, playing the role of the ordinary trace generator in $\slf(2|2)$, and that must be removed to pass to $\pslf(2|2)$. So it is precisely the identification $T=T_J$ that leads to the reduced current algebra.

\subsection{\label{subsec:TSU2-BRST}Semi-infinite cohomology via bosonisation}

In the previous section we developed a free field realisation of $V_1(\pslf(2|2))$ by brute force, similar to what has been done previously for many four-dimensional VOAs. The problem at hand enjoys an advantage over analogous constructions in four dimensions owing to the Abelian nature of the semi-infinite cohomology~\eqref{eq:psl22-BRST}. Indeed, we will see below that by means of bosonisation we can derive the free field realisation of the last section directly from first principles. 

\subsubsection*{Outline of strategy}

Before getting started, we outline here the general strategy we will employ. The elementary symplectic bosons and free fermions (henceforth \emph{gauge theory variables}) can be bosonised so as to realise the ungauged vertex algebra as a subalgebra of a certain lattice vertex (super)algebra $V_L$,
\begin{equation}
    \mathcal{X}\hookrightarrow V_{L}^! \hookrightarrow V_L~.
\end{equation}
Here $V_{L}^!$ is a subalgebra of $V_L$ and is defined by restricting certain lattice momenta to lie in a degenerate sublattice---this is a generalisation of the half-lattice subalgebras of \cite{berman2002representations}. By performing a change of basis for the lattice and a minor refinement, the BRST differential can be restricted to act only within a rank-two sublattice vertex algebra $V_\eta$,
\begin{equation}
    V^!_L \hookrightarrow V^{!}_{\tilde L} \cong V_\eta\otimes \wt{V}_{\rm FFR}~.
\end{equation}

A relatively straightforward application of a vanishing theorem of Voronov for semi-infinite cohomology \cite{voronov1994semi} then shows that the (relative) semi-infinite cohomology of $V_\eta$ is one-dimensional (the trivial vertex algebra), so we have
\begin{equation}
    H^{\frac{\infty}{2}+\bullet}(\gf,\gf_0,V_{\tilde{L}}^!) \cong H^{\frac{\infty}{2}+\bullet}(\gf,\gf_0,V_\eta)\otimes \wt{V}_{\rm FFR} \cong V_{\rm FFR}~.
\end{equation}
where $V_{\rm FFR}$, the restriction of $V_L^!$ to no $\widetilde{\eta}$ momentum, is precisely the free field vertex algebra introduced above in our \emph{ad hoc} free field realisation.

The image of $\mathcal{X}$ in $V_{L}^!$ can be characterised as the intersection of kernels of two screening charges which, after changing basis, act trivially on $V_\eta$, so can be represented as screening charges acting on $V_{\rm FFR}$. We thus arrive at the result that the solution to the gauge theory BRST cohomology problem can be identified with the intersection of the kernel of two screening charges on $V_{\rm FFR}$. Direct computation then confirms that the generators of $V_1(\mathfrak{psl}(2|2))$ proposed above do indeed lie in this intersection of kernels.

\subsubsection{\label{subsubsec:bosonise_matter}Bosonising elementary fields and relating semi-infinite cohomologies}

The elementary symplectic bosons and free fermions appearing in the gauge theory description of this vertex algebra can all be bosonised, and expressed in terms of (subalgebras) of appropriate lattice vertex algebras. The details of the construction are reviewed in Appendix~\ref{app:lattice_bosons}. In the present example, this amounts to introducing Heisenberg vertex algebras $V_{\rho_i,\sigma_i}$ generated by currents $J_{\rho_i}$ and $J_{\sigma_i}$ and the Heisenberg vertex algebra $V_{\gamma_i}$ generated by the current $J_{\gamma_i}$, with
\begin{equation}
    J_{\rho_i}(z)J_{\rho_i}(w)\sim \frac{\ph{1}}{(z-w)^2}~, \qquad
    J_{\sigma_i}(z)J_{\sigma_i}(w)\sim \frac{-1}{(z-w)^2}~,\qquad
    J_{\gamma_i}(z)J_{\gamma_i}(w)\sim \frac{\ph{1}}{(z-w)^2}~.
\end{equation}
and taking lattice extensions
\begin{equation}
    V^!_{L_{\texttt{Sb}},i} = \bigoplus_{n\in\mathbb{Z}}V_{\rho_i,\sigma_i}\cdot e^{n(\rho_i-\sigma_i)}~,\qquad 
    V^!_{L_{\texttt{Ff}},i} = \bigoplus_{n\in\mathbb{Z}}V_{\gamma_i}\cdot e^{n\gamma_i}~.
\end{equation}
We then have the embedding
\begin{equation}\label{eq:TSU2-bos-embed}
    \mathcal{X}\hookrightarrow V_L^!\coloneqq V^!_{L_{\texttt{Sb}},1}\otimes V^!_{L_{\texttt{Sb}},2}\otimes V_{L_{\texttt{Ff}},1}\otimes V_{L_{\texttt{Ff}},2}~,
\end{equation}
In fact, there are several different such embeddings, corresponding to the choices of the $\mathfrak{i}^+$ or $\mathfrak{i}^-$ for each of the two symplectic bosons, as well as a corresponding choice of sign convention for each free fermion (see Appendix~\ref{app:lattice_bosons} for these embedding conventions). 

In the case of a single symplectic boson, the two different embeddings are equivalent up to a redefinition of $(X,Y)\to (Y,-X)$. This is therefore a matter of convention. In the gauge theory context, the different choices of embeddings for the various symplectic bosons can be genuinely inequivalent due to the assignment of gauge charges, which are not invariant under the previous redefinition. (On the other hand, the choice of realisation for the free fermions is always a matter of convention.) 

It will prove convenient for us to correlate the $\pm$ type of the realisation of the $i$'th free fermion with that of the $i$'th symplectic boson. This implies that contribution of the $i$'th symplectic boson/free fermion pair to the gauged $\gf=\mathfrak{gl}_1$ current will be given by
\begin{equation}
    (X_iY^i)+(\chi_{i}\xi^i) = \epsilon_i(J_{\sigma_i}+J_{\gamma_i})~,
\end{equation}
where $\epsilon_i$ is the sign of the $i$'th symplectic boson embedding. (Here there is no Einstein summation convention in force.) 

In the context of the VOA for $T[\mathrm{SU}(2)]$, there are na\"ively four choices of embeddings, which we denote by a pair of signs (corresponding to the choice of embedding type for the first and second symplectic boson, respectively). However, $(++)$ and $(--)$ are essentially equivalent, as are $(+-)$ and $(-+)$, as they are pairwise related upon redefining the gauge group so that the $Y^i$ have charge $+1$.

The next question is how the above bosonisation of $\mathcal{X}$ in terms free lattice fields interacts with taking relative semi-infinite cohomology. As usual, we have a short exact sequence of $\gf$-modules,
\begin{equation}\label{eq:SES}
    0 \rightarrow \mathcal{X} \xhookrightarrow{\mathfrak{i}} V_L^! \twoheadrightarrow \mathcal{Y} \rightarrow 0 ~,
\end{equation}
where $\mathcal{Y}$ is defined to complete the sequence. This induces a long-exact sequence in cohomology
\begin{eqnarray*}
    \cdots \rightarrow
    H^{\frac{\infty}{2}+i-1}\left(\gf,\gf_0,\mathcal{Y} \right)  &\xrightarrow{\delta}&
    H^{\frac{\infty}{2}+i}\left(\gf,\gf_0,\mathcal{X} \right) \rightarrow H^{\frac{\infty}{2}+i}\left(\gf,\gf_0,V_L^! \right)
    \rightarrow H^{\frac{\infty}{2}+i}\left(\gf,\gf_0,\mathcal{Y} \right) \rightarrow \cdots ~.
\end{eqnarray*}
We will see presently that the connecting homomorphism $\delta$ vanishes (it is the zero map), so we in fact have short exact sequences
\begin{equation}\label{eq:SES_coh}
    0 \rightarrow H^{\frac{\infty}{2}+i}\left(\gf,\gf_0 ,\mathcal{X} \right)\xhookrightarrow{\mathfrak{i}} H^{\frac{\infty}{2}+i}\left(\gf,\gf_0 ,V_L^!\right) \twoheadrightarrow H^{\frac{\infty}{2}+i}\left(\gf,\gf_0 ,\mathcal{Y}\right) \rightarrow 0~,
\end{equation}
giving rise the inclusion
\begin{equation}
     H^{\frac{\infty}{2}+i}(\gf,\gf_0,\mathcal{X})
    \hookrightarrow
    H^{\frac{\infty}{2}+i}(\gf,\gf_0,V_L^!)~.
\end{equation}

To see that the connecting homomorphism vanishes, it is useful to express $V^!_L$ in terms of Fock modules in a way that separates out the action of the symplectic boson screening currents from the $\gf$ action,
\begin{equation}
    V_L^! \cong\bigoplus_{a,b\in\mathbb{Z}} \left(\mathcal{F}_{a\rho_1}\otimes\mathcal{F}_{b\rho_2}\right)\otimes\left(\mathcal{F}_{-a\sigma_1-b\sigma_2}\otimes V_{L_{\texttt{Ff}},1}\otimes V_{L_{\texttt{Ff}},2}\right)~.
\end{equation}
For each term in the direct sum, the screening currents $e^{\rho_1}$ and $e^{\rho_2}$ act exclusively in the first parenthetical grouping of Fock modules, while the $\gf$ action generated by a linear combination of the $\sigma_i$ and $\gamma_i$ plays out exclusively on the second parenthetical grouping. Thus, as one can always find a section of a surjective map of vector spaces, there is a direct sum decomposition \emph{as $\gf$ modules},
\begin{equation}
    V^!_L\cong \mathcal{X}\oplus\mathcal{Y}~,
\end{equation}
\emph{i.e.}, the short exact sequence \eqref{eq:SES} splits in $\cU(\gf)$-Mod. This, in particular, implies the vanishing of the connecting homomorphism $\delta$.

The gauge theory cohomology can therefore be constructed within the cohomology of the bosonised free field VOA, and what's more, the gauge theory cohomology will be characterised in the free field cohomology by the action of the same screening charges, now treated at the level of cohomology classes. We will see more explicitly how this works after we determine the semi-infinite cohomology of $V_L^!$ below.

\subsubsection{\label{subsubsec:change_of_basis}A new basis for the free field space}

For the present discussion, we fix the $(+-)$ bosonisation choice; this turns out to be advantageous for various reasons. We return to the inequivalent $(++)$ case in Section~\ref{sec:TSU2-sec-bos}.

For the $(+-)$ embedding, the gauged Heisenberg current is realised in $V_L^!$  by
\begin{equation}
    \cJ = J_{\sigma_1} + J_{\gamma_1} - J_{\sigma_2} - J_{\gamma_2}~.
\end{equation}
To characterise the $(\gf,\gf_0)$ relative semi-infinite cohomology of $V_L^!$ with respect to the action of $\cJ$, we will reorganise our free field space in order to isolate the $\cJ$ action. To this end, let us introduce a new basis for our chiral bosons/Heisenberg algebra which sequesters the generator $\cJ$ away from the remaining bosons. In particular, in terms of the original chiral bosons, $\{\sigma_1 , \sigma_2 , \rho_1 , \rho_2 , \gamma_1 , \gamma_2 \}$, we define\footnote{We remark here that it is not a notational infelicity that we reuse here the letters $\phi$ and $\delta$ here. These bosons will turn out to be essentially the same as the ones introduced in our previous \emph{ad hoc} free field realisation.}
\begin{equation}\label{eq:new-basis-psl22}
    \left(
    \begin{array}{c}
    \eta        \\
    \wt{\eta}   \\
    \phi        \\
    \delta      \\
    \omega_1    \\
    \omega_2    \\
    \end{array}
    \right)=
    \left(
    \begin{array}{cccccc}
     \ph{1} & -1 & \ph{0} & \ph{0} & \ph{1} & -1 \\
     \ph{1} & -1 & -1 & \ph{1} & \ph{0} & \ph{0} \\
     -1 & -1 & \ph{0} & \ph{0} & \ph{0} & \ph{0} \\
     \ph{0} & \ph{0} & \ph{1} & \ph{1} & \ph{0} & \ph{0} \\
     \ph{\frac{1}{2}} & -\frac{1}{2} & -\frac{1}{2} & \ph{\frac{1}{2}} & \ph{1} & \ph{0} \\
     -\frac{1}{2} & \ph{\frac{1}{2}} & \ph{\frac{1}{2}} & -\frac{1}{2} & \ph{0} & \ph{1} \\
    \end{array}
    \right)
    \left(
    \begin{array}{c}
    \sigma_1    \\
    \sigma_2    \\
    \rho_1      \\
    \rho_2      \\
    \gamma_1    \\
    \gamma_2    \\
    \end{array}
    \right)~.
\end{equation}
In this new basis, the bilinear form $B$ encoding the norms of the chiral bosons is transformed to
\begin{equation}\label{eq:modified_bilinear_form}
     P_{(+-)} B P_{(+-)}^T = 
    \left(
    \begin{array}{cccccc}
    \ph{0} & -2 & \ph{0}  & \ph{0} & \ph{0} & \ph{0} \\
   -2 & \ph{0}  & \ph{0}  & \ph{0} & \ph{0} & \ph{0} \\
    \ph{0} & \ph{0}  & -2 & \ph{0} & \ph{0} & \ph{0} \\
    \ph{0} & \ph{0}  & \ph{0}  & \ph{2} & \ph{0} & \ph{0} \\
    \ph{0} & \ph{0}  & \ph{0}  & \ph{0} & \ph{1} & \ph{0} \\
    \ph{0} & \ph{0}  & \ph{0}  & \ph{0} & \ph{0} & \ph{1} \\
    \end{array}
    \right)~,
\end{equation}
where $P_{(+-)}$ is the change-of-basis matrix in \eqref{eq:new-basis-psl22}. Inspecting the inverse change-of-basis matrix,
\begin{equation}\label{eq:COB-inverse}
    P_{(+-)}^{-1}=
    \left(
    \begin{array}{cccccc}
     \ph{\frac{1}{2}} & \ph{\frac{1}{2}} & -\frac{1}{2} & \ph{0} & -\frac{1}{2} & \ph{\frac{1}{2}}\\
     -\frac{1}{2} & -\frac{1}{2} & -\frac{1}{2} & \ph{0} & \ph{\frac{1}{2}} & -\frac{1}{2}\\
     \ph{\frac{1}{2}} & \ph{0} & \ph{0} & \ph{\frac{1}{2}} & -\frac{1}{2} & \ph{\frac{1}{2}}\\
     -\frac{1}{2} & \ph{0} & \ph{0} & \ph{\frac{1}{2}} & \ph{\frac{1}{2}} & -\frac{1}{2}\\
     \ph{0} & -\frac{1}{2} & \ph{0} & \ph{0} & \ph{1} & \ph{0}\\
     \ph{0} & \ph{\frac{1}{2}} & \ph{0} & \ph{0} & \ph{0} & \ph{1}\\
    \end{array}
    \right)~,
\end{equation}
we see that our original (sub)lattice of charges $L^!$ takes the form
\begin{equation}
    L^! = \mathbb{Z}\left(-\tfrac{\wt{\eta}}{2}+\tfrac{\phi+\delta}{2}\right)\oplus
    \mathbb{Z}\left(\tfrac{\wt{\eta}}{2}+\tfrac{\phi+\delta}{2}\right)\oplus
    \mathbb{Z}\left(\tfrac{\wt{\eta}}{2}-\omega_1\right)\oplus
    \mathbb{Z}\left(\tfrac{\wt{\eta}}{2}+\omega_2\right)~.
\end{equation}
It will be useful to introduce the following (integral) refinement of this lattice,
\begin{equation}\label{eq:TSU2-VFRR-lattice}
    L^!\subset\tilde{L}^! = \mathbb{Z}\left(\tfrac{\wt{\eta}}{2}\right)\oplus\mathbb{Z}\left(\tfrac{\phi+\delta}{2}\right)\oplus
    \mathbb{Z}\omega_1\oplus
    \mathbb{Z}\omega_2~,
\end{equation}
in terms of which $L^!$ appears as the sublattice of points whose integer coordinates sum to an even number. 

The $\omega_{1,2}$ directions can be re-interpreted as bosonisations of free fermions $\psi_A^\alpha$,
\begin{alignat}{4}\label{eq:tsu2_fer_bosonisation}
    &\psi_1^1 &&\coloneqq -e^{\omega_2}~,\qquad &&\psi_2^2 &&\coloneqq e^{-\omega_2}~,\\
    &\psi_1^2 &&\coloneqq  e^{-\omega_1}~,\qquad &&\psi_2^1 &&\coloneqq e^{\omega_1}~.
\end{alignat}
The OPEs for these are the same as those from our previous \emph{ad hoc} construction~\eqref{eq:ferm-OPE},
\begin{equation}
    \psi_A^\alpha (z) \psi_B^\beta \sim \frac{\epsilon_{AB}\epsilon^{\alpha \beta}}{z-w}~.
\end{equation}
Thus after the final refinement, we see that we can identify the bosonised VOA according to
\begin{equation}\label{eq:ff_factorisation}
    V_{\tilde L}^! \cong V_\eta \otimes \wt{V}_{\rm FFR}~,
\end{equation}
where $V_\eta$ is the lattice extension of the rank-two Heisenberg algebra generated by $J_{\eta}$ and $J_{\wt{\eta}}$ along the lattice $\mathbb{Z}(\frac{\wt{\eta}}{2})$. As announced in the overview for this subsection, $V_{\rm FFR}$ is precisely the free field vertex algebra introduced above in our \emph{ad hoc} free field realisation.

\medskip

\paragraph{Remark.} Note that $V_\eta$ itself is isomorphic to the algebra of global sections of CDOs on $\mathbb{C}^\times$, \ie, $\cD^{(ch)}(\mathbb{C}^\times)$. In general, for CDOs on an algebraic group $\cG$ there are two commuting ${\rm Lie}(\cG)$ affine current algebras whose levels add up to $-2h_{\cG}^\vee$. In the current identification, we can think of $J_{\eta}$---the level-zero $\wt{\glf}_1$ current---as the generator of the diagonal subalgebra, where in the Abelian case one sets $h^\vee=0$.

\medskip

The details of the choice of basis in \eqref{eq:new-basis-psl22} have been specified in order to simplify the analysis of our BRST problem. In particular, this basis manifests the following important properties.
\begin{itemize}
    \item The basis vector $\eta$ is precisely the (null) vector corresponding to the gauged $\wt{\glf}_1$ current. 
    \item One additional basis $\wt{\eta}$ vector is non-trivially paired with the first.
    \item There are just two basis vectors which pair non-trivially with the original fermionic directions $\gamma_i$. Each is a null ``bosonic'' correction to $\gamma_i$, with the correction chosen to enforce orthogonality to the $\eta_i$.
    \item The remaining basis directions $\phi$ and $\delta$, which are purely bosonic, are mutually orthogonal to all of the above and are of indefinite signature $(1,1)$.
\end{itemize}
Though the choices made here may not be the unique ones (for the fixed bosonisation) that allow for the following analysis, any such choice is tightly constrained.

Looking forward, we observe that the change-of-basis matrix \eqref{eq:new-basis-psl22} can be expressed entirely in terms of the charge data of our theory (and the choice of bosonisation) in a manner that guarantees all of the aforementioned good properties will hold. Denote by $Q_\epsilon$---for a sign vector $\epsilon$ determining the chosen bosonisation and a charge matrix $Q$---the charge matrix obtained from $Q$ by flipping the signs of its columns according to $\epsilon$. This defines an equivalent theory with different charge conventions for the hypermultiplets. Denote by $\wt{Q}_\epsilon$ a mirror charge matrix of the theory defined by $Q_\epsilon$ (and $q_\epsilon$) according to the prescription of Section~\ref{sec:tsu2}. In the present case, this gives (\cf\ \eqref{eq:mirror-sqed2-def}),
\begin{equation}\label{eq:plus_minus_charge_matrices}
\begin{split}
    Q_{(+-)} &= (1,-1)~,\\
    \wt{Q}_{(+-)} &= (-1,-1)~.
\end{split}
\end{equation}
Then, we can write the complete change of basis matrix as follows,
\begin{equation}\label{eq:TSU2-COB}
    P_{(+-)} = \left(
        \begin{array}{c:c:c}
        ~~Q_{(+-)}~~                & ~~0~~                     & ~~Q_{(+-)}~~ \\ \hdashline
        &&\\[\dimexpr-\normalbaselineskip+2pt]
        ~~Q_{(+-)}~~                & ~~-Q_{(+-)}~~             & ~~0~~ \\ \hdashline
        &&\\[\dimexpr-\normalbaselineskip+2pt]
        ~~\wt{Q}_{(+-)}~~    & ~~0~~                     & ~~0~~ \\ \hdashline
        &&\\[\dimexpr-\normalbaselineskip+2pt]
        ~~0~~                       & ~~-\wt{Q}_{(+-)}~~ & ~~0~~ \\ \hdashline
        &&\\[\dimexpr-\normalbaselineskip+2pt]
        ~~F_{(+-)}^T~~              & ~~-F_{(+-)}^T~~           & ~~\mathbf{1}_{2\times 2}~~
        \end{array} 
    \right)~,
\end{equation}
where
\begin{equation}\label{eq:Fdef}
    F_{(+-)} = 
    \begin{pmatrix}
        Q_{(+-)}\\ 
        \wt{Q}_{(+-)}
    \end{pmatrix}^{-1}
    \begin{pmatrix}
        Q_{(+-)} \\ 
        0
    \end{pmatrix} 
    =
    \begin{pmatrix}
        \ph{\frac12} & -\frac12\\
        -\frac12 & \ph{\frac12}
    \end{pmatrix}~.
\end{equation}
We will see in the next section that this expression for the change of basis generalises well to other Abelian gauge theories.

\subsubsection{\label{subsubsec:lattice_semi_infinite}Semi-infinite cohomology of the bosonised algebra}

Now let us address the computation of the relative semi-infinite cohomology of $V_L^!$. We can equally consider the semi-infinite cohomology of $V_{\tilde{L}}^!$ and subsequently restrict to the sublattice $L$. Given the factorisation \eqref{eq:ff_factorisation} and the fact that the $\gf$ action is generated by $\eta$, we have the simplification
\begin{equation}
    H^{\frac{\infty}{2}+\bullet}(\gf,\gf_0, V^!_{\tilde{L}}) = H^{\frac{\infty}{2}+\bullet}(\gf,\gf_0, V_\eta)\otimes \wt{V}_{{\rm FFR}}~.
\end{equation}
Thus our only task is to evaluate the relative semi-infinite cohomology of the half-lattice vertex algebra $V_\eta\cong\mathcal{D}^{ch}(\mathbb{C}^\times)$. To perform this evaluation, we observe the following simple structure of $V_\eta$ as a $\gf$ module,
\begin{equation}\label{eq:CDO_g_decomp}
    \mathcal{D}^{ch}(\mathbb{C}^\times) \cong \bigoplus_{\lambda\in\frac12\mathbb{Z}}\big(\mathcal{U}(\mathfrak{n}_{-})\otimes\mathbb{C}_\lambda\otimes\mathcal{U}(\mathfrak{n}_{+})^\ast\big)~.
\end{equation}
Here we are utilising the semi-infinite structure of $\wt{\mathfrak{gl}}_1$ in the sense of Voronov~\cite{voronov1994semi}---reviewed in Appendix~\ref{app:lattice_bosons}, \cf\ \eqref{eq:semi-infinite}. In particular, $\mathbb{C}_\lambda$---corresponding to the lattice momentum $e^{-\frac{\lambda}{2}\wt{\eta}}$---is the one-dimensional representation of the zero mode of $J_{\eta}$ with charge $\lambda$. The only slightly nontrivial observation incorporated into \eqref{eq:CDO_g_decomp}, though it is a straightforward calculation, is that the $J_{\wt{\eta}}$ Fock space, as an $\mathfrak{n}_+$-module, is isomorphic to the contragredient dual of the free module $\cU(\mathfrak{n}_+)$.

From this, it is immediate to see that we can write
\begin{alignat}{3}
    &\mathcal{D}^{ch}(\mathbb{C}^\times) \cong \mathcal{U}(\mathfrak{n}_+)^\ast\otimes \left(\bigoplus_{\lambda\in\frac12\mathbb{Z}}\mathcal{U}(\mathfrak{n}_{-})\otimes\mathbb{C}_\lambda\right)\qquad &&{\rm in}\qquad &&\mathfrak{n_{+}}{\rm -mod}~,\\
    &\mathcal{D}^{ch}(\mathbb{C}^\times) \cong \mathcal{U}(\mathfrak{g})\otimes_{\mathcal{U}(\mathfrak{b}_+)}\left(\bigoplus_{\lambda\in\frac12\mathbb{Z}}\mathbb{C}_\lambda\otimes\mathcal{U}(\mathfrak{n}_{+})^\ast\right)\qquad &&{\rm in}\qquad &&\mathfrak{g}{\rm -mod}~.
\end{alignat}
The first line expressed $\mathcal{D}^{ch}(\mathbb{C}^\times)$ as a direct sum of infinitely many copies of $\mathcal{U}(\mathfrak{n}_+)^\ast$, which is injective (in fact, co-free) as an $\mathfrak{n}_+$ module. The second line makes it apparent that it is also projective relative to $\mathfrak{b}_+$. The key result is then an unnumbered Lemma of Voronov (specialised here for our purposes):
\begin{theorem}[\label{thm:voronov}Section 3.11.2 of \cite{voronov1994semi}]
Let $\gf$ be a Lie algebra with a semi-infinite structure, and let $M$ be a $\mathfrak{g}$-module that is injective as an $\mathfrak{n_+}$-module and relatively projective over $\mathfrak{b}_+$. Then the relative semi-infinite cohomology of $M$ is given as
\begin{equation}
    H^{\frac{\infty}{2}+i}(\mathfrak{g},\mathfrak{g}_0,M) = \begin{cases}
        M^{\mathfrak{n}_+}_{\mathfrak{b}_-}~,\qquad &i = 0~,\\
        0~,\qquad &i\neq 0~.
    \end{cases}
\end{equation}
Here $M^{\mathfrak{n}}_{\mathfrak{b}}$ denotes the \emph{semivariants} of $M$,
\begin{equation}
    M^{\mathfrak{n}}_{\mathfrak{b}} = \frac{\{m\in M\,|\,\mathfrak{n}m=0\}}{\{m\in M\,|\,\mathfrak{n}m=0~{\rm and}~m=bm'~{\rm for}~b\in\mathfrak{b}\}}~.
\end{equation}
\end{theorem}

Applying the above theorem to the present case, we see that the semi-infinite cohomology of $V_\eta$ will just restrict to semivariants in degree zero (and vanish in other degrees). But it is easy to see that in fact,
\begin{equation}\label{eq:no-lambda-mom}
    \left(V_{\eta,\wt{\eta}}\cdot e^{\frac{\lambda}{2}\wt{\eta}}\right)^{\mathfrak{n}_+}_{\mathfrak{b}_-}=\mathbb{C}\delta_{\lambda,0}~,
\end{equation}
so we have the nice result
\begin{equation}\label{eq:TSU2-VFFFR-coh}
    H^{\frac{\infty}{2}+\bullet}(\gf,\gf_0,V_\eta) =\mathbb{C}\qquad\Longrightarrow\qquad     H^{\frac{\infty}{2}+\bullet}(\gf,\gf_0, V^!_{\tilde{L}}) \cong \wt{V}_{{\rm FFR}}~.
\end{equation}
As the cohomology of $V_\eta$ is supported entirely in the sector with no $\wt{\eta}$ momentum, the restriction to the sublattice $L^!$ is immediate, and we have
\begin{equation}
    H^{\frac{\infty}{2}+\bullet}(\gf,\gf_0, V_{L}^!) \cong V_{{\rm FFR}}~,
\end{equation}
where $V_{\rm FFR}$ is the subalgebra of $V_{L}^!$ with lattice momenta restricted to the sublattice,
\begin{equation}
    \mathbb{Z}\left(\phi+\delta\right)\oplus
    \mathbb{Z}\left(\tfrac{\phi+\delta}{2}+\omega_1\right)\oplus
    \mathbb{Z}\left(\tfrac{\phi+\delta}{2}+\omega_2\right)~.
\end{equation}
The \emph{ad hoc} realisation of Subsection~\ref{subsec:free-field-ansatz} indeed lies in this subalgebra.

\subsubsection{\label{subsubsec:embedding_gauge}Embedding of \texorpdfstring{$\mathcal{X}$}{X}}

To identify the semi-infinite cohomology for the original gauge theory construction, it remains to characterise the image of the $\mathcal{X}$ in $V_{\rm FFR}$. Abstractly, this is completely characterised in terms of the screening charges, now retrofitted to act on $V_{\rm FFR}$. Observe that in terms of the alternative basis for our full free field space, the screening operators characterising the image of ${\mathcal X}$ are as follows,
\begin{equation}\label{eq:tsu2_reduced_screenings}
\begin{split}
    \mathfrak{s}_1 &= {\rm Res}\,e^{\rho_1} = {\rm Res}\,e^{\frac12{\eta}+\frac12{\delta}-\frac12{\omega_1}+\frac12{\omega_2}}~,\\
    \mathfrak{s}_2 &= {\rm Res}\,e^{\rho_2} = {\rm Res}\,e^{-\frac12{\eta}+\frac12{\delta}+\frac12{\omega_1}-\frac12{\omega_2}}~.
\end{split}
\end{equation}
The screening charges commute with the BRST differential as they act in completely different Fock spaces, so one finds that elements of $V_{\rm FFR}$ that lie in the injective image of the gauge theory cohomology are precisely those that are in the joint kernel of the two restricted screening charges,
\begin{equation}\label{eq:TSU2-screening}
\begin{split}
    \tilde{\mathfrak{s}}_1 &=  {\rm Res}\,e^{\frac12{\delta}-\frac12{\omega_1}+\frac12{\omega_2}}~,\\
    \tilde{\mathfrak{s}}_2 &=  {\rm Res}\,e^{\frac12{\delta}+\frac12{\omega_1}-\frac12{\omega_2}}~.
\end{split}
\end{equation}
We have not endeavoured to independently characterise this joint kernel. However, one can verify that the generators of $V_{1}(\pslf(2|2))$ determined in our \emph{ad hoc} realisation do lie in the kernel of both of these screening charges. Since we know that these are strong generators for the full $A$-twisted vertex algebra, our free field realisation indeed gives the simple quotient.

For practical purposes, there is a simpler procedure for generating (at least a subalgebra of) the gauge theory vertex algebra leveraging these free field results. We can consider simple gauge/BRST-invariant operators written in terms of the gauge theory variables and rewrite these in terms of the bosonised degrees of freedom. Performing the relevant change of basis, a choice of BRST representative can always be made to precisely set to zero all terms involving the $\eta$ bosons, thus projecting into $V_{\rm FFR}$.

In this theory, the obvious candidate gauge/BRST-invariant operators are exactly the BRST-invariant quadratic invariants built out of the elementary gauge theory fields,
\begin{alignat}{3}
    &X_1Y^2~,\quad&&X_2Y^1~,\quad &&X_1Y^1 - X_2Y^2~,\nonumber\\
    &\chi_1 \xi^2~,\quad&&\chi_2\xi^1~,\quad&&\chi_1 \xi^1 - \chi_2 \xi^2~,\label{eq:tsu2-closed-ops}\\
    &Y^j \chi_\beta~,\quad&&X_i\xi^\alpha~.&&\nonumber
\end{alignat}
As we have seen, these generate the full $V_1(\pslf(2|2))$ vertex algebra. (If we did not know that this was the full cohomology, we would only be able to say that these will generate a (not necessarily strict) sub-VOA of the full VOA.) 

The details of all the steps in the aforementioned procedure are shown in Table~\ref{tab:TSU2-free-field}. We recall that the current $J_\psi$ appearing in the table is the Cartan generator of the additional $V_1(\slf_2)$ current algebra in the free fermion system (\cf\ \eqref{eq:ferm-currents}),
\begin{equation}
    J_\psi =-\frac{\epsilon_{\alpha\beta} \psi_1^\alpha\psi_2^\beta}{2}~.
\end{equation}
We note that after changing variables, these operators are all manifestly independent of $\wt{\eta}$, which is a simple consequence of gauge/BRST invariance. (It is also clear that they are in the kernel of the screening charges~\eqref{eq:TSU2-screening}.) In the final step, which amounts to finding good BRST representatives, all terms involving $J_{\eta}$ are set to zero. Indeed, this can be done quite na\"ively, because the manifest $\wt{\eta}$ independence evade normal ordering issues as well as trivialising the problem of finding the BRST-exact correction to eliminate such terms.

Finally, observe that while the free field expressions for the last four mixed currents $\Theta^{\alpha}_{A2}$ is not written in a manner that is manifestly covariant with respect to ${\rm SL}(2)_o$, they are in fact identically equal to the expected covariant expressions~\eqref{eq:tsu2_ffr_lowest_weight_computed},
\begin{equation}\label{eq:TSU2-cov-theta}
  \Theta_{A2}^{\alpha} =  \frac{1}{2} \left(\partial \psi_{A}^{\alpha}-\psi_{A}^{\alpha}\partial\delta + \frac{1}{3} \epsilon_{\beta\gamma}\epsilon^{BC}\psi_B^\alpha \psi_A^\beta \psi_C^\gamma \right)e^{-\frac{1}{2}(\phi+\delta)}~.
\end{equation}
As anticipated, we have exactly recovered our \emph{ad hoc} free field realisation but using a first principles approach.

\begin{table}[t!]
    \centering
    \resizebox{\textwidth}{!}{
    \begin{tabular}{c|c|c|c|c}
        & Gauge theory & Bosonisation & Change of Variables & FFR Representative\\ 
        \hline
        $e$                          & $X_1Y^2$         & $e^{-\sigma_1-\sigma_2+\rho_1+\rho_2}$                               &  $e^{\phi+\delta}$ & $e^{\phi+\delta}$ \\
        $f$                          & $X_2Y^1$         & $-J_{\rho_1} J_{\rho_2} e^{\sigma_1+\sigma_2-\rho_1-\rho_2}$         &  $\left((J_\psi+\frac{1}{2}J_{\eta})^2-\frac{1}{4}J_\delta^2\right) e^{-\phi-\delta}$ &  $\left(J_\psi^2-\frac{1}{4}J_\delta^2\right) e^{-\phi-\delta}$   \\
        $h$                          & $X_1Y^1-X_2Y^2$  & $J_{\sigma_1}+J_{\sigma_2}$                                              & $-J_\phi$ & $-J_\phi$ \\
        $\Theta_{11}^1$              & $-Y^2\chi_2$     & $-e^{-\sigma_2+\rho_2 + \gamma_2}$                                       & $\psi_1^1 e^{\frac{1}{2}(\phi+\delta)}$ & $\psi_1^1 e^{\frac{1}{2}(\phi+\delta)}$\\
        $\Theta_{11}^2$              & $Y^2\chi_1$      & $e^{-\sigma_2+\rho_2 - \gamma_1}$                                       & $\psi_1^2 e^{\frac{1}{2}(\phi+\delta)}$ & $\psi_1^2 e^{\frac{1}{2}(\phi+\delta)}$\\
        $\Theta_{21}^1$              & $X_1\xi^1$       & $e^{-\sigma_1+\rho_1 + \gamma_1} $                                       & $\psi_2^1 e^{\frac{1}{2}(\phi+\delta)}$ & $\psi_2^1 e^{\frac{1}{2}(\phi+\delta)}$ \\
        $\Theta_{21}^2$              & $X_1\xi^2$       & $e^{-\sigma_1+\rho_1 -\gamma_2} $                                        & $\psi_2^2 e^{\frac{1}{2}(\phi+\delta)}$ & $\psi_2^2 e^{\frac{1}{2}(\phi+\delta)}$ \\
        $\Theta_{12}^1$              & $Y^1\chi_2$      & $J_{\rho_1}e^{\sigma_1 - \rho_1 + \gamma_2}$                             & $-\left( \frac{1}{2}(J_\delta + J_{\eta} )+ J_\psi\right) \psi_1^1 e^{-\frac{1}{2}(\phi+\delta)}$ & $-\left( \frac{1}{2}J_\delta + J_\psi\right) \psi_1^1 e^{-\frac{1}{2}(\phi+\delta)}$ \\
        $\Theta_{12}^2$              & $-Y^1\chi_1$     & $-J_{\rho_1} e^{\sigma_1-\rho_1-\gamma_1}$                               & $-\left( \frac{1}{2}(J_\delta + J_{\eta} )+ J_\psi\right) \psi_1^2 e^{-\frac{1}{2}(\phi+\delta)}$ & $-\left( \frac{1}{2}J_\delta + J_\psi\right) \psi_1^2 e^{-\frac{1}{2}(\phi+\delta)}$  \\
        $\Theta_{22}^1$              & $X_2\xi^1 $      & $-J_{\rho_2}e^{\sigma_2-\rho_2+\gamma_1} $                               & $-\left( \frac{1}{2}(J_\delta - J_{\eta} )- J_\psi\right) \psi_2^1 e^{-\frac{1}{2}(\phi+\delta)}$ & $-\left( \frac{1}{2}J_\delta - J_\psi\right) \psi_2^1 e^{-\frac{1}{2}(\phi+\delta)}$ \\
        $\Theta_{22}^2$              & $X_2\xi^2$       & $-J_{\rho_2} e^{\sigma_2-\rho_2 -\gamma_2} $                             & $-\left( \frac{1}{2}(J_\delta - J_{\eta} )- J_\psi\right) \psi_2^2 e^{-\frac{1}{2}(\phi+\delta)}$ & $-\left( \frac{1}{2}J_\delta - J_\psi\right) \psi_2^2 e^{-\frac{1}{2}(\phi+\delta)}$ \\
        $J_1^1$                & $\frac12(\chi_1\xi^1-\chi_2\xi^2) $ & $\frac{1}{2}(e^{\gamma_1}e^{-\gamma_1} - e^{-\gamma_2}e^{\gamma_2}$ )     & $ \frac{1}{2}(\psi_1^1\psi_2^2+\psi_1^2\psi_2^1)$&$ \frac{1}{2}( \psi_1^1\psi_2^2+\psi_1^2\psi_2^1)$\\
        $J_2^2$                & $-\frac12(\chi_1\xi ^1-\chi_2\xi^2) $ & $-\frac12(e^{\gamma_1}e^{-\gamma_1} - e^{-\gamma_2}e^{\gamma_2})$      & $-\frac12 (\psi_1^1\psi_2^2+\psi_1^2\psi_2^1)$&$  -\frac12 (\psi_1^1\psi_2^2+\psi_1^2\psi_2^1 )$\\
        $J_1^{2}$                    & $\chi_1\xi^2$       & $e^{\gamma_1}e^{\gamma_2}$ & $\psi_1^2\psi_2^2$&$ \psi_1^2\psi_2^2$ \\
        $J_2^{1}$                    & $\chi_2\xi^1$       & $e^{-\gamma_2}e^{-\gamma_1}$  & $\psi_2^1 \psi_1^1$&$  \psi_2^1 \psi_1^1$
    \end{tabular}
    }
    \caption{\label{tab:TSU2-free-field}Outline of derivation of the free field realisation of $V_1(\pslf(2|2))$ of type $(+-)$.}
\end{table}

\subsubsection{\label{subsubsec:free_stress}Conformal vectors}

The direct solution for the semi-infinite cohomology presented above allows us to express the stress tensor/conformal vector of $V_1(\pslf(2|2))$ in the free field variables easily. The starting point is the stress tensor for the gauge theory degrees of freedom, including the $(b,c)$ ghosts associated to the bulk vector multiplet,
\begin{equation}
   T = \frac{1}{2}\left(\partial X_a Y^a - \partial Y^a X_a  \right) + \frac12 (\partial \chi^a \xi_a+  \partial \xi_a \chi^a ) + (\partial c)b~.
\end{equation}
This operator is always BRST closed. The realisation of the same operator in $V_{L^!}$ (in terms of the redefined basis) is as follows,
\begin{equation}
   T = T_\phi + T_\delta + T_{\omega_1} + T_{\omega_2} - \frac{1}{2}J_{\eta}J_{\wt{\eta}} + (\partial c)b~,
\end{equation}
where we have the following expressions for stress tensors of the various chiral bosons appearing in the free field realisation,
\begin{alignat}{4}\label{eq:free-stress-+-}
      &T_{\phi} &&= -\frac{1}{4}J_\phi J_\phi~,\qquad 
      &&T_{\delta} &&=  \frac{1}{4}J_\delta J_\delta -\frac{1}{2}\partial J_{\delta}~,\\
      &T_{\omega_1} &&= \frac{1}{2} J_{\omega_1}J_{\omega_1}~,
      &&T_{\omega_2} &&=  \frac{1}{2} J_{\omega_2}J_{\omega_2}~.
\end{alignat}
On the other hand, the stress tensor for $V_\eta$ and that for the $(b,c)$ ghosts combine to be BRST-exact,
\begin{equation}
    - \frac{1}{2}J_{\eta}J_{\wt{\eta}}  + (\partial c)b = -\frac{1}{2}d\left( J_{\wt{\eta}}b  \right)~,
\end{equation}
so indeed, in $V_{\mathrm{FFR}}$ we have
\begin{equation}
    T = T_\phi + T_\delta + T_{\omega_1} + T_{\omega_2}~.
\end{equation}
This assigns conformal weights $\pm\frac12$ to $e^{\pm \frac{1}{2}(\phi+\delta)}$, while the fermions $\psi_A^\alpha$ are all symmetrically weighted (with conformal weight $\frac{1}{2}$). It can be confirmed by direct calculation that in our free field realisation, this is identically equal to the Sugawara stress tensor for $V_{1}(\pslf(2|2))$ as well as the sum of Sugawara stress tensors for $V_1(\slf_2)$ and $V_{-1}(\slf_2)$ subalgebras. Thus the null vector relating these Sugawara vectors is indeed set to zero in this vertex algebra---this reinforces the expectation that the VOA is the simple quotient.

\subsubsection{\label{subsubsec:outer_aut_tsu2}A note on outer automorphism and symmetry groups}

We conclude this subsection with some remarks about the outer automorphism symmetry of this vertex algebra. Pragmatically speaking, the complete $\mathrm{SL}(2)_o$ covariance of our free field realisation follows directly from the relatively simple covariant expressions for the $\Theta_{A1}^{\alpha}$, as the $\Theta_{A2}^\alpha$ then appear in the OPES of these with with $f$, which is an $\mathrm{SL}(2)_o$ singlet. (Covariance of the other generators is obvious). This argument already applies at the level of the \emph{ad hoc} analysis earlier in this section. 

Alternatively, and perhaps more instructively, we can observe that the screening charges in \eqref{eq:tsu2_reduced_screenings} themselves transform in the standard representation of the additional $\slf_2$ algebra generated by the free fermions $V_{\rm FFR}$. In bosonised terms, this is the zero-mode algebra of the current algebra generated by
\begin{equation}\label{eq:extra_fermion_currents_bosonised}
    e^{\omega_1-\omega_2}~,\qquad e^{\omega_2-\omega_1}~,\qquad \partial\omega_1 - \partial\omega_2~.
\end{equation}
Indeed, it is easy to see that the $\mathbb{C}$-span of $\tilde{\mathfrak{s}}_{1,}$ and $\tilde{\mathfrak{s}}_2$ is invariant under this $\slf_2$. However, as the corresponding currents are not themselves annihilated by the screening charges, this structure gives rise to an \emph{outer automorphism} of the $V_1(\pslf(2|2))$ algebra that forms the joint screening charge kernel.

It was suggested in~\cite{Costello:2018swh} that this outer automorphism should be understood as a reflection of the (IR enhanced) global symmetry of the \emph{Coulomb branch} of $T[\mathrm{SU}(2)]$. In fact, as explained there, in the absence of extra boundary degrees of freedom the Neumann boundary condition utilised in this construction breaks the $\mathrm{U}(1)$ topological symmetry, but the inclusion of additional free fermions restores a mixture of that symmetry and the global symmetry of the fermions. On these grounds, the outer automorphism symmetry is expected to coincide with at least a subgroup of the IR symmetry acting on local bulk Coulomb branch operators. In the present context we have the full Coulomb branch symmetry acting.

\medskip

\paragraph{Remark.} We observe that at the level of groups, the symmetry acting on the free fermions is in fact the full $\mathrm{SU}(2)_C$, rather than the $\mathrm{SO}(3)_C$ Coulomb branch symmetry that acts faithfully in the bulk~\cite{Eckhard:2019jgg,
Gang:2018wek}. Incidentally, the same holds for the (inner) Higgs branch flavour symmetry: the odd generators $\Theta_{Aa}^\alpha$ transform in representations of $\mathrm{SU}(2)_H$ instead of $\mathrm{SO}(3)_H$ (the symmetry acting faithfully on bulk Higgs branch operators). It would be interesting to explore consequences of these phenomena in the context of the generalised symmetries of the bulk theory, which were recently explored in~\cite{Bhardwaj:2023zix}.

\subsection{\label{subsec:TSU2-geom}Shadows of chiral localisation on the Higgs branch}

We propose to interpret this free field realisation as a kind of localisation of $V_1(\pslf(2|2))$ on its associated variety. As discussed in Appendix~\ref{app:lattice_bosons}, the bosonisation procedure for symplectic bosons \emph{can} be understood precisely as a localisation of CDOs from $\mathbb{C}$ to $\mathbb{C}^\times$, so more precisely we can certainly interpret our realisation as arising from localisation before performing BRST reduction.

The purpose of this subsection is to develop an instructive, \emph{finite-dimensional} (non-chiral) analogy for this operation. In this simpler context, we can concretely understand the relationship between localisation before and after gauging (symplectic reduction).

\medskip

\paragraph{Remark.} It would be interesting to incorporate the free fermion degrees of freedom into the finite-dimensional analogue discussed here by considering the symplectic reduction of $V\oplus \Pi V$, where $\Pi V$ is the Grassmann-odd vector space where the free fermions live. The physical interpretation of such an analogy is somewhat unclear to us, but mathematically it may be relevant for understanding the interplay between boundary fermionic degrees of freedom and localisation on the associated variety.

\subsubsection{\label{subsubsec:pm_localisation}Localisation of type $(+-)$ and symplectic reduction} 

We recall the computation of the Higgs branch/Higgs branch chiral ring via holomorphic symplectic reduction. Recall that we had (see \eqref{eq:TSU2_Higgs_chiral_ring_sympquotient}),
\begin{equation}
    \mathbb{C}[\cM_H]=\left(\mathbb{C}[T^\ast\mathbb{C}^2]/\langle X_1Y^1+X_2Y^2\rangle\right)^{T_{G,\mathbb{C}}}\cong\mathbb{C}[A_1]~.
\end{equation}
Performing the analogous localisation to that arising in our bosonisation to CDOs amounts to removing the hyperplanes $X_1=0$ and $Y^2 = 0$ from $T^\ast\mathbb{C}^2$ before performing symplectic reduction. At the level of rings of functions, this leaves us with the localisation
\begin{equation}
    \mathbb{C} \left[T^\ast  (\mathbb{C}^\times)^2\right] = \mathbb{C}\left[  \left( Y^2 \right)^{\pm 1} , (X_1)^{\pm 1} , Y^1 , X_2 \right]~.
\end{equation}
We can equivalently write this as
\begin{equation}\label{eq:A1-reduction}
\begin{split}
     \mathbb{C} \left[T^\ast  (\mathbb{C}^\times)^2\right] &\cong \mathbb{C}\left[X_1Y^1 + X_2 Y^2 , \left(Y^2\right)^{\pm 1} \right] \otimes \mathbb{C}\left[X_1Y^1  , \left(X_1 Y^2\right)^{\pm 1} \right]  \\
     &\cong \mathbb{C} \left[ T^\ast  T_{G,\mathbb{C}}\right] \otimes \mathbb{C} \left[ T^\ast  \mathbb{C}^\times \right]~.
\end{split}
\end{equation}
In doing so, we have identified $X_1Y^1+X_2Y^2$ with the cotangent directions to $T_{G,\mathbb{C}}$ via the complex moment map, $(Y^2)$ with the standard coordinate on $T_{G,\mathbb{C}}$ (as it has weight $+1$), and the Poisson brackets follow from~\eqref{eq:mirror-sqed2-def}. On the other hand, the coordinates on the complementary $T^\ast\mathbb{C}^\times$ are gauge invariant, so symplectic reduction just gives
\begin{equation}
    R_{(+-)} \coloneqq \mathbb{C} \left[ T^\ast  \mathbb{C}^\times \right]~,
\end{equation}
with coordinate functions
\begin{equation}
    \wt{e}^{\,\pm1} \coloneqq (X_1 Y^2)^{\pm1}~,\qquad \wt{h} \coloneqq X_1Y^1 = \frac{1}{2}(X_1 Y^1-X_2Y^2)~.
\end{equation}
There is a Poisson algebra homomorphism
\begin{equation}\label{eq:TSU2-patch-map}
\begin{split}
    R_H &~~\hookrightarrow~~ R_{(+-)}~,\\        
    (e,\,h,\,f) &~~\mapsto~~ \left(\wt{e},\,\wt{h},\,-\frac{\wt{h}^2}{4\wt{e}}\right)~,
\end{split}
\end{equation}
and from this we see that there is an isomorphism of Poisson algebras,
\begin{equation}\label{eq:TSU2-patch-iso}
    R_{(+-)} \cong R_H [e^{-1}]~.
\end{equation}
This identifies the spectrum of $R_{(+-)}$ with an open subvariety of (the smooth locus of) the Higgs branch,
\begin{equation}
    T^\ast  \mathbb{C}^\times \hookrightarrow \cM_H~.
\end{equation}
So indeed, we see that our localisation upstairs induces a localisation on the associated variety/Higgs branch, and this is precisely the finite-dimensional geometric construction that has been proposed to underlie many of the four-dimensional free field realisations discussed elsewhere by one of the authors.

\subsubsection{\label{subsubsec:coulomb_perspective}The mirror Coulomb branch perspective}

The simple finite-dimensional analysis above fits nicely with certain recent developments on restriction functors between algebras associated to three-dimensional Coulomb branches~\cite{Kamnitzer:2022zkv}.

We recall that given an Abelian theory defined by a charge matrix $\wt{Q}$---which for now we allow to be general, as the following will apply more broadly than the present example---there exists a map from its Coulomb branch chiral ring to that of a theory obtained by removing an arbitrary collection of weights from $\wt{Q}$~\cite{Braverman:2016wma}. A recent result of~\cite{Kamnitzer:2022zkv} identifies certain conditions under which these maps can be understood as resulting from an open immersion of the Coulomb branch of the latter theory into that of the former. In the case of Abelian gauge theories, these conditions can be understood as the requirement of non-negativity of the pairing between certain selected monopole operators (specified by co-characters) and the characters encoded in the columns of $\wt{Q}$. 

The above considerations regarding localisation on the Higgs branch of $T[\mathrm{SU}(2)]$ are an extremely simple instance of this general result when applied to the Coulomb branch of the (same) mirror theory defined by the charge matrix $\wt{Q}_{(+-)}$ (\cf\ \eqref{eq:plus_minus_charge_matrices}). From the mirror Coulomb branch perspective, the operator $e$ is the monopole operator $v_1$ associated to the co-character $A=(1)$. Let us denote by $w_{\pm 1}$ the elementary monopole operators of a pure $\mathrm{U}(1)$ gauge theory. These satisfy the simple relation
\begin{equation}
    w_1 w_{-1} = 1~,
\end{equation}
and along with the vector multiplet scalar, $\varphi$, can be identified with coordinate functions on $T^\ast\mathbb{C}^\times$. The homomorphism~\eqref{eq:TSU2-patch-map} can be interpreted as a homomorphism between Coulomb branch chiral rings
\begin{equation}
    v_- \mapsto \varphi^{2} w_-~,\qquad v_+ \mapsto w_+~,\qquad\varphi \mapsto \varphi~,
\end{equation}
where we are using the same symbols for the dual photons of both theories. According to Theorem 2.9 of~\cite{Kamnitzer:2022zkv}, since the pairing between the co-character $A=(1)$ and the characters $\wt{Q}$ are positive---$\langle A,\wt{Q}_{(+-)}^j\rangle=1$ for $j=1,2$---this homomorphism becomes an isomorphism upon inverting $e$ in the Coulomb branch chiral ring of $T[\mathrm{SU}(2)]$. Indeed, this is precisely the content of the isomorphism~\eqref{eq:TSU2-patch-iso}.

\subsubsection{\label{subsubsec:alt_localistion}An alternative localisation: \texorpdfstring{$(++)$}{(++)}}

In the analysis to this point, we made the particular $(+-)$ choice of bosonisation. Here we consider the inequivalent $(++)$ bosonisation at the level of Higgs branches. (We will return to vertex algebras in the next subsection.) This means we are localising with respect to $X_1$ and $X_2$ prior to symplectic reduction.

After localising, we again have the cotangent to an algebraic torus $T^\ast(\mathbb{C}^\times)^2$, this time with coordinate ring
\begin{equation}
    \mathbb{C}\left[T^\ast(\mathbb{C}^\times)^2\right]=\mathbb{C}\left[X_1^{\pm 1},X_2^{\pm 1}, Y^1, Y^2\right]~.
\end{equation}
As in~\eqref{eq:A1-reduction}, we can rearrange this so as to separate out the coordinate ring of the (cotangent of) the gauge torus $T_{G,\mathbb{C}}$ from an otherwise gauge-invariant part,
\begin{equation}\label{eq:A1-reduction-alt}
\begin{split}
     \mathbb{C} \left[T^\ast (\mathbb{C}^\times)^2\right] &\cong \mathbb{C}\left[Y^1X_1 + Y^2 X_2 , \left(X_2\right)^{\pm 1} \right] \otimes \mathbb{C}\left[Y^1X_1  , \left(X_1  X_2^{-1}\right)^{\pm 1} \right]  \\
     &\cong \mathbb{C} \left[ T^\ast T_{G,\mathbb{C}}\right] \otimes \mathbb{C} \left[ T^\ast \mathbb{C}^\times \right]~,
\end{split}
\end{equation}
where all isomorphisms are as Poisson algebras. Symplectic reduction then kills off $T^\ast T_{G,\mathbb{C}}$ and we have (at the level of functions),
\begin{equation}
    R_{(++)} =  \mathbb{C} [T^\ast \mathbb{C}^\times] \cong  \mathbb{C} \left[ \wt{h} , \wt{e}^{\,\pm 1} \right]
\end{equation}
where
\begin{equation}
    \wt{e} \coloneqq X_1 X_2^{-1}~,\qquad\wt{h}\coloneqq Y^1X_1~.
\end{equation}
Now for our Poisson algebra homomorphism we have
\begin{equation}
\begin{split}
    R_H &~~\hookrightarrow~~R_{(++)}~,\\
    (e,\,h,\,f) &~~\mapsto~~ \left( -\wt{h}\wt{e},\,\wt{h},\,\wt{h} \wt{e}^{\,-1}\right)~.
\end{split}
\end{equation}
We can see that this does not determine an open immersion of $R_{(++)}$ into $\cM_H$ (let alone into the smooth locus). Indeed, all of $\{\wt{h}=0\}$ maps to the origin of $\cM_H$, so the fibre over zero is a full copy of $\mathbb{C}^\times$.

However, there \emph{is} an open immersion into the resolution of the Higgs branch $\cM_{H,\zeta}\cong T^\ast\mathbb{C}\mathbb{P}^1$ (defined by symplectic reduction with FI parameter $\zeta < 0$). Indeed, then the $\mathbb{C}^\times$ fibre of our map over the origin of $\cM_H$ is sent to the $\mathbb{C}^\times\subset\mathbb{CP}^1$ subset of the exceptional divisor where the north and south poles are removed. (This picture can be illuminated by a more careful analysis in terms of the $\mathrm{Proj}$ construction of the resolved Higgs branch, or in terms of GIT stability conditions. We postpone such discussions until Section~\ref{sec:fin-geom} where we provide a bit more detail.)

\subsection{\label{sec:TSU2-sec-bos}The alternative free field realisation}

Finally, let us return to analyse the free field realisation of $V_1(\pslf(2|2))$ obtained by implementing the $(++)$ type bosonisation. The steps to obtain the free field realisation are as before \emph{mutatis mutandis}, and for the sake of notational simplicity we will recycle the same symbols for the free field bosons here. We make the slightly different identification for the free fermions in terms of their lattice bosons (the difference with respect to~\eqref{eq:tsu2_fer_bosonisation} arises from our different conventions for bosonisation),
\begin{equation}\label{eq:fer-red}
\begin{split}
    \psi_1^1 \coloneqq -e^{-\omega_2}~,\qquad \psi_2^2 \coloneqq e^{\omega_2}~,\\
    \psi_1^2 \coloneqq \ph{e^{-\omega_1}}~,\qquad  \psi_2^1 \coloneqq e^{\omega_1}~,
\end{split}
\end{equation}
and the current $J_\psi$ is again given by the expression~\eqref{eq:ferm-currents}. We display the relevant expressions for the BRST derivation of a free field realisation in Table~\ref{tab:TSU2-free-field-2}. The expressions in the final column give rise to the same $V_1(\pslf(2|2))$ OPEs as before, but there are some notable differences to observe here compared to the $(+-)$ case. 

\begin{table}[t!]
    \centering
    \resizebox{\textwidth}{!}{
    \begin{tabular}{c|c|c|c|c}
        & Gauge theory & Bosonisation & Change of Variables & FFR Representative\\
        \hline
        $e$  & $X_1Y^2$ & $J_{\rho_2} e^{-(\sigma_1-\sigma_2)+(\rho_1-\rho_2)}$                             &   $\left(\frac{1}{2}(J_\delta+J_{\eta} )+J_\psi \right)e^{-\phi-\delta}$ & $\left(J_\psi+\frac{1}{2}J_\delta\right)e^{-\phi-\delta}$\\
        $f$  & $X_2Y^1$         & $J_{\rho_1} e^{(\sigma_1-\sigma_2)-(\rho_1-\rho_2)}$ & $\left(-\frac{1}{2}(J_\delta-J_{\eta}) +J_\psi \right) e^{\phi+\delta}$      &$\left(J_\psi-\frac{1}{2}J_\delta\right) e^{\phi+\delta}$\\
        $h$  & $X_1Y^1-X_2Y^2$  & $J_{\sigma_1}-J_{\sigma_2}$  & $J_\phi$ & $J_\phi$\\
        $\Theta_{11}^1$             & $-Y^2\chi_2$     & $-J_{\rho_2}e^{\sigma_2-\rho_2 - \gamma_2}$                                        & $\left(\frac{1}{2}(J_\delta+J_{\eta}) +J_\psi \right)\psi_1^1 e^{-\frac{1}{2}(\phi+\delta)}$ & $\left(\frac{1}{2}J_\delta +J_\psi \right)\psi_1^1 e^{-\frac{1}{2}(\phi+\delta)}$\\
        $\Theta_{11}^2$             & $Y^2\chi_1$      & $J_{\rho_2}e^{\sigma_2-\rho_2 - \gamma_1}$                                        & $\left(\frac{1}{2}(J_\delta+J_{\eta}) +J_\psi \right)\psi_1^2 e^{-\frac{1}{2}(\phi+\delta)}$ & $\left(\frac{1}{2}J_\delta +J_\psi \right)\psi_1^2 e^{-\frac{1}{2}(\phi+\delta)}$\\
        $\Theta_{21}^1$             & $X_1\xi^1$       & $e^{-\sigma_1+\rho_1 + \gamma_1} $                                                 & $\psi_2^1 e^{\frac{1}{2}(-\phi+\delta)}$  & $\psi_2^1 e^{-\frac{1}{2}(\phi+\delta)}$\\
        $\Theta_{21}^2$             & $X_1\xi^2$       & $e^{-\sigma_1+\rho_1 +\gamma_2} $                                                  & $\psi_2^2 e^{\frac{1}{2}(-\phi+\delta)}$  & $\psi_2^2 e^{-\frac{1}{2}(\phi+\delta)}$\\
        $\Theta_{12}^1$             & $Y^1\chi_2$      & $J_{\rho_1}e^{\sigma_1 - \rho_1 - \gamma_2}$                                       & $\left(\frac{1}{2}(J_\delta-J_{\eta})-J_\psi \right)\psi_1^1 e^{\frac{1}{2}(\phi+\delta)}$  & $\left(\frac{1}{2}J_\delta-J_\psi \right)\psi_1^1 e^{\frac{1}{2}(\phi+\delta)}$ \\
        $\Theta_{12}^2$             & $-Y^1\chi_1$     & $-J_{\rho_1} e^{\sigma_1-\rho_1-\gamma_1}$                                         & $\left(\frac{1}{2}(J_\delta-J_{\eta})-J_\psi \right)\psi_1^2 e^{\frac{1}{2}(\phi+\delta)}$   & $\left(\frac{1}{2}J_\delta-J_\psi \right)\psi_1^2 e^{\frac{1}{2}(\phi+\delta)}$\\
        $\Theta_{22}^1$             & $X_2\xi^1 $      & $e^{-\sigma_2+\rho_2+\gamma_1} $                                                    & $\psi_2^1 e^{\frac{1}{2}(\phi+\delta)}$&$\psi_2^1 e^{\frac{1}{2}(\phi+\delta)}$\\
        $\Theta_{22}^2$             & $X_2\xi^2$       & $e^{-\sigma_2+\rho_2 +\gamma_2} $                                                  & $ \psi_2^2 e^{\frac{1}{2}(\phi+\delta)}$&$ \psi_2^2 e^{\frac{1}{2}(\phi+\delta)}$ \\
        $J_1^1$  & $\frac12(\chi_1\xi^1-\chi_2\xi^2) $ & $\frac12(e^{\gamma_1}e^{-\gamma_1} - e^{\gamma_2}e^{-\gamma_2})$ & $\frac12(\psi_1^1\psi_2^2+\psi_1^2\psi_2^1)$ & $\frac12(\psi_1^1\psi_2^2+\psi_1^2\psi_2^1)$\\
        $J_2^2$  & $-\frac12(\chi_1\xi^1-\chi_2\xi^2) $ & $-\frac12(e^{\gamma_1}e^{-\gamma_1} - e^{\gamma_2}e^{-\gamma_2})$ & $-\frac12(\psi_1^1\psi_2^2+\psi_1^2\psi_2^1)$ & $-\frac12(\psi_1^1\psi_2^2+\psi_1^2\psi_2^1)$ \\
        $J_1^{2}$ & $\chi_1\xi^2$ & $e^{\gamma_1}e^{-\gamma_2}$ & $ \psi_1^2\psi_2^2$ & $\psi_1^2\psi_2^2$ \\
        $J_2^{1}$ & $\chi_2\xi^1$ & $e^{\gamma_2}e^{-\gamma_1}$ & $ \psi_2^1\psi_1^1$ & $\psi_2^1\psi_1^1$
    \end{tabular}
    }
    \caption{\label{tab:TSU2-free-field-2}Outline of derivation of the free field realisation of $V_1(\pslf(2|2))$ of type $(++)$.}
\end{table}

First, observe that the conformal weights of the free fields are much different here. This can be seen empirically from the free field realisation itself, but can be derived from first principles by analysing stress tensors. As before, the stress tensor of the gauge theory degrees of freedom can be rewritten so that up to the BRST-exact contribution $- d(J_{\wt{\eta}} b + \partial b)$ (which is set to zero in $\wt{V}_{\mathrm{FFR}}$) it takes the form
\begin{equation}
   T = T_\phi + T_\delta + T_{\omega_1} + T_{\omega_2}~
\end{equation}
where now there are different background charges for the various chiral bosons compared to~\eqref{eq:free-stress-+-},
\begin{alignat}{4}
      &T_{\phi} &&=-\frac{1}{4}J_\phi J_\phi~,\qquad 
      &&T_{\omega_1} &&= \frac{1}{2} J_{\omega_1}J_{\omega_1}  + \frac{1}{2}\partial J_{\omega_1}~,\\
      &T_{\delta} &&=  +\frac{1}{4}J_\delta J_\delta~,\qquad
      &&T_{\omega_2} &&= \frac{1}{2} J_{\omega_2}J_{\omega_2}  + \frac{1}{2}\partial J_{\omega_2}~.
\end{alignat}
The fermions are assigned asymmetric conformal weights, with $\psi^\alpha_1$ being of dimension zero and $\psi^\alpha_2$ being of dimension one. Moreover, $e^{\pm \frac{1}{2}(\phi+\delta)}$ carries conformal weight $0$. This is in alignment with the finite dimensional localisation picture described above, where $e^{\phi+\delta}$ should be thought of as the chiral version of the function that parameterises a patch on the exceptional divisor, and so is weight zero under the conformal $\mathbb{C}^\times$ action on the Higgs branch (which fixes the exceptional divisor pointwise).

Finally, the screening charges in this case are given by
\begin{equation}
\begin{split}
    \tilde{\mathfrak{s}}_1 &= e^{\rho_1} = e^{-\frac{1}{2}\delta - \frac{1}{2}\omega_1 - \frac{1}{2}\omega_2}~,\\
    \tilde{\mathfrak{s}}_2 &= e^{\rho_2} = e^{\frac{1}{2}\delta - \frac{1}{2}\omega_1 - \frac{1}{2}\omega_2}~.
\end{split}
\end{equation}
We see that these are no longer related by the $\slf_2$ action generated by the zero modes of the additional currents built from the (bosonised) free fermions, which in this case are given by
\begin{equation}
    e^{\pm(\omega_1+\omega_2)}~,\qquad \partial\omega_1+\partial\omega_2~.
\end{equation}
(Relatedly, the stress tensor is not invariant under the full symmetry group of the fermions). Indeed this free field realisation manifests only the $\mathfrak{u}(1)$ outer automorphism generated by $J_\psi$.

From a physical perspective, we interpret this lack of outer automorphism covariance as a result of this free field realisation arising by localisation on the \emph{resolved} Higgs branch $\cM_{H,\zeta}$. Resolving the Higgs branch by turning on an FI parameter in the underlying gauge theory breaks that $\slf_2$ symmetry of the Coulomb branch to a Cartan subalgebra. We anticipate that in general, it may be important for manifesting outer automorphism symmetries to identify the bosonisation choices that relate to localisation on an open subset that enjoys an open immersion into the unresolved Higgs branch.

\medskip

\paragraph{Remark.} Free field realisations corresponding to localisation on the resolved $\cM_{H,\zeta}$ in particular appear in the closely related setting in \cite{kuwabara2021vertex}, though in that work the free field expressions were all given purely in terms of symplectic bosons and a Heisenberg algebra. It should be possible to rewrite the present free field expressions in terms of a symplectic boson rather than a half-lattice vertex algebra by appropriate field redefinitions. This would amount to ``adding the north pole back in'' to the exceptional divisor.

\section{\label{sec:general_abelian}Generalisation for \texorpdfstring{$A$}{A}-twisted Abelian gauge theories}

In this section, we consider the generalisation of the above analysis of the $A$-twisted vertex algebra for $T[\mathrm{SU}(2)]$ to a more general class of $\cN=4$ Abelian gauge theories. Generally speaking, the strategy is the same as in our main example: we bosonise and change variables in order to solve an appropriate semi-infinite cohomology problem, and produce a free field realisation for the relevant VOA as the kernel of a collection of screening charges. Though at times technical, our conclusion will ultimately be that essentially all important features of the $T[\mathrm{SU}(2)]$ case persist in the more general setting.

For general Abelian theories, the gauge group is a rank-$k$ torus,
\begin{equation}
    T_G=\prod_{i=1}^k \mathrm{U}(1)~,
\end{equation}
and the matter content will be $N$ hypermultiplets packaged into a quaternionic $T_G$-representation,
\begin{equation}
    V = T^\ast\mathbb{C}^N~.
\end{equation}
Let $T_V\cong \mathrm{U}(1)^N$ again denote the standard torus acting on $\mathbb{C}^N$ (lifted to a Hamiltonian action on $V$), and then the flavour torus $T_H$ is the quotient $T_V/T_G$. 

We will restrict to cases where there is a short exact sequence generalising~\eqref{eq:TSU2-SES},
\begin{equation}
    0 \rightarrow \mathfrak{t}_{G,\mathbb{Z}} \cong \mathbb{Z}^k \xrightarrow{Q^T}  \mathfrak{t}_{V,\mathbb{Z}} \cong \mathbb{Z}^N \xrightarrow{\wt{Q}}  \mathfrak{t}_{H,\mathbb{Z}} \cong \mathbb{Z}^{N-k} \rightarrow 0~,
\end{equation}
where $\mathfrak{t}_{G,\mathbb{Z}}$, $\mathfrak{t}_{V,\mathbb{Z}}$, $\mathfrak{t}_{H,\mathbb{Z}}$ are co-weight lattices of the tori $T_G$, $T_V$, $T_H$. In particular, this assumption requires that the gauge group has rank $k\leqslant N$, and that the map $\wt{Q}$ is surjective. We additionally require $\wt{Q}$ to be unimodular, meaning that any minor consisting of $(N-k)$ columns of $\wt{Q}$ has determinant $\pm 1$. This implies in particular that $\wt{Q}$ can be completed to a matrix $\wt{\mathbf{Q}}$ of determinant $\pm 1$. The matrix $\wt{\mathbf{Q}}$ is then invertible over the integers, so we can define
\begin{equation}\label{eq:gen-mirror-def}
    \wt{\mathbf{Q}} = 
    \begin{pmatrix} \tilde{q} \\ \wt{Q} \end{pmatrix}~,
    \mathbf{Q}=
    \begin{pmatrix} Q \\ q \end{pmatrix}
    =(\wt{\mathbf{Q}}^{-1})^{T}~.
\end{equation}
This provides splitting and projection matrices $q: \mathfrak{t}_{H,\mathbb{Z}}\rightarrow \mathfrak{t}_{V,\mathbb{Z}}$ and $\wt{q}: \mathfrak{t}_{V,\mathbb{Z}}\rightarrow \mathfrak{t}_{G,\mathbb{Z}}$ for the above short exact sequence. The matrices $q$ and $\wt{q}$ encode choices of (integer) flavour charges of the original theory and those of its mirror, respectively. 

\medskip

\paragraph{Remark.} Unimodularity is a technical condition that simplifies our analysis in some places, but for many arguments is not required. We expect the approach of this paper to apply, possibly with minor modifications, in more general situations where $k\leqslant N$.

\medskip

Besides the flavour symmetry $G_H$ (with maximal torus $T_H$ with charge matrix given by $q$), these theories enjoy a UV topological symmetry $T_C$ whose charge operators measure the topological type of gauge bundles on a sphere surrounding an operator,
\begin{equation}
    T_C \cong \pi_1(T_G)^\vee \cong Z(T_G^\vee) \cong \mathrm{U}(1)^k~.
\end{equation}
This symmetry can experience a non-Abelian enhancement to a larger Coulomb branch symmetry $G_C$ at the infrared fixed point. Under mirror symmetry, the flavour symmetry $G_H$ of a theory defined by $\mathbf{Q}$ is identified with the enhanced Coulomb branch symmetry $\wt{G}_C$ of the mirror theory defined by $\wt{\mathbf{Q}}$.

\subsection{\label{subsec:general_higgs_coulomb_branches}Moduli spaces for general Abelian gauge theories}

The $T_G$ action on $V$ is tri-Hamiltonian, with a triplet of moment maps (organised here according to a choice of complex structure),
\begin{equation}
    \mu_{\mathbb{R}}: V\to \mathfrak{t}_{G,\mathbb{R}}^\ast~,\qquad
    \mu_{\mathbb{C}}: V\to \mathfrak{t}_{G,\mathbb{C}}^\ast~.
\end{equation}
The Higgs branch of vacua is given by the hyperk\"ahler quotient of $V$ by this $T_G$ action,
\begin{equation}
   \cM_H =  V /\!\!/\!\!/ T_G~,
\end{equation}
which is equivalently the $T_G$ quotient of the zero-locus of the moment maps,
\begin{equation}\label{eq:hyp-quot-general}
   \cM_H \cong \mu^{-1}_{\mathbb{C}}(0) \cap \mu_{\mathbb{R}}^{-1}(0)/ T_G~.
\end{equation}
Choosing a basis for $\mathfrak{t}_{G}^\ast$, we can write these moment maps in components,
\begin{alignat}{2}
    &\mu_{\mathbb{C}}^a = \sum_{i=1}^N Q^a_i Y^i X_i\,,\quad &&a\in\{1,\ldots,k\}~,\\
    &\mu_{\mathbb{R}}^a = \sum_{i=1}^N \frac{Q^a_i}{2}\left(|X_i|^2-|Y^i|^2\right)\,,\quad &&a\in \{1,\ldots,k\}~.
\end{alignat}
As an algebraic symplectic variety, the Higgs branch is given by the holomorphic symplectic quotient,
\begin{equation}\label{eq:higgs-quot-general}
    \cM_H \cong \mu_{\mathbb{C}}^{-1}(0)/\!\!/T_{G,{\mathbb{C}}}~,
\end{equation}
in terms of the complexified the gauge group $T_{G,\mathbb{C}} = (\mathbb{C}^\times)^k$. The Higgs branch chiral ring/coordinate ring of $\cM_H$ is given by
\begin{equation}
    R_{H} = \mathbb{C} \left[ \mu_{\mathbb{C}}^{-1}(0) \right]^{G_\mathbb{C}}~.
\end{equation}
By definition, this is a hypertoric variety of quaternionic dimension $(N-k)$ which carries a Hamiltonian action of the (complexified) torus $T_{H,\mathbb{C}}\cong(\mathbb{C}^\times)^{N-k}$.

Background vector multiplets for $T_C$ (as defined in the UV) provide resolution parameters for $\cM_H$. More precisely, the bottom component of a vector multiplet is a real FI parameter $\zeta \in \mathfrak{t}_{G,\mathbb{\mathbb{R}}}^*)$, and when set to nonzero values one has for the Higgs branch the resolved space
\begin{equation}\label{eq:higgs-quot-resolved}
    \cM_{H,\zeta} \cong  \mu_{\mathbb{C}}^{-1}(0) /\!\!/_{\zeta} T_{G,{\mathbb{C}}}~,
\end{equation}
which is the projective GIT quotient of the zero locus of the complex moment map. This is provides a resolution of singularities
\begin{equation}
    \cM_{H,\zeta} \twoheadrightarrow \cM_H~.
\end{equation}

The structure of this resolution and the dependence on the FI parameters admits an elegant combinatorial formulation in terms of hyperplane arrangements. We identify $\zeta$ with its image under $\wt{q}^T$ in $\mathfrak{t}^*_{V,\mathbb{R}}$ and define the following hyperplanes in $\mathfrak{t}^*_{H,\mathbb{R}}$,
\begin{equation}\label{eq:hyperplanes}
    H_i \coloneqq \{x \in \mathfrak{t}^*_{H,\mathbb{R}} \ | \ x \cdot \wt{Q}^i + \zeta^i = 0 \}~.
\end{equation}
The collection of hyperplanes $\mathcal{A}_{\zeta} =  \{H_1,\ldots,H_{N}\}$ is known as an \emph{arrangement}. An arrangement is called \emph{simple} if every subset of $m$ hyperplanes with nonempty intersection intersects in codimension $m$, and is called \emph{unimodular} if every collection of $(N-k)$ column vectors $\{\wt{Q}^{i_1},\ldots \wt{Q}^{i_{N-k}}\}$ spans $\mathfrak{t}^*_{H,\mathbb{Z}}$ over the integers.

A hyperplane arrangement completely determines the corresponding (resolved) hypertoric variety $\cM_{H,\zeta}$. This variety is sometimes written as $\cM(\mathcal{A}_\zeta)$ to emphasise this feature. The following are standard results relating properties of the arrangements to the geometry of the hypertoric varieties~\cite{proudfoot2004hyperkahler,proudfoot2008survey}:
\begin{itemize}
    \item The variety $\cM(\mathcal{A}_\zeta)$, is a holomorphic symplectic variety with a Hamiltonian $T_H$ action. It is independent of the signs of the row vectors $\wt{Q}_i$. (As a consequence, it is also independent of the signs of the column vectors of $Q$, which are determined by $\wt{Q}_i$).
    \item The resolution $\cM(\mathcal{A}_\zeta)$ has at worst orbifold singularities if and only if $\mathcal{A}$ is simple, and is smooth if and only if $\mathcal{A}$ is simple and unimodular.
\end{itemize}
For the theories we are considering, one can always find choices of resolution parameters so that the resulting arrangements are simple and unimodular. In other words, the Higgs branches in question always admit (physical) smooth resolutions.

\subsubsection{\label{subsubsec:chiral_ring_gens_general}Chiral ring generators}

It will prove useful in our general analysis, and especially in examples, to have a sense for the algebraic generators of the Higgs branch chiral ring, as well as their description in terms of the Coulomb branch of the mirror theory. We now summarise the characterisation of these chiral ring generators, as explained in detail in, \emph{e.g.}, \cite{Bullimore:2015lsa,Bullimore:2016nji}.

The ``mesons'' $X_iY^i$ for $i=1,\ldots,N$ are manifestly gauge-invariant, but $k$ linear combinations of these are set to zero by the moment map condition $\mu_{\mathbb{C}}=0$. As a basis for linearly independent, non-zero meson operators one can use precisely the $(N-k)$ complex moment maps for the $T_{H,\mathbb{C}}$ action,
\begin{equation}\label{eq:flav-0-gen}
    \mu_{H,\mathbb{C}}^a \coloneqq \sum_{i=1}^{N-k} q^a_i X_iY^i~.
\end{equation}
These are all weight-zero with respect to the $T_{H,\mathbb{C}}$ action on $\cM_H$.

For the remaining generators one takes gauge-invariant monomials of the elementary hypermultiplet scalars, which will have non-zero $T_{H,\mathbb{C}}$ weights. For any $\mathbb{Z}^N$-valued ``weight vector'' $w$, one defines
\begin{equation}
    w_{\pm,i} \coloneqq 
    \begin{cases} 
        |w_i| \quad & \text{if}\quad  \pm w_i \geqslant 0~,\\
        0 & \text{otherwise}~,
    \end{cases}
\end{equation}
as well as
\begin{equation}\label{eq:mon-def}
    W^w \coloneqq \prod_{i=1}^N (X_i)^{w_{+,i}}(Y^i)^{w_{-,i}}~.
\end{equation}
Then operators of the form
\begin{equation}\label{eq:higg-mon-2}
    W^A \coloneqq W^{\wt{Q}^T \cdot A}~,
\end{equation}
are gauge-invariant for any vector $A \in \mathbb{Z}^{N-k}$ and descend non-trivially to $\mathbb{C}[\mu^{-1}_{\mathbb{C}}(0)]^{T_{G,\mathbb{C}}}$. One can choose a finite set of these monomials such that together with~\eqref{eq:flav-0-gen} they algebraically generate the Higgs branch chiral ring. (The precise requirement is that the vectors $A$ for the set of monomials must generate the integral lattice inside each cone in so-called ``canonical slice'' in $\mathfrak{t}_{H,\mathbb{R}}^\ast$ \cite{Bullimore:2016nji}.)

From the mirror perspective, the combinations~\eqref{eq:higg-mon-1} correspond to the $(N-k)$ complexified vector multiplet scalars $\varphi_{a} = \phi_{1,a} + i\phi_{2,a}$,
\begin{equation}
    \mu_{H,\mathbb{C}}^a \leftrightarrow \wt{\varphi}_a~.
\end{equation}
A vector $A\in \mathbb{Z}^{N-k}$ determines by a co-character of the mirror gauge group,
\begin{equation}\label{eq:co-ch}
    A : \mathbb{C}^\times \rightarrow (\mathbb{C}^\times)^{N-k}~.
\end{equation}
and to this co-character $A$ one associates an elementary mirror monopole operator $v_A$. These monopole operators are dual to the operators $W^A$ defined above.

\subsection{\label{subsec:general_boundary_VOAs}\texorpdfstring{$A$}{A}-twisted boundary VOAs}

The boundary VOA for a general $A$-twisted Abelian gauge theory of the type reviewed above is a BRST quotient of $N$ sets of symplectic bosons $X_a(z),Y^b(w)$ and $N$ sets of complex fermions $\xi^\alpha,\chi_\beta$ \cite{Costello:2018fnz} with standard OPEs,
\begin{equation}
    Y^a(z)X_b(w)\sim \frac{\delta^a_b}{z-w}~,\qquad\chi_\alpha(z)\xi^\beta(w) \sim \frac{\delta_\alpha^\beta}{z-w}~.
\end{equation}
We denote the symplectic boson, complex fermions, and total matter VOAs as before,
\begin{equation}
    \mathcal{X}\coloneqq \bigotimes_{i=1}^N\left(\texttt{Sb}_i\otimes\texttt{Ff}_i\right)~.
\end{equation}
The $\mathrm{U}(1)^k$ gauge symmetry corresponds to a rank $k$, level-zero Heisenberg vertex subalgebra of $\mathcal{X}$, generated by the currents
\begin{equation}
    \cJ^a = \sum_{i=1}^{N} Q^a_i (X_iY^i +\chi_i\xi^i )~,\qquad a=1,\ldots,k.
\end{equation}
BRST reduction proceeds by introducing $k$ ghost systems  $(b_a,c^a)$ (all of weight $(1,0)$),
\begin{equation}
    b_a(z)c^b(w)\sim\frac{\delta_a^b}{z-w}~,
\end{equation}
and forming the Feigin standard complex for the $\gf=(\wt{\mathfrak{gl}_1})^k$ module $\cal X$,
\begin{equation}
    \mathcal{X}^{\bullet} = \mathcal{X} \otimes (b,c)^{\otimes k}~. 
\end{equation}
Cohomological grading is by ghost number (with the $c_a$ assigned ghost number $+1$ and the $b^a$ assigned ghost number $-1$). Generalising~\eqref{eq:psl22-BRST}, the differential is identified with the zero-mode of the BRST current
\begin{equation}
    J^{\text{BRST}}= c_a\cJ^a~,\qquad d = J^{\text{BRST}}_0~.
\end{equation}
The VOA for the gauge theory is identified with the cohomology of this BRST operator acting on the subcomplex of $\mathcal{X}^\bullet$ that is annihilated by the zero modes of all $b$-ghosts and gauge currents,
\begin{equation}\label{eq:general_theory_as_semi_infinite}
    H^\bullet_{d}\left(\psi\in\mathcal{X}^\bullet\,|\,\cJ^a_0\psi=b^a_0\psi=0\right) \cong H^{\frac{\infty}{2}+\bullet}(\gf,\gf_0,\mathcal{X})~.
\end{equation}

\subsubsection{\label{subsubsec:general_bosonisation}Semi-infinite cohomology via bosonisation}

To calculate the semi-infinite cohomology \eqref{eq:general_theory_as_semi_infinite}, we can directly adapt the strategy presented in Section~\ref{subsec:tsu2_voa_BRST} in the context of our main example. Consider the vertex algebra $V^!_L$ comprising $N$ copies of $V_{L_{\texttt{Sb}}}^!$ (with attendant chiral bosons $\{\rho_i,\sigma_i\}$, $i=1,\ldots,N$) and $N$ copies of $V_{L_{\texttt{Ff}}}$ (with chiral bosons $\{\gamma_i\}$, $i=1,\ldots,N$),
\begin{equation}
    V^!_L = \bigotimes_{i=1}^N \left(V_{L_{\texttt{Sb}},i}^!\otimes V_{L_{\texttt{Ff}},i}\right)~.
\end{equation}
For a given sign vector $\epsilon$, we bosonise the $i$'th symplectic boson pair in $V_{L_{\texttt{Sb}}}^!$ and the $i$'th free fermion in $V_{L_{\texttt{Ff}}}$ according to the relevant sign in $\epsilon$, and so realise an $\epsilon$-dependent embedding,
\begin{equation}
    \mathcal{X}\xhookrightarrow{\epsilon} V_{L}^! \hookrightarrow V_L~.
\end{equation}
In particular, the sum of the $i$'th symplectic boson and free fermion $\glf_1$ currents bosonise as $(X_iY^i)+(\chi_{i}\xi^i) = \epsilon_i(J_{\sigma_i}+J_{\gamma_i})$~.

We first consider a change of basis analogous to that introduced in our treatment of $T[\mathrm{SU}(2)]$. In fact, we consider a direct generalisation of~\eqref{eq:TSU2-COB},
\begin{equation}\label{eq:gen-COB}
\left(
        \begin{array}{c}
        \boldsymbol\eta \\[1pt] \hdashline
        \\[\dimexpr-\normalbaselineskip+3pt]
        \boldsymbol{\tilde{\eta}}\\[1pt] \hdashline
        \\[\dimexpr-\normalbaselineskip+3pt]
        \boldsymbol\phi \\[1pt] \hdashline
        \\[\dimexpr-\normalbaselineskip+3pt]
        \boldsymbol\delta \\[1pt] \hdashline
        \\[\dimexpr-\normalbaselineskip+3pt]
        \boldsymbol\omega
        \end{array} 
    \right)
    =
    \left(
        \begin{array}{c:c:c}
        ~~Q_{\epsilon}~~  &  ~~0~~ & ~~Q_{\epsilon}~~ \\[1pt] \hdashline
        &&\\[\dimexpr-\normalbaselineskip+3pt]
        ~~Q_{\epsilon}~~  &  ~~-Q_{\epsilon}~~ & ~~0~~ \\[1pt] \hdashline
        &&\\[\dimexpr-\normalbaselineskip+3pt]
        ~~\wt{Q}_{\epsilon}~~  & ~~0~~  & ~~0~~ \\[1pt] \hdashline
        &&\\[\dimexpr-\normalbaselineskip+3pt]
        ~~0~~  &  ~~-\wt{Q}_{\epsilon}~~ \ & ~~0~~ \\[1pt] \hdashline
        &&\\[\dimexpr-\normalbaselineskip+3pt]
        ~~F_{\epsilon}^T~~  &  ~~-F_{\epsilon}^T~~ & ~~\mathbf{1}_{N\times N}~~
        \end{array} 
    \right)
     \left(
        \begin{array}{c}
        {}\\[-5pt]
        \boldsymbol\sigma \\[7pt] \hdashline
        \\[\dimexpr-\normalbaselineskip+2pt]
        {}\\[-9pt]
        \boldsymbol\rho \\[7pt] \hdashline
        \\[\dimexpr-\normalbaselineskip+2pt]
        {}\\[-9pt]
        \boldsymbol\gamma\\[-4pt]
        {}
        \end{array} 
    \right)~,
\end{equation}
where we define the $N\times N$ matrices
\begin{equation}\label{eq:F-A-def}
    A_\epsilon \coloneqq 
    \begin{pmatrix}
        Q_\epsilon \\ \wt{Q}_\epsilon
    \end{pmatrix}~, \qquad 
    B_\epsilon \coloneqq
    \begin{pmatrix}
        Q_\epsilon \\ 0
    \end{pmatrix}~,\qquad
    F_\epsilon = A_\epsilon^{-1}B_\epsilon~,
\end{equation}
and as before, $Q_\epsilon$, $\wt{Q}_\epsilon$ are the charge matrices $Q$ with the signs of the columns flipped according to $\epsilon$ and its mirror. In the expression \eqref{eq:gen-COB} we have assembled various elements of our rank $3N$ Heisenberg algebra into vectors denoted by bold variables, $\boldsymbol\sigma=(\sigma_1, \ldots , \sigma_N)^T$, $\boldsymbol\eta = (\eta_1, \ldots ,\eta_k)^T$, and so on. (Below, we will continue to vectorise various quantities and denote these with boldfaced versions of the same symbols as used to denote the individual components.)

This improved basis enjoys many good properties that generalise those utilised in our analysis of the $T[\mathrm{SU}(2)]$ $A$-twisted VOA. First and foremost, the BRST operator now takes the form
\begin{equation}
    J^{\rm BRST} = c_aJ_{\eta_a}~,
\end{equation}
while the bilinear form encoding the OPEs of Heisenberg currents takes the relatively simple schematic form,
\begin{equation}\label{eq:OPE-blocks}
    \left(
    \begin{array}{c:c:c:c:c}
        ~~0~~ \ & \ \mathrm{OPE}(\boldsymbol{\eta}\times\boldsymbol{\tilde{\eta}}) \ & \ ~~0~~ \ & \ ~~0~~ \ & \ ~~0~~ \\[1pt]\hdashline
        &&\\[\dimexpr-\normalbaselineskip+3pt]
       \mathrm{OPE}(\boldsymbol{\eta}\times\boldsymbol{\tilde{\eta}}) \ & \ ~~0~~ \ & \ 0 \ & \ ~~0~~ \ &\ ~~0~~ \\[1pt]\hdashline
        &&\\[\dimexpr-\normalbaselineskip+3pt]
        ~~0~~ \ &  \ ~~0~~ \ & \ ~~\mathrm{OPE}(\boldsymbol{\phi})~~ \ & \ ~~0~~ & \ ~~0~~ \\[1pt]\hdashline
        &&\\[\dimexpr-\normalbaselineskip+3pt]
        ~~0~~ \ & \ ~~0~~ \ & \ ~~0~~ \ & \ ~~\mathrm{OPE}(\boldsymbol{\delta})~~ \ & \ ~~0~~ \\[1pt]\hdashline
        &&\\[\dimexpr-\normalbaselineskip+3pt]
        ~~0~~ \ & \ ~~0~~ \ & \ ~~0~~ \ & \ ~~0~~ \ & \ ~~\mathrm{OPE}(\boldsymbol{\omega})~~
    \end{array} 
    \right)~,
\end{equation}
where the locations of nontrivial blocks of OPEs are indicated by $\mathrm{OPE}(-)$. The specific form of the change of basis matrix further guarantees that we have
\begin{equation}
    \mathrm{OPE}(\boldsymbol{\delta}) = - \mathrm{OPE}(\boldsymbol{\phi})~,\qquad\mathrm{OPE}(\boldsymbol{\omega}) = \mathbf{1}_{N\times N}~.
\end{equation}
In general, $\mathrm{OPE}(\boldsymbol{\eta}\times\boldsymbol{\tilde{\eta}})$ and $\mathrm{OPE}(\boldsymbol\phi)$ are not diagonal (though they are always non-degenerate). They instead reflect the matrix of inner products amongst the gauge charge vectors (rows of $Q_\epsilon$) in the former case, and amongst the flavour charge vectors encoded in the rows of $\wt{Q}_\epsilon$ in the latter.

The lattice $L^!$ along with the Heisenberg vertex algebra is extended can be expressed in the basis of these new variables by inspecting the inverse of the change-of-basis matrix in \eqref{eq:gen-COB} (technical details can be found in Appendix~\ref{app:cob}). We ultimately find the expression
\begin{equation}
    L^! = \bigoplus_{i=1}^N \mathbb{Z} \alpha_i \oplus \mathbb{Z} \beta_i~,
\end{equation}
where
\begin{equation}
\begin{split}
    \boldsymbol\alpha &\coloneqq A^{-1}_\epsilon
    \begin{pmatrix}
        \boldsymbol{\tilde{\eta}} \\ \mathbf{0}
    \end{pmatrix} 
    + A^{-1}_\epsilon
    \begin{pmatrix}
        \mathbf{0} \\ \boldsymbol\phi + \boldsymbol\delta
    \end{pmatrix}~,\\
    \boldsymbol\beta &\coloneqq A^{-1}_\epsilon
    \begin{pmatrix}
        \boldsymbol{\tilde{\eta}} \\ \mathbf{0}
    \end{pmatrix}
    +\boldsymbol\omega~.
     \end{split}
\end{equation}
Note in particular that $L^!$ does include any vectors with nonzero $\boldsymbol\eta$ components. In analogy with~\eqref{eq:TSU2-VFRR-lattice}, we define a lattice refinement, $\tilde{L}^!$ according to\footnote{Here for dimensional reasons there is a redundancy amongst the lattice directions appearing in each of the first two direct sums; in examples we will write this extension more compactly by identifying basis vectors over $\mathbb{Z}$ for these vector spaces.}
\begin{equation}\label{eq:new-lattice}
    L^!\subset\tilde{L}^! = \left(\bigoplus_{i=1}^{N} \mathbb{Z}\left(A^{-1}_\epsilon \right)_i
    \begin{pmatrix}
        \boldsymbol{\tilde{\eta}} \\ \mathbf{0}
    \end{pmatrix}\right)
    \oplus\left(\bigoplus_{i=1}^N\mathbb{Z}\left(A^{-1}_\epsilon\right)_i
    \begin{pmatrix}
        \mathbf{0} \\ \boldsymbol{\phi}+\boldsymbol{\delta}
    \end{pmatrix}\right)
    \oplus\left(\bigoplus_{i=1}^N\mathbb{Z}\,\omega_i\right)~.
\end{equation}

When the Heisenberg vertex algebra is extended by this refined lattice, then the $\omega_i$ can be treated as bosonisations of $N$ complex free fermions,
\begin{equation}\label{eq:gen_free_fermions}
    \psi^i = e^{-\omega_i}~,\qquad~\wt{\psi}_i = e^{\omega_i}~,
\end{equation} 
with singular OPEs,
\begin{equation}
    \psi^i(z) \wt{\psi}_j(w) \sim \frac{\delta^i_j}{z-w}~.
\end{equation}
We then have the decomposition
\begin{equation}
    V_{\tilde{L}}^! \cong V_{\boldsymbol{\eta}} \otimes \wt{V}_{\mathrm{FFR}}~,
\end{equation}
where $V_{\boldsymbol\eta}$ is the relevant lattice extension (by the first parenthetical grouping in \eqref{eq:new-lattice}) of the Heisenberg vertex algebra generated by the $J_{\eta_i}$ and $J_{\tilde{\eta}_i}$, and $V_{\rm FFR}$ is the free field vertex algebra in which we will ultimately locate our free field realisation, which consists of the lattice extension of the $J_{\phi_i}$ and $J_{\delta_i}$ Heisenbergs by the second parenthetical grouping in \eqref{eq:new-lattice}, as well as the $N$ complex free fermions \eqref{eq:gen_free_fermions}.

The BRST problem is now isolated $V_{\boldsymbol\eta}$, so we have for our semi-infinite cohomology,
\begin{equation}
    H^{\frac{\infty}{2}+\bullet}(\gf,\gf_0, V^!_{\tilde{L}}) \cong H^{\frac{\infty}{2}+\bullet}(\gf,\gf_0, V_{\boldsymbol{\eta}})\otimes \wt{V}_{{\rm FFR}}~.
\end{equation}
Now by a change of basis for the components of $\boldsymbol{\tilde{\eta}}$, we can always diagonalise $\mathrm{OPE}(\boldsymbol{\eta} \times \boldsymbol{\tilde{\eta}})$. Consequently, $V_{\boldsymbol\eta}$ is essentially (up to a possible restriction of charge vectors) just the tensor product of $k$ copies of $V_\eta$ from our main example various gauged $\wt{\glf}_1$'s act separately in each copy. Our analysis for the semi-infinite cohomology from before then carries over essentially without modification, and we have
\begin{equation}
    H^{\frac{\infty}{2}+\bullet}(\gf,\gf_0, V_{\boldsymbol{\eta}})\cong\mathbb{C}~,\qquad  H^{\frac{\infty}{2}+\bullet}(\gf,\gf_0, V^!_{L}) \cong V_{{\rm FFR}}~,
\end{equation}
where $V_{{\rm FFR}}$ is the subalgebra of $V_L^!$ that has no momentum in the $\boldsymbol{\tilde{\eta}}$ direction.

\paragraph{Remark.} By a rational change of variable, the matrix of inner products $\mathrm{OPE}(\boldsymbol{\delta}) = - \mathrm{OPE}(\boldsymbol{\phi})$ can be diagonalised as well. Thus, subject to a refinement of the lattice by which the Heisenberg vertex algebras are extended, we can identify
\begin{equation}\label{eq:refinement}
    V_{\rm FFR}\subseteq V^{\rm ref}_{\rm FFR}\cong \mathcal{D}^{(ch)}(\mathbb{C}^\times)^{\otimes (N-k)}\otimes \texttt{Ff}^{\otimes N}~.
\end{equation}
Though this is a nice and uniform expression, in our study of examples, we will not always perform such a change of basis and refinement.

\subsection{\label{subsec:general_embedding}Embedding of the \texorpdfstring{$A$}{A}-twisted vertex algebra}

As in our main example, the gauge theory VOA is now identified with the intersection of the kernels of the $N$ screening currents $\mathfrak{s}_i =  e^{\rho_i}$ restricted to act on $V_{\mathrm{FFR}}$. These restrictions can be identified using the expression for the $\rho_i$ in terms of the modified Heisenberg basis,
\begin{equation}\label{eq:rhodef}
    \boldsymbol\rho =  A^{-1}_\epsilon  
    \begin{pmatrix}
        \boldsymbol{\eta} \\ -\boldsymbol{\delta}
    \end{pmatrix}
    - F_{\epsilon}\,\boldsymbol{\omega}~.
\end{equation}
We write the restrictions to $V_{\mathrm{FFR}}$ as
\begin{equation}\label{eq:general_screenings}
    \tilde{\mathfrak{s}}_i = e^{\bar{\rho}_i}~,
\end{equation}
where
\begin{equation}\label{eq:free_field_screenings}
    \boldsymbol{\bar{\rho}} =  -A^{-1}_\epsilon
    \begin{pmatrix}
        Q_\epsilon \boldsymbol{\omega} \\ \boldsymbol{\delta}
    \end{pmatrix}~.
\end{equation}
Thus we have a general characterisation for these $A$-twisted vertex algebras as $N$-fold intersections of kernels of the screening charges.

In the examples to follow, we will not endeavour to systematically determine the joint kernel of these screening charges for our free field realisations. Rather, as in the main example of $T[\mathrm{SU}(2)]$, we will identify a standard set of gauge/BRST-invariant operators which can be bosonised and rewritten (up to the addition or BRST exact terms) as elements of $V_{\rm FRR}$). Conjecturally, these will generate (perhaps even strongly generate) the corresponding $A$-twisted VOAs.

Our standard set of BRST-invariant operators (written in gauge theory variables), which extrapolates in the obvious manner the list in \eqref{eq:tsu2-closed-ops} for $T[\mathrm{SU}(2)]$, is as follows:
\begin{itemize}
     \item Affine currents in correspondence with the $T_H$ moment maps $\mu_{H,\mathbb{C}}^a$ for the flavour torus acting on hypermultiplets only,
    \begin{equation}\label{eq:closed-curr-1}
        K_a = \sum_{i=1}^N \wt{Q}_{a}^i X_i Y^i~,\qquad a=1,\ldots,N-k~.
    \end{equation}
    \item Fermionic currents $X_i \xi^i$ and $Y^i\chi_i$ for $i=1,\ldots,N$.
    \item Affine currents corresponding to the fermionic versions of the flavour symmetry moment maps,
    \begin{equation}\label{eq:closed-curr-2}
        J_a = \sum_{i=1}^N \wt{Q}_{a}^i \chi_i \xi^i~,\qquad a=1,\ldots,N-k~.
    \end{equation}
    \item Chiral analogues of the monomial generators $W^A$ of the Higgs branch chiral ring. These can always be written unambiguously (no ordering ambiguities arise) as the same monomials as the Higgs branch operators, with $X$'s and $Y$'s reinterpreted as symplectic bosons.
    \item Cousins of the $W^A$ obtained by performing all possible collections of substitutions of the form $X_i \leftrightarrow \chi_i$ or $Y^i \leftrightarrow \xi^i$.
\end{itemize}
We note that the above list is overcomplete as a potentially minimal set of strong generators. In particular, the affine currents $J_a$ as well as the fermionic cousins of the $W^A$ described in the last bullet point appear in the (iterated) singular OPEs of the fermionic currents with the $K_a$ and the $W^A$, respectively.

It is easy to see that the generators we used in the $T[\mathrm{SU}(2)]$ case correspond indeed to the operator listed above specialised to the example. We do not have a proof that this list forms a set of (weak) generators; if this is not the case, then the above operators will generate a sub-VOA and other operators would need to be included. It would be interesting to address this question in general using the screening charge characterisation of these VOAs.

(Unique) representatives in $V_{\rm FFR}$ can be obtained for all of these operators in the same algorithmic manner as in the previous section. In particular, the currents in~\eqref{eq:closed-curr-1}~and~\eqref{eq:closed-curr-2} take the very simple forms
\begin{equation}
    \bar{J}_a = \sum_{i=1}^N \wt{Q}_{\epsilon,a}^i \psi^i \wt{\psi}_i~,\qquad \bar{K}_a = -J_{\phi_a}~.
\end{equation}
The BRST representatives of the other operators can be obtained by performing the now-familiar explicit steps illustrated in Tables~\ref{tab:TSU2-free-field}~and~\ref{tab:TSU2-free-field-2}. We will see many examples of this below in Section~\ref{sec:voa_examples}.

\subsection{\label{subsec:outer_out}Outer automorphisms from free fermionic symmetries}

We now turn to the question of realising outer automorphism symmetries in these general free field realisations. There are $k$ interesting affine currents, which we can build from the gauge theory fermions as follows,
\begin{equation}
    \sum_{j=1}^N Q_i^j\chi_j \xi^j~,\qquad i=1,\ldots,k~.
\end{equation}
These are \emph{not BRST closed}; one needs to add the analogous currents with symplectic bosons replacing the free fermions to realise the BRST closed (and exact) gauge currents. However, the \emph{zero modes} of these operators commute with the BRST operator when restricted to the relative subcomplex of the Feigin standard complex (where the $c$-ghost zero modes are absent).\footnote{We could equivalently consider the affine currents built only from the symplectic bosons with the same charge matrices; the two are related (up to a sign) by the addition of the BRST-exact $\mathcal{J}_a$. We will find it useful to think of these in fermionic terms.} These zero modes therefore generate a $\mathfrak{u}(1)^k$ outer automorphism symmetry of the $A$-twisted vertex algebra---it is natural to interpret this as an avatar of the $\mathrm{U}(1)_C^k$ UV topological symmetry of the Coulomb branch. 

It is expected that this outer automorphism symmetry will be enlarged when there is an IR enhancement of the Coulomb branch symmetry. We saw in our main example that this enhancement could be made manifest at the level of the free field realisation, but only for certain choices of bosonisation vector $\epsilon$.

In the free field realisations described above, these additional currents can be written in terms of the modified (free field) fermions and the component of $\boldsymbol{\tilde\eta}$ (\cf\ Appendix \ref{app:cob}),
\begin{equation}
    \sum_{j=1}^N Q_{\epsilon,i}^j\wt{\psi}_j\psi^j-Q_{\epsilon,i}^a(A_\epsilon^{-1})_a^{\phantom{a}b}J_{\tilde{\eta}_b}~,\qquad i=1,\ldots,k~.
\end{equation}
In the zero mode of this current, the $\tilde\eta$ terms only act to shift an operator by BRST exact ($\eta$-dependent) terms. Restricting to representatives in $V_{\rm FFR}$, we have an action of the zero modes of the operators
\begin{equation}
    \sum_{j=1}^N Q_{\epsilon,i}^j\wt{\psi}_j\psi^j~.
\end{equation}
Though the currents themselves are not annihilated by the screenings \eqref{eq:general_screenings}, their zero modes are diagonalised by the screening currents. Consequently they act within the kernel of the screening charges so give rise to outer automorphisms of the kernel vertex algebra.

\medskip

It is then relatively clear what should be the natural mechanism leading to outer automorphism symmetry enhancement, generalising the observed structure of Section \ref{subsubsec:outer_aut_tsu2}. From the complete set of free fermions in $V_{\rm FFR}$ we can realise a standard construction of an $V_1(\mathfrak{so}_{2N})$ current algebra, whose Cartan subalgebra is generated by both the $T_C$ outer automorphism currents above and the moment map currents $K_a$. Then complementary to the subset of  $V_1(\mathfrak{so}_{2N})$ currents that lie in the kernel of the screening charges enhancing the inner affine symmetry of the vertex algebra, there may be further currents whose zero modes act \emph{preserve the $\mathbb{C}$-span of the screening charges}. These will act as outer automorphisms in the kernel vertex algebra.

The physical expectation is that the Coulomb branch symmetry will be enhanced when there is a \emph{balanced} subgroup of the gauge group, $\mathrm{U}(1)_S\subset \mathrm{U}(1)^k$ where there are exactly two hypermultiplets charged under $\mathrm{U}(1)_S$ (each with charge $\pm1$), as this is the setting where elementary monopole operators act as additional moment map operators in the Coulomb branch chiral ring.

It turns out that indeed, there is an outer automorphism enhancement by zero modes of additional generators of the fermionic $V_1(\mathfrak{so}_{2N})$ precisely when there is such a balanced subgroup \emph{and the corresponding charged hypermultiplets are bosonised appropriately}. We have the following.
\begin{lemma}
    For a given bosonisation $\epsilon$, we will call a subgroup $\mathrm{U}(1)_S\subset \mathrm{U}(1)^k$ \emph{sign balanced} if exactly two hypermultiplets are charged under it, and they are assigned charges $+1$ and $-1$ after the sign-flips dictated by $\epsilon$. Then a $\mathfrak{u}(1)\subset \mathfrak{u}(1)^k$ outer automorphism will enjoy an $\slf_2$ enhancement generated by zero modes of free fermionic currents if and only if the corresponding $\mathrm{U}(1)$ subgroup of the gauge group is sign balanced.
\end{lemma}

\begin{proof}   
    The key to proving this is to rewrite the screening charges in a more convenient form. Recall the expression from \eqref{eq:free_field_screenings}),
    \begin{equation}\label{eq:screening-form-0}
        \tilde{\mathfrak{s}}_i = e^{\bar\rho_i}~,\qquad \boldsymbol{\bar{\rho}} = A_\epsilon^{-1}\begin{pmatrix}
        Q_\epsilon \boldsymbol{\omega} \\
       -\boldsymbol{\delta}
    \end{pmatrix}~.
    \end{equation}
    In light of the definition \eqref{eq:F-A-def} of $A_\epsilon$, this can be rewritten as
    \begin{equation}\label{eq:screening-form}
        \boldsymbol{\bar{\rho}} = A_\epsilon^{-1}
        \begin{pmatrix}
            \mathbf{0} \\ -\boldsymbol\delta + \widetilde{Q}_\epsilon \boldsymbol\omega
        \end{pmatrix}+\boldsymbol\omega~.
    \end{equation}
    Now in order for any two screening currents---say $\tilde{\mathfrak{s}}_1$ and $\tilde{\mathfrak{s}}_2$---to be related by a free fermionic symmetry of the above type, they must have the same dependence on the components of $\boldsymbol{\delta}$, so in particular we would require 
    \begin{equation}\label{eq:screening-row-cond}
        \left(A_\epsilon^{-1}\right)_1^j = \left(A_\epsilon^{-1}\right)_2^j~,\quad j>k~. 
    \end{equation}
    Thus the only possibility is that the outer automorphism in question is generated by the zero modes of
    \begin{equation}
        \psi^1\wt{\psi}_2 = e^{-\omega_1+\omega_2}~,\qquad \psi^2\wt{\psi}_1 = e^{-\omega_2+\omega_1}~,\qquad J_{\omega_1}-J_{\omega_2}~.
    \end{equation}
    The two screening currents in question will indeed transform in the standard representation of this $\slf_2$, and the remaining screening currents will be invariant, if and only if the dependence on $\omega_1$ and $\omega_2$ in $\wt{Q}_\epsilon\boldsymbol\omega$ is entirely through the combination $\omega_1+\omega_2$. This in particular means that the vector $(1,-1,0,\ldots)$ must be orthogonal to the rows of $\wt{Q}_\epsilon$, and must therefore be in the span of the rows of $Q_\epsilon$.

    To complete the proof, we need only observe that if $(1,-1,0,\ldots)$ \emph{is} in the gauge charge lattice (so is in the span of the rows of $Q_\epsilon$), then the form of $A_\epsilon$ and the basic relation $A_\epsilon A_\epsilon^{-1}=\mathbf{1}_{N\times N}$ implies \eqref{eq:screening-row-cond}.
\end{proof}

\medskip

\paragraph{Remark.} The analysis here does not preclude the existence of additional outer automorphisms of $A$-twisted vertex algebras that do not arise in the manner described here for \emph{any} choice of bosonisation/free field realisation. As an example, the $A$-twisted vertex algebra for SQED[$1$] is just the symplectic fermion vertex algebra \cite{Gaiotto:2016wcv,Costello:2018fnz}, which enjoys an $\slf_2\cong\mathfrak{sp}_1$ outer automorphism symmetry. However, the free field realisation for this theory produced by our methods here is just the usual one, in which the symplectic fermion is realised as the kernel of a single screening charge in the free fermion VOA. In this case the outer automorphism does not act in the free field space at all.

Physically, this scenario amounts to the presence of an under-balanced $\mathrm{U}(1)$, giving an \emph{ugly} theory in the sense of \cite{Gaiotto:2008ak}. So in this case, the enhanced Coulomb branch symmetry is associated to emergent free fields. Perhaps unsurprisingly, this seems to be a case that requires special treatment

\section{\label{sec:fin-geom}Finite-dimensional considerations}

In this section we offer a brief investigation into the analogue of our free field constructions at the level of the finite-dimensional geometry of the Higgs branch, generalising the remarks presented in the $T[\mathrm{SU}(2)]$ case in Section~\ref{subsec:TSU2-geom}. The guiding principle is that a given bosonisation prescription corresponds to a localisation on a particular Zariski open subset of the associated variety, which we expect to be isomorphic to the Higgs branch of the respective theory. In VOAs arising in connection with four-dimensional SCFTs, these open subsets have been found to often take the form $T^\ast\left(\mathbb{C}^\times\right)^m \times T^\ast\left(\mathbb{C}\right)^n$, endowed with a standard symplectic structure. The finite-dimensional procedure in the present examples will lead to just such a result.

This section is structured as follows. First, we review how a given bosonisation prescription corresponds to restriction to a subspace $T^\ast\left(\mathbb{C}^\times\right)^N\subset T^\ast\mathbb{C}^N$ of the free hypermultiplet Higgs branch prior to performing symplectic reduction, which then results in the cotangent bundle $T^\ast\left(\mathbb{C}^\times\right)^{N-k}$ (we recall here that $N$ and $k$ are the number of hypermultiplets and the rank of the gauge group, respectively). Second, working in terms of hyperplane arrangements, we show that for each bosonisation, for some choice of resolution parameter $\zeta \in \mathfrak{t}_G^*$, the ``reduced torus'' $T^\ast\left(\mathbb{C}^\times\right)^{N-k}$ can be identified with an open subset of the resolved Higgs branch $\cM_{H,\zeta}$. Finally, we show that for particular choices, the reduced torus enjoys an open immersion into the (unresolved) singular affine variety $\cM_H$ itself. From a mirror perspective, we relate these embeddings to those arising from restriction functors for the (mirror) Coulomb branch chiral rings discussed in~\cite{Kamnitzer:2022zkv}. (This perspective has the additional benefit of predicting when some of the $T^\ast\mathbb{C}^\times$ components can be extended across the origin, providing open subsets of the form $T^\ast\left(\mathbb{C}^\times\right)^m \times T^\ast\left(\mathbb{C}\right)^n$. This may suggest a potential for \emph{de-bosonisation} of some free field realisations in terms of symplectic bosons, \cf\ \ref{subsubsec:sln_slower_free}.)

Throughout this section, all groups $T_G$, $T_H$, $T_C$, $T_V$ will be complexified unless otherwise stated and we omit the subscript $\mathbb{C}$. The corresponding complex Lie algebras will similarly be denoted by $\mathfrak{t}_G$, $\mathfrak{t}_H$, $\mathfrak{t}_C$, $\mathfrak{t}_V$. 

\subsection{\label{sec:sympl-red}Hypertoric symplectic reduction}

As in Section~\ref{subsec:TSU2-geom}, to bosonisation of the symplectic bosons $(X_i,Y^i)$ we associate a finite-dimensional (non-chiral) analogue. Namely, considering the corresponding functions $(Y^i,X_i)$ on $T^\ast\mathbb{C}$, we can choose to identify two different subspaces $T^\ast\mathbb{C}^\times\subset T^\ast\mathbb{C}$ by inverting either $X_i$ or $Y^i$. Given $N$ symplectic boson pairs, a choice of bosonisation is encoded in a sign vector $\epsilon$, and we identify this with a choice of subset $T^\ast\left(\mathbb{C}^\times\right)^N \subset T^\ast V \cong T^\ast\left(\mathbb{C}\right)^N$, where either $X_i$ or $Y^i$ is inverted according to the sign $\epsilon_i$ ($X_i$ if $+$, $Y^i$ if $-$). As a matter of convenience, throughout this section we will utilise an equivalent way of encoding these choices in terms of subsets $U\subset \{1,\ldots , N\}$, indicating that we invert $Y^i$ if $i \in U$ and $X_i$ if $i \not\in U$. Thus, for a given choice of bosonisation $\epsilon$
\begin{equation}
    U \coloneqq \big\{ i \in \{1,\ldots , N \}~|~\epsilon_i=-\big\}~.
\end{equation}
We denote the corresponding subset $T^\ast\left( \mathbb{C}^\times\right)^N\subset T^\ast\mathbb{C}^N$ by $T^\ast V_U$. This is a holomorphic symplectic variety with symplectic form pulled back from $T^\ast V$. It carries an action of $T_G$ and we can consider the corresponding holomorphic symplectic quotient.

\begin{lemma} 
    The symplectic reduction of $T^\ast V_U$ by $T_G$ is isomorphic as a holomorphic symplectic variety to $T^\ast T_H \cong T^\ast \left(\mathbb{C}^\times\right)^{N-k}$.
\end{lemma}

\begin{proof} 
    We proceed by defining appropriate generators for functions on the torus $T^\ast V_U$. In particular, we will see that
    \begin{equation}\label{eq:prod-tori}
        \mathbb{C}\left[T^\ast V_U \right] \cong \mathbb{C}\left[ T^\ast T_G \right] \otimes  \mathbb{C}\left[ T^\ast T_{H} \right]~,
    \end{equation}
    and the result follows immediately. 

    We include amongst our generators $\mu_{\mathbb{C}}^a$, $a=1,\ldots,k$, and $\mu_{H,\mathbb{C}}^b$, $b=1,\ldots N-k$. Then for a $\mathbb{Z}^N$-vector $w$ with components $w_i$, we define 
    \begin{equation}\label{mon-def}
        W_{U}^w \coloneqq \prod_{i \in U} {\left( X_i \right)}^{w_{i}} \prod_{i \not\in U} {Y^i}^{-w_{i}}~.
    \end{equation}
    For any $w$, $W_U^w$ is invertible by construction. We define the additional functions
    \begin{equation}
        W^a_{U,H} \coloneqq W^{\wt{Q}_a}_U~,\qquad W^b_{U} \coloneqq W^{\wt{q}_b}_U~,
    \end{equation}
    for $a=1,\ldots,k$, $b=1,\ldots,N-k$. Then one has for Poisson brackets,\footnote{Here we use the following elementary relations amongst (mirror) charge vectors,
    \begin{equation}
    \begin{alignedat}{3}
        &\wt{Q}Q^T &&= {\bf 0}_{k\times (N-k)}~,\qquad\qquad &&\wt{Q}q^T = \mathbf{1}_{(N-k)\times (N-k)}~,\\
        &\wt{q}q^T &&= {\bf 0}_{k\times N-k}~,\qquad &&\wt{q}Q^T = \mathbf{1}_{k\times k}~.
    \end{alignedat}
    \end{equation}
    }
    \begin{equation}
        \{\mu_\mathbb{C}^a , {W_U}^a \} = {W_U}^a~,\qquad \{\mu_{H,\mathbb{C}}^a,W_{U,H}^{b}\} = W_{U,H}^b~,
    \end{equation}
    with all other brackets vanishing identically. Thus we have the following identifications of Poisson algebras,
    \begin{equation}
    \begin{split}
        &\mathbb{C}\left[\mu_\mathbb{C}^1,\ldots , \mu_{\mathbb{C}}^k , W_{U}^1 , \ldots W_{U}^k, \left( W_{U}^1\right)^{-1} , \ldots , \left( W_{U}^k \right)^{-1} \right] \cong \mathbb{C}\left[ T^\ast T_G \right]~,\\
        &\mathbb{C}\left[\mu_{H,\mathbb{C}}^1,\ldots , \mu_{H,\mathbb{C}}^{N-k} , W_{U,H}^1,\ldots,W_{U,H}^{N-k},\left(W_{U,H}^1 \right)^{-1}, \ldots , \left(W_{U,H}^{N-k}\right)^{-1}  \right] \cong \mathbb{C}\left[ T^\ast T_{H} \right]~.
    \end{split}
    \end{equation}
    Then \eqref{eq:prod-tori} follows immediately upon noting that:
    \begin{itemize}
        \item As $\wt{Q}$ is unimodular by assumption, a monomial built from the $W_U^a$ and $W_{U,H}^b$ will recover $Y^i$, $(Y^i)^{-1}$ for $i\in U$ and $X_i$, $X_i^{-1}$ for $i\not\in U$;
        \item Every monomial $Y^iX_i$ can be written as a linear combination of the $\mu^a_{\mathbb{C}}$ and $\mu_{H,\mathbb{C}}^a$.
    \end{itemize}
\end{proof}

\subsection{\label{sec:open-imm}Embedding localised reductions}

We have seen that for each choice of localisation before symplectic reduction, the resulting space can be identified with $T^\ast T_H$. Depending on the localisation, however, these reductions will come with different maps into $\cM_{H}$ and its resolutions. In this section we will show that in general, the localised reductions can be understood as open subsets of the resolutions $\cM_{H,\zeta}$ for appropriate choices of resolution $\zeta$. We will also see that for certain choices of localisation, the localised reduction is an open subset of (the smooth locus) of the singular Higgs branch $\cM_H$.

In the first instance, we will demonstrate that for each $U$ there exists a resolution parameter $\zeta$ such that the symplectic reduction $T^\ast T_H$ of $T^\ast T_U$ enjoys an open immersion into $\cM_{H,\zeta}$.

Let us fix a resolution parameter $\zeta$ so that $\cM_{H,\zeta}$ is fully resolved, and pick as in~\eqref{eq:hyperplanes} a lift via $\wt{q}$ that in an abuse of notation, we will also denote this lift by $\zeta$. The real and complex moment maps for the flavour action $T_H$ on $\cM_{H,\zeta}$ are given by
\begin{equation}
\begin{alignedat}{3}\label{eq:higg-mon-1}
    &\mu_{H,\mathbb{C}}^a &&= \sum_{i=1}^N q^a_i Y^i X_i~,\qquad &&a \in \{1, \ldots , N-k\}~,  \\
    &\mu_{H,\mathbb{R}}^a &&= \sum_{i=1}^N \frac{q^a_i}{2} \left(|X_i|^2-|Y^i|^2-2\zeta^i\right)~,\qquad &&a \in \{1, \ldots,N-k\}~.
\end{alignedat}
\end{equation}
These lead to a fibration,
\begin{equation}\label{eq:higgs-fib}
   (\pi_{\mathbb{C}},\pi_{\mathbb{R}}) : \cM_{H,\zeta} \rightarrow \mathfrak{t}^*_{H,\mathbb{C}}\oplus \mathfrak{t}^*_{H,\mathbb{R}}~,
\end{equation}
with generic fibres $\left(S^1\right)^{N-k}$ that degenerate along hyperplanes where $Y^i=X_i=0$~(\cf\ \cite{bielawski2000geometry}). If we impose the constraints $Y^i\neq 0$ or $X_i \neq 0$ for every $i$ before symplectic reduction, then we will have removed all of the degeneration loci of this fibration (and more). What remains will be open subsets that (if non-empty) are isomorphic to $\left( S^1 \times \mathbb{R}^3\right)^{N-k} \cong T^\ast \left(\mathbb{C}^\times\right)^{N-k}$. Whether these subsets are actually non-empty or not can depend on the value of $\zeta$.

We can treat this in greater detail using a description of hypertoric varieties in terms of cotangent bundles to toric varieties, glued together in a combinatorial fashion. The relation to toric geometry arises by first restricting to the zero set of the complex moment map, also known as the extended core
\begin{definition}
    The \emph{extended core} $\mathcal{E}$ is the subset $\pi_{\mathbb{C}}^{-1}\left(0\right) \subset \cM_{H,\zeta}$.
\end{definition}
Equivalently, $\mathcal{E}$ consists of those points $p\in\cM_{H,\zeta}$ such that for each $i$, $X_iY^i =0$ at $p$. Now let $A\subset \{1,\ldots N\}$ and define
\begin{equation}
    E_A \coloneqq \{(Y^1,X_1,\ldots ,Y^N, X_N) \in T^\ast V \ | \ X_i = 0 \text{ if } i \in A \text{ and } Y^i=0 \text{ if } i \not\in A \} \
\end{equation}
(note that here we are requiring $Y^i$ or $X_i$ to be zero, rather than non-vanishing). Then we have that
\begin{equation}
    X_A \coloneqq E_A  /\!\!/_{\zeta} T_G~,
\end{equation}
are $(N-k)$-dimensional K\"ahler, Lagrangian subvarieties of $\cM_{H,\zeta}= \cM(\mathcal{A}_\zeta)$. For example, if $A=\emptyset$, then
\begin{equation}
     X_{\emptyset} = \mathbb{C}^N/\!\!/_{\zeta} T_G~,
\end{equation}
where we have set to zero the cotangent fibres $Y^i$. The collection of these K\"ahler Lagrangian subvarieties cover $\mathcal{E}$. Furthermore, their normal bundle in $M_{H_\zeta}$ can be identified with their cotangent bundle (as the subvarieties are Lagrangian), and the collection of $T^\ast X_A$ covers the hypertoric variety $\cM_{H,\zeta}$ itself. Additionally, the $X_A$ themselves are actually toric varieties, and their gluing can is described in a completely combinatorial fashion.

Key to these facts is the identification of the codomain of the fibration~\eqref{eq:higgs-fib}, restricted to $\mathfrak{t}^*_{H,\mathbb{R}}$, with the ambient space of the hyperplane arrangement $\mathcal{A}_\zeta$. Upon this identification, one has that
\begin{equation}
     \pi^{-1}_{\mathbb{R}}(H_i) = \{p \in \cM_{H,\zeta} \ | \ |X_i|^2 - |Y^i|^2 = 0\ \}~.
\end{equation}
This implies that $X_A$ can be identified with the pre-image of the polyhedron
\begin{equation}\label{eq:polyhedra}
    P_A \coloneqq \{x\in \mathfrak{t}_{H,\mathbb{R}}^* \ | \ x \cdot \wt{Q}_i + \zeta^i \leqslant 0 \text{ if } i\not\in A \ \text{ and } \  x \cdot \wt{Q}_i + \zeta^i \geqslant 0 \text{ if } i \in A  \}~.
\end{equation}
It follows that the K\"ahler variety $X_A$ is a toric variety classified by the respective polyhedron $P_A$ (as essentially follows from the fibration structure). Some of the polyhedra $P_A$ could be empty, but their union covers $\mathfrak{t}_{H,\mathbb{R}}^*$. Moreover, the gluing of the K\"ahler varieties $X_A$ is dictated by the intersections of the respective polyhedra. The above observations are summarised in the following:

\begin{lemma}[Remark 2.1.6 of~\cite{proudfoot2008survey}]\label{lem:hypertoric}
    Let $\mathcal{A}_\zeta$ be a hyperplane arrangement determining a hypertoric variety $\cM(\mathcal{A}_\zeta)$. Let $A\subset \{1,\ldots N\}$ and $P_A$ be the polyhedra defined in~\eqref{eq:polyhedra}. Then $\cM(\mathcal{A}_\zeta)$ can be described in terms of a union of cotangent bundles $T^\ast X_A$ to toric varieties $X_A$ classified by $P_A$, glued together equivariantly and symplectically based on their intersections in $\mathfrak{t}^*_{H,\mathbb{R}}$. 
\end{lemma}

This description of hypertoric varieties is a useful starting point for our analysis of the localisations determined by $U$. Suppose that we remove the coordinate hyperplanes $X_i=0$ or $Y^i=0$. If we impose $Y^i \neq 0$, on the extended core we necessarily have $X_i=0$, and \emph{vice-versa}. This means that by imposing a choice of localisation $U$, in $\mathfrak{t}_{H,\mathbb{R}}^*$ we are left with the interior a single polyhedron,
\begin{equation}\label{eq:polyhedra_2}
    O_U \coloneqq \{x\in \mathfrak{t}_{H,\mathbb{R}}^* \ | \ x \cdot \wt{Q}^i + \zeta^i > 0 \text{ if } i\not\in U \ \text{ and } \  x \cdot \wt{Q}^i + \zeta^i < 0 \text{ if } i \in U  \}~.
\end{equation}
Now for a given $U$, it remains to be shown that there exists a choice of $\zeta$ so that $O_U$ is the toric diagram for $\left(\mathbb{C}^\times\right)^{N-k}$. In other words, we want to prove the following:
\begin{lemma}\label{lem:FI-existence}
    For each $U$, there exists a (not necessarily unique) $\zeta$ such that $O_U$ is non-empty and classifies a toric variety $T_{U,\zeta} \cong \left(\mathbb{C}^\times\right)^{N-k} \cong T_H$.
\end{lemma}
Notice that it then follows that there is an open embedding $T^\ast T_{U,\zeta} \hookrightarrow \cM_{H,\zeta}$.
\begin{proof}
    Since $O_U$ has no faces, it suffices to prove that there is a lift of $\zeta = (\zeta^1 , \dots , \zeta^k )$ such that $O_U$ is non-empty. The existence of such a $\zeta$ can for instance be argued in terms of bounded linear functions on Gale dual arrangements~\cite{braden2010gale}.
    
    Alternatively, chasing through the definitions, this is equivalent to showing that for a fixed $\epsilon$ there are at least a $\zeta = (\zeta^1 , \dots , \zeta^k )$ and a point $(x^1,\dots ,x^{N-k})\in \mathfrak{t}_{H,\mathbb{R}}^*\cong \mathbb{R}^{N-k}$ such that the $N$ inequalities
    \begin{equation}
        \begin{pmatrix} \tilde{q}_\epsilon^T & \tilde{Q}_\epsilon^T  
        \end{pmatrix}
    \begin{pmatrix}
    \zeta^1 , &
    \dots ,
    \zeta^k , &
    x^1 ,&
    \dots  ,
    x^{N-k}
    \end{pmatrix}^T >0~
    \end{equation}
    are satisfied. But since $ \left( \tilde{q}_\epsilon^T  \ \tilde{Q}_\epsilon^T \right)$ has full rank by assumption, this is guaranteed.
\end{proof}

\subsection{\label{subsec:aff-imm}Embeddings into singular varieties}

We next want to ask whether a resolution is at all necessary. In the language of the previous section, this is equivalent to asking whether setting some of the components of $\zeta $ to zero (thus taking some hyperplanes in $\mathfrak{t}_{H,\mathbb{R}}^*$ to pass through the origin) ruins the identification of $O_U$ as classifying the toric variety $T_H$.

To put this question into context, it is useful to introduce the notion of the core. Suppose that for a given $U$ we can find a $\zeta$ satisfying the requirements of the previous section. That is, we have made a choice of localisation $U$ and found an open immersion $T^\ast T_H\hookrightarrow\cM_{H,\zeta}$. Now there is always a surjective map,
\begin{equation}\label{eq:higgs-res}
    \cM_{H,\zeta} \twoheadrightarrow \cM_{H}~,
\end{equation}
given by affinisation. This map is in fact an equivariant orbifold resolution. One then defines the \emph{core} as follows,
\begin{definition}
    The core $L(\mathcal{A}_\zeta) \subset \cM_{H,\zeta} \cong \cM(\mathcal{A}_\zeta)$ is the pre-image of the most singular point $0\in\cM_H$ under the surjective map~\eqref{eq:higgs-res}.
\end{definition}
For our purposes, it is best to express the core in terms of the polyhedra $P_A$. In fact, the following simple characterisation holds~\cite{proudfoot2004hyperkahler},
\begin{equation}
    L(\mathcal{A}_\zeta) = \bigcup_{P_{A} \text{ bounded}} X_A~,
\end{equation}
where the bounded polyhedra are by definition the polyhedra of finite volume. Now it is clear that whenever we set some of the components of $\zeta$ to zero, some of the bounded polyhedra will degenerate. The question of whether the open subsets $T^\ast T_{U,\zeta} \cong T^\ast T_H$ remain openly embedded in $\cM_H$ corresponds to the question of whether the respective polyhedra degenerate or not.

This question is relevant for the manifestation (at the level of free-field realisations) of the enhanced outer automorphism group of the $A$-twisted boundary VOAs observed in Lemma~\ref{lem:outer-aut} of Section~\ref{subsec:outer_out}. We noted therein that for a theory with a balanced gauge group (a $U(1)_S$ subgroup of the gauge symmetry $T_G$ whose action on $\mathbb{C}^N \subset V$ has only two non-zero weights $\pm 1$), the enhanced outer automorphism group is manifest only if the $U(1)_S$ subgroup is \emph{sign}-balanced with respect to a choice of localisation $U$ (or equivalently bosonisation $\epsilon$). Translated into the language of this section, we speculated that this might be due to the necessity of resolving $\mathcal{M}_H$ by a resolution parameter $\zeta$ associated to $U(1)_S$ in order of $O_U$ to be non-empty. We now show that this is the case.
\begin{lemma}\label{lem:fin-outer}
    Suppose that in $(T_G,V)$ there is a balanced gauge group $U(1)_S\subset T_G$, so that without loss of generality $Q$ has the first row equal to $(1,1,0,\dots 0)$. Then there exists a resolution parameter $\zeta$ satisfying $\zeta^1=0$ so that $O_U$ is non-empty only if the respective sign-vector $\epsilon$ is of the form $(+-\dots)$ or $(-+\dots)$, and in particular $U(1)_S$ is sign-balanced with respect to $\epsilon$ in the sense of Proposition~\ref{lem:outer-aut}.
\end{lemma}
\begin{proof}
    We now from the proof of Lemma~\ref{lem:FI-existence} that $O_U$ is non-empty if and only if 
    \begin{equation}
        \begin{pmatrix} \tilde{q}_\epsilon^T & \tilde{Q}_\epsilon^T  
        \end{pmatrix}
    \begin{pmatrix}
    \zeta^1 , &
    \dots ,
    \zeta^k , &
    x^1 ,&
    \dots  ,
    x^{N-k}
    \end{pmatrix}^T >0~
    \end{equation}
    for some $\zeta = (\zeta^1 , \dots , \zeta^k )$ and a point $(x^1,\dots ,x^{N-k})\in \mathfrak{t}_{H,\mathbb{R}}^*\cong \mathbb{R}^{N-k}$. Now notice that $ \left( \tilde{q}_\epsilon^T  \ \tilde{Q}_\epsilon^T \right) = \tilde{\mathbf{Q}}^T_\epsilon$, and by~\eqref{eq:gen-mirror-def}
    \begin{equation}
    \mathbf{Q}_\epsilon \tilde{\mathbf{Q}}_\epsilon^T = \mathbf{1}_{N\times N}~.
    \end{equation}
    Thus
    \begin{equation}
    \mathbf{Q}_\epsilon \tilde{\mathbf{Q}}_\epsilon^T   \begin{pmatrix}
    \zeta^1 , &
    \dots ,
    \zeta^k , &
    x^1 ,&
    \dots  ,
    x^{N-k}
    \end{pmatrix}^T  =  \begin{pmatrix}
    \zeta^1 , &
    \dots ,
    \zeta^k , &
    x^1 ,&
    \dots  ,
    x^{N-k}
    \end{pmatrix}^T ~.
    \end{equation}
    From this we can derive $\zeta^1 = \epsilon_1 a + \epsilon_2 b$ for some $a,b>0$, and this can be satisfied for $\zeta^1=0$ only if $\epsilon = (+- \cdots  )$ or $\epsilon = (-+ \dots )$.
\end{proof}

Finally, it is interesting to ask under which conditions $U$ determines a non-degenerate chamber $O_U$ even if \emph{all} components of $\zeta$ are set to zero. Suppose for simplicity that $U=\{1,\ldots N\}$ and consider the hyperplane arrangement $\mathcal{A}_{\zeta=0}$. If there is $x \in \mathbb{R}^{N-k} \cong  \mathfrak{t}^*_{H,\mathbb{R}}$ such that $x \cdot \wt{Q}^i > 0$ for all $i\in \{1,\ldots , N\}$, then by the proof of Lemma~\ref{lem:FI-existence} indeed $O_U$ will be non-empty. For other choices of $U$, one simply has to adjust the signs in the previous inequalities. The goal of our next subsection will be to understand this condition from the perspective of restriction functors on (mirror) Coulomb branch chiral rings, as described in~\cite{Kamnitzer:2022zkv}.

\subsection{\label{sec:mirr-coul-per}The mirror Coulomb branch perspective}

Still in the unresolved case, we can adopt a mirror perspective and make contact with some general results of~\cite{Kamnitzer:2022zkv} related to restriction functors on Coulomb branches. More precisely, we can view $\cM_{H}$ as the Coulomb branch of the mirror theory defined by the charge matrix $\wt{Q}$ (we denote this by $\cM_C[\wt{Q}]$). Then the following lemma, which is a simple consequence of Theorem 2.9 of~\cite{Kamnitzer:2022zkv}, applies.
\begin{lemma}\label{lem:coul-fun}
    Let $\mathbb{C}[\cM_C [\wt{Q}]] $ be the Coulomb branch chiral ring associated to a an Abelian theory with charge matrix $\wt{Q}$. Let $A$ be a co-character (defining a monopole operator $v_A$) that pairs non-negatively with the weights encoded in $\wt{Q}$,
    \begin{equation}
        \langle A , \wt{Q}^i \rangle \geqslant 0~.
    \end{equation}
    Then there is an isomorphism
    \begin{equation}
        \mathbb{C}\left[\cM_C [\wt{Q}]\right][v_A^{-1}] \cong \mathbb{C}\left[\cM_C[\wt{Q}^{(0)}]\right]~,
    \end{equation}
    where $\wt{Q}^{(0)}$ is the representation obtained by restricting to those weights whose pairing with $A$ is zero.
\end{lemma}
In general, there is an explicit homomorphism from $\mathbb{C}[\cM_C [\wt{Q}]]$ to $\mathbb{C}[\cM_C [\wt{Q}^{(0)}]]$~(see \cite{Braverman:2016wma}),
\begin{equation}\label{eq:coulomb-iso}
    v_{A} \mapsto  \prod_{\langle\wt{Q}^i,A\rangle<0} \prod_{j=\langle\wt{Q}^i,A\rangle}^{-1} \varphi^{\wt{Q}^i} w_A~,
\end{equation}
where the operators $w_A$ satisfy the monopole operator algebra of the new theory. The above Lemma from~\cite{Kamnitzer:2022zkv} tells us that under the given assumptions, upon inverting the monopole operator $v_A$ this map becomes an isomorphism.

Now suppose that $\langle A,\wt{Q}^i\rangle >0$ for all $i$. Then the above lemma tells us that there is an open subset of $\cM_C [\wt{Q}] \cong \cM_{H,\zeta}$ that can be identified with the Coulomb branch of a rank-$(N-k)$ pure gauge theory. This Coulomb branch is isomorphic as a holomorphic symplectic variety to $T^\ast \left(\mathbb{C}^\times\right)^{N-k}$, and so we see some alignment with our results from the previous section; indeed, the existence of such a co-character $A$ is equivalent to the statement that the interior of the chamber $O_{\emptyset}$ is non-degenerate.

If for some $i$ $\langle A,\wt{Q}^i\rangle <0$ instead, then we can consider the equivalent theory defined by a charge matrix with $\wt{Q}^i$ replaced by $-\wt{Q}^{i}$. If $\langle A,\wt{Q}^i\rangle =0$, however, then we may not be able to find such an openly embedded $T^\ast (\mathbb{C}^\times)^{N-k}$ in $\cM_C[\wt{Q}]$. Indeed, based on our results in the previous sections, we expect that generically we will be forced to introduce deformation parameters in order to find such patches.

This mirror perspective highlights an interesting additional feature compared to what followed from our previous considerations. Namely, there is a special case in which, for a given co-character $A$, $\wt{Q}^{(0)}$ turns out to define $m\leqslant N-k$ copies of a $\mathrm{U}(1)$ gauge theory a single hypermultiplet of charge one (SQED[$1$]). This theory mirror to a free hypermultiplet, so the Coulomb is just $T^\ast \mathbb{C}$. Thus, by inverting the monopole operator defined by $A$ one finds the ring of functions on $T^\ast \mathbb{C}^m \times T^\ast (\mathbb{C}^\times)^{N-k-m}$. While we will not explore this feature in full generality, we will see some examples of this phenomenon in relation to our free-field construction of the boundary VOA of SQED[$N$].

\subsection{\label{subsec:fin-geom-AN}\texorpdfstring{$A_{N-1}$}{A[N-1]} singularities} 

Consider the rank $k=(N-1)$ theory defined with $N$ hypermultiplets defined by the following charge matrices
\begin{equation}\label{eq:charge-AN}
    \mathbf{Q} =  \left(
    \begin{array}{ccccc}
    -1 & 1 & 0 &  \cdots & 0 \\
     \vdots &  &  & \ddots &  \\
     -1 & 0 & 0 & \cdots & 1 \\
     -1 & 0 & 0 & \cdots & 0 \\
    \end{array}
    \right)~,\qquad  \widetilde{\mathbf{Q}} =\left(
    \begin{array}{ccccc}
     0 & 1 & 0 &  \cdots & 0 \\
     \vdots &  &  & \ddots &  \\
     0 & 0 & 0 & \cdots & 1 \\
     -1 &- 1 & -1 & \cdots & -1 \\
    \end{array}
    \right)~.
\end{equation}
The charge matrix $Q$ encodes the $A_{N-1}$ quiver. The chiral ring is generated by the operators
\begin{equation}
    Z = \prod_{i=1}^N Y^i~, \quad W = \prod_{i=1}^N X_i~,
\end{equation}
as well as 
\begin{equation}
    h = \frac{X_1Y^1}{2}~,
\end{equation}
with
\begin{equation}
    WZ = \frac{(-1)^N}{2^N} h^N~,
\end{equation}
which are the equations determining the $A_{N-1}$ singularity
\begin{equation}
    A_{N-1} \cong \mathbb{C} / \mathbb{Z}_N~.
\end{equation}
The singularity can be resolved by introducing real FI parameters $\zeta\in\mathfrak{t}_G^*\cong \mathbb{R}^{N-1}$.

\subsubsection{\label{subsubbsec:open_patches}Open patches}

For simplicity let us focus on the $A_2$ case; the generalisation to $A_{N-1}$ is relatively straightforward. We will look at our hyperplane arrangement in $\mathfrak{t}^*_{H,\mathbb{R}}\cong \mathbb{R}$, which is identified with the values of the moment map $\mu_{H,\mathbb{R}}$. There are three hyperplanes,
\begin{equation}
\begin{split}
    H_1 : &\ x = 0~, \\
    H_2 : &\ x = \zeta^1~,\\
    H_3 : &\ x = \zeta^2~, 
\end{split}
\end{equation}
where  which in this case are just points whose positions depend on the values of the two FI parameters. Such an arrangement is displayed in Fig. \ref{fig:a2_hyperplanes}.
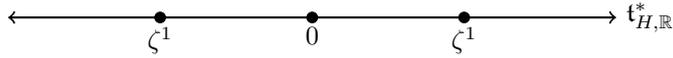
\begin{figure}[t!]
\begin{center}
\begin{tikzpicture}
\draw[<->,black, thick] (-4,0) -- (4,0) node[anchor=west]{$\mathfrak{t}^*_{H,\mathbb{R}}$};
\filldraw[black] (-2,0) circle (2pt) node[anchor=north]{$\zeta^1$};
\filldraw[black] (0,0) circle (2pt) node[anchor=north]{$0$};
\filldraw[black] (2,0) circle (2pt) node[anchor=north]{$\zeta^2$};
\end{tikzpicture}
\end{center}
\caption{\label{fig:a2_hyperplanes}Hyperplane arrangement for the $A_2$ theory}
\end{figure}
with the normal vectors all pointing to the left. Generically (meaning, away from the hyperplanes) above each point we have a copy of $\mathbb{R}^2\times S^1$ in $\cM_{H,\zeta}$. $\mathbb{R}^2$ corresponds to the values of the complex moment map, whereas $S^1$ is the generic fibre of the fibration~\eqref{eq:higgs-fib}. We set the complex moment map to zero to inspect the extended core $\mathcal{E}$ and the toric varieties associated to the various polyhedra (intervals) into which $\mathfrak{t}^*_{H,\mathbb{R}}$ is divided.

Let us fix $\zeta^1 < 0 < \zeta^2$. There are then four non-degenerate polyhedra $P_{A_i}$ associated to the choices $A_1=\{1,2,3\}$, $A_2=\{1,3\}$, $A_3=\{3\}$, $A_4=\emptyset$. Up to rescaling, these polyhedra are recognised to be $\mathbb{R}^{\geqslant 0}$, $[0,1]$, $[0,1]$, $\mathbb{R}^{\geqslant 0}$ respectively, and subsequent polyhedra intersect at a point. Thus, the extended core is build out of K\"ahler varieties $\mathbb{C}$, $\mathbb{P}^1$, $\mathbb{P}^1$, $\mathbb{C}$. glued together at points. In this case, we can easily visualise this by depicting the radius of the $S^1$ fibre, which provides the classical picture of the resolution of the $A_2$ singularity
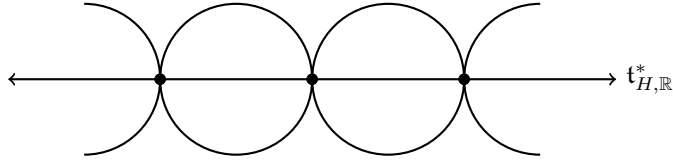
\begin{center}
\begin{tikzpicture}
\draw[color = black,  thick] (-3,-1) arc (-90:90:1);
\draw[black , thick] (3,1) arc (90:270:1);
\draw[color = black, thick] (-1,0) circle (1);
\draw[color = black,  thick] (1,0) circle (1);
\filldraw[black] (-2,0) circle (2pt);
\filldraw[black] (0,0) circle (2pt);
\filldraw[black] (2,0) circle (2pt);
\draw[<->,black, thick] (-4,0) -- (4,0) node[anchor=west]{$\mathfrak{t}^*_{H,\mathbb{R}}$};
\end{tikzpicture}
\end{center}

Let us now consider the localisation choice $U=\emptyset$. This amounts to inverting the function $e\coloneqq W$. In algebraic terms, we are left with a ring of functions
\begin{equation}
    \mathbb{C}[p,e,e^{-1}]~
\end{equation}
where $p= h/2$. This is clearly the ring of functions on a copy of $T^\ast \mathbb{C}^\times$. In fact, we can see this very clearly in the above picture. Recall that picking the inversion choice $U=\emptyset$ corresponds to removing all chambers (that is, polyhedra) but the interior of $P_{\{1,2,3\}}$. If we colour in red the chambers that we lose by inverting $X_i$, then
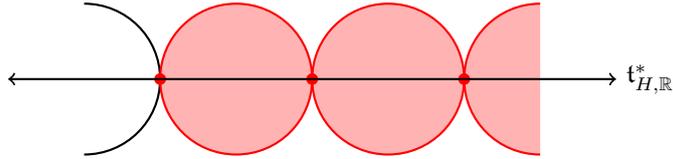
\begin{center}
\begin{tikzpicture}
\filldraw[color = red, fill=red!30 , thick] (3,1) arc (90:270:1);
\draw[black , thick] (-3,-1) arc (-90:90:1);
\filldraw[color = red, fill=red!30 , thick] (-1,0) circle (1);
\filldraw[color = red, fill=red!30 , thick] (1,0) circle (1);
\filldraw[red] (-2,0) circle (2pt);
\filldraw[red] (0,0) circle (2pt);
\filldraw[red] (2,0) circle (2pt);
\draw[<->,black, thick] (-4,0) -- (4,0) node[anchor=west]{$\mathfrak{t}^*_{H,\mathbb{R}}$};
\end{tikzpicture}
\end{center}
so that we are indeed left with a $\mathbb{C}^\times$ in the extended core, or a $T^\ast \mathbb{C}^\times$ in $\cM_{H,\zeta}$. Moreover, notice that setting $\zeta=0$ (that is, collapsing the hyperplanes) the left-over chamber interior would remain untouched. It is an interior of an unbounded chamber, which is not part of the core. In other words, after setting the FI parameters to zero this $T^\ast \mathbb{C}^\times$ embeds openly into the $A_2$ singularity itself.

Let us now choose instead $U=\{1,2\}$. Algebraically, we are inverting the function
\begin{equation}
    e' = (Y^1)^{-1}(Y^2)^{-1}X_3~,
\end{equation}
and we are left with a ring of functions
\begin{equation}
      \mathbb{C}[p, e', e'^{-1}]~.
\end{equation}
To see which $T^\ast \mathbb{C}^\times$ this corresponds to, let us depict in red the chambers lost from inverting $X_3$ and in blue the chamber lost from inverting $Y^1$, $Y^2$.
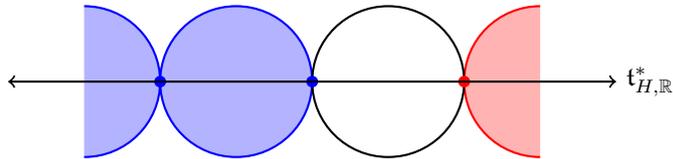
\begin{center}
\begin{tikzpicture}
\filldraw[color = blue, fill=blue!30 , thick] (-3,-1) arc (-90:90:1);
\filldraw[color = red, fill=red!30 , thick] (3,1) arc (90:270:1);
\filldraw[color = blue, fill=blue!30 , thick] (-1,0) circle (1);
\draw[black , thick] (1,0) circle (1);
\filldraw[blue] (-2,0) circle (2pt);
\filldraw[blue] (0,0) circle (2pt);
\filldraw[red] (2,0) circle (2pt);
\draw[<->,black, thick] (-4,0) -- (4,0) node[anchor=west]{$\mathfrak{t}^*_{H,\mathbb{R}}$};
\end{tikzpicture}
\end{center}
We clearly see that are left with one copy of $\mathbb{C}^\times$, but this time it is part of the core. If we set $\zeta=0$, it collapses with the core itself.

Finally, let us consider the inversion choice $U=\{1,3\}$. With this choice of FI parameter we would remove the entire (resolved) affine singularity. However, we know based on our previous abstract arguments that there must be a value of the FI parameter $\zeta^a$ that would make $O_{\{1,3\}}$ isomorphic to $\mathbb{C}^\times$. It can be verified that one such choice is $\zeta^2 < 0 < \zeta^1$.

\subsubsection{Mirror perspective}

Let us now take a mirror perspective and consider the charge matrix
\begin{equation}
    \widetilde{Q} = (-1,-1,\cdots , -1)~.
\end{equation}
If we invert the monopole operator defined by the co-character $A=(-1)$ (corresponding to the operator $W$ in the Higgs branch of the original theory) then clearly $\langle \widetilde{Q}^i , A \rangle = 1 $ for all $i$ and so $\widetilde{Q}^{(0)}$ is empty. Therefore we are left with the Coulomb branch of a $\mathrm{U}(1)$ pure gauge theory, which is a copy of $T^\ast \mathbb{C}^\times$. This is the open patch that we discussed above for the choice $U=\emptyset$.

We can alternatively flip the signs of all the hypermultiplets, and invert the mirror co-character $A=(-1)$ instead to find another openly embedded patch $T^\ast \mathbb{C}^\times$. This choice corresponds to the inversion of the operator $Z$ in the Higgs branch of the original theory.

Notice that no further sign flip is compatible with the condition $\langle \widetilde{Q} , A \rangle \geqslant 0$, which is a necessary assumption in Lemma~\ref{lem:coul-fun}. 

\subsection{\label{subsec:fin-geom-min-nil}Minimal nilpotent orbit closures} 

Now we consider SQED[$N$]. The charge matrix $\mathbf{Q}$ and its mirror charge matrix $\widetilde{\mathbf{Q}}$ read
\begin{equation}\label{eq:charge-SQEDN}
    \mathbf{Q} = \left(
    \begin{array}{ccccc}
     1 & 1 & 1 & \cdots & 1 \\
     0 & 1 & 0 &  \cdots & 0 \\
     \vdots &  &  & \ddots &  \\
     0 & 0 & 0 & \cdots & 1 \\
    \end{array}
    \right)~,\qquad
    \widetilde{\mathbf{Q}} = \left(
    \begin{array}{ccccc}
     \ph{1} & 0 & 0 & \cdots & 0 \\
     -1 & 1 & 0 &  \cdots & 0 \\
     \vdots &  &  & \ddots &  \\
     -1 & 0 & 0 & \cdots & 1 \\
    \end{array}
    \right)~.
\end{equation}
Famously, the Higgs branch of SQED[$N$] is the minimal nilpotent orbit closure in $\slf_N^\ast$. As in the case $N=2$, this can be seen by putting the collection of gauge invariant operators $Y^iX_j$ into a $N\times N$ charge matrix $M$. The complex moment map reads
\begin{equation}
    \mu_{\mathbb{C}} =\sum_{i=1}^N Y^iX_i~.
\end{equation}
Setting this to zero results in the conditions $M^N = \tr (M) = 0$, which define the minimal nilpotent orbit closure. If we introduce a non-zero real FI parameter $\zeta \in \mathbb{R}$, then it is easy to see that 
\begin{equation}
    \cM_{H,\zeta} \cong T^\ast \mathbb{C}\mathbb{P}^{N-1}~,
\end{equation}
with core $\mathbb{C}\mathbb{P}^{N-1}$. This is the Springer resolution. There are two chambers for the resolutions, $\zeta>0$ and $\zeta < 0$, corresponding to whether we are requiring at least one of the $Y^i$s to be invertible, or one of the $X_i$s.

\subsubsection{Open patches in SQED[\texorpdfstring{$3$}{3}]}

We now discuss the open patches in these minimal nilpotent orbit closures that are realised in our approach for $N=3$. Recall that the mirror matrix reads
\begin{equation}
    \widetilde{Q} = 
    \begin{pmatrix}
        -1 & 1 & 0 \\
        -1 & 0 & 1 
    \end{pmatrix}~.
\end{equation}
Thus, we have three hyperplanes that up to overall translation can be taken to be
\begin{equation}
\begin{split}
H_1 : &\ -x_1-x_2 = 0~, \\
H_2 : &\ x_1 = 0~, \\
H_3 : &\ x_2 = \zeta~.
\end{split}
\end{equation}
with $\zeta>0$. We can represent the hyperplane arrangement as in the Figure below. 

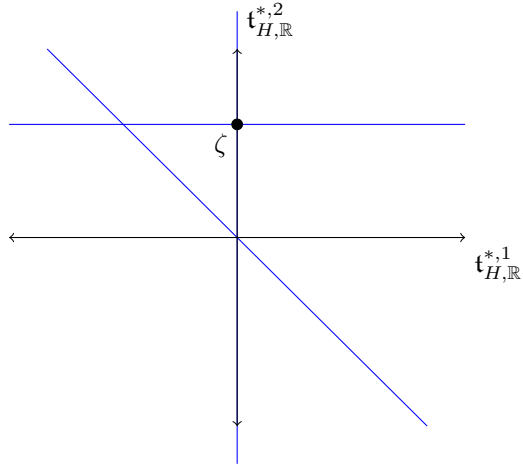
\begin{center}
\begin{tikzpicture}
\draw[draw=blue, thin , fill=gray!50!white, -] (2.5,-2.5)--(-2.5,+2.5) ;
\draw[draw=blue, thin , fill=gray!50!white, -] (0,-3)--(0,3) ;
\draw[draw=blue, thin , fill=gray!50!white, -] (-3,1.5)--(3,1.5) ;
\draw[<->] (-3,0)--(3,0) node[below right]{$ \mathfrak{t}^{*,1}_{H,\mathbb{R}}$};
\draw[<->] (0,-2.5)--(0,2.5) node[above right]{$ \mathfrak{t}^{*,2}_{H,\mathbb{R}}$};
\filldraw[black] (0,1.5) circle (2pt) node[anchor=north east]{$\zeta$};
\end{tikzpicture}
\end{center}

We consider two choices of inversions $U$ that are compatible with $\zeta>0$: $U=\{2,3\}$, which descends to a viable choice on the nilpotent orbit closure, and $U=\{1,2,3\}$, which does not. Following the above convention, we colour in blue the chambers that we lose by inverting $X_i$s, and in red the chambers that we lose by inverting $Y^i$s.

\paragraph{Choice $U=\{1,2,3\}$} With this choice we get a $T^\ast \left(\mathbb{C}^\times\right)^2$ where the that has components in the core (the area that remains white).

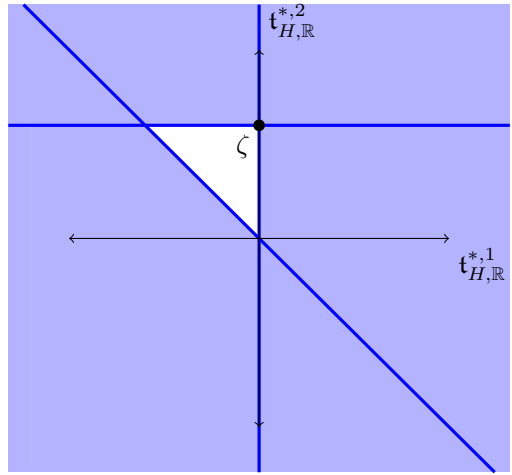
\begin{center}
\begin{tikzpicture}
\fill[blue!30]      (0,-3) rectangle (3.3,3) ;
\fill[blue!30]      (-3,1.5) rectangle (3,3) ;
\fill[blue!30]      (-3.3,-3) rectangle (-3,3) ;
\fill[blue!30]      (-3.3,3) rectangle (3.3,3.1) ;
\fill[blue!30]      (-3.3,-3) rectangle (3.3,-3.1) ;
\fill[blue!30]      (3,-3) -| (-3,3) -- cycle;
\draw[draw=blue, very thick ,  -] (-3.1,3.1)--(3.1,-3.1) ;
\draw[draw=blue, very thick , -] (0,-3.1)--(0,3.1) ;
\draw[draw=blue, very thick , -] (-3.3,1.5)--(3.3,1.5) ;
\draw[<->] (-2.5,0)--(2.5,0) node[below right]{$ \mathfrak{t}^{*,1}_{H,\mathbb{R}}$};
\draw[<->] (0,-2.5)--(0,2.5) node[above right]{$ \mathfrak{t}^{*,2}_{H,\mathbb{R}}$};
\filldraw[black] (0,1.5) circle (2pt) node[anchor=north east]{$\zeta$};
\end{tikzpicture}
\end{center}

\paragraph{Choice $U=\{2,3\}$} In this case, the remaining $T^*(\mathbb{C}^\times)^2$ does not degenerate when $\zeta = 0$. In fact, we are inverting the operators $X_1Y^2$ and $e_2 = X_1Y^3$, which are both in the chiral ring of the minimal nilpotent orbit closure.

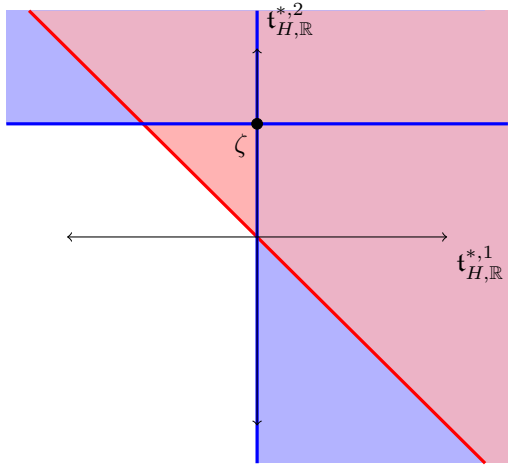
\begin{center}
\begin{tikzpicture}
\fill[blue!30]      (0,-3) rectangle (3.3,3) ;
\fill[blue!30]      (-3.3,1.5) rectangle (3.3,3) ;
\fill[purple!30]      (3,-3) rectangle (3.3,3) ;
\fill[purple!30]      (-3,3) -| (3,-3) -- cycle;
\fill[red!30]      (-1.5,1.5) -| (0,0) -- cycle;
\draw[draw=red, very thick ,  -] (3,-3)--(-3,3) ;
\draw[draw=blue, very thick , -] (0,-3)--(0,3) ;
\draw[draw=blue, very thick , -] (-3.3,1.5)--(3.3,1.5) ;
\draw[<->] (-2.5,0)--(2.5,0) node[below right]{$ \mathfrak{t}^{*,1}_{H,\mathbb{R}}$};
\draw[<->] (0,-2.5)--(0,2.5) node[above right]{$ \mathfrak{t}^{*,2}_{H,\mathbb{R}}$};
\filldraw[black] (0,1.5) circle (2pt) node[anchor=north east]{$\zeta$};
\end{tikzpicture}
\end{center}

\subsubsection{\label{sec:sqed3-mirror-geom}Mirror perspective}

We now take a mirror perspective and discuss open patches that openly embeds in the (unresolved) minimal nilpotent orbit closures for general $N$.

To this end, consider the charge matrix $\widetilde{Q}$. We can flip the signs of the charges of the last $N-1$ hypermultiplets to obtain the matrix
\begin{equation}
    \widetilde{Q}_U = \left(
    \begin{array}{ccccc}
     -1 & -1 & \ph{0} & ~\cdots & \ph{0} \\
     -1 & \ph{0} & -1 & ~\cdots & \ph{0} \\
     \vdots &  \ph{\vdots} &  & ~\ddots &  \\
     -1 & \ph{0} & \ph{0} & ~\cdots & -1 \\
    \end{array}
    \right)~.
\end{equation}

Then we can invert the mirror monopole operator defined by the co-character $(-1,0,\cdots ,0)$ (corresponding to the operator $X_1Y^2$ of the original theory) to obtain 
\begin{equation}
    \widetilde{Q}_U^{(0)} = \left(
   \begin{array}{ccc}
         \ph{0}     & ~~\cdots & \ph{0} \\
         -1         & ~~\cdots & \ph{0} \\
         \ph{0}     & ~~\cdots & \ph{0} \\
        \ph{\vdots} & ~~\ddots & \ph{\vdots} \\
        \ph{0}      & ~~\cdots & -1 \\
    \end{array}
    \right)~.
\end{equation}
This is the matrix defining a $\mathrm{U}(1)$ pure gauge theory and $N-2$ copies of SQED[$2$], we have therefore identified a patch $T^\ast \mathbb{C}^\times \times T^\ast \mathbb{C}^{N-2}\subset \mathcal{M}_C[\widetilde{Q}]\cong \mathcal{M}_H$. 

This patch should in principle be good enough to construct a free field realisation, and in fact we are going to discuss such a construction in Section~\ref{sec:sqed3-example}. However, we could further invert the mirror monopoles associated to the various mirror co-characters $(\ldots ,0,-1,0\ldots )$ to obtain an openly embedded patch $T^\ast (\mathbb{C}^\times)^{N-2}\subset \mathcal{M}_H$. This are precisely the operators that we invert with the $U=\{2,\cdots ,N\}$ choice.

\section{\label{sec:voa_examples}Further examples of free field realisation}

In this final section, we apply our general free field analysis to SQED[$N$] and $A_{N-1}$ quiver gauge theories. The (singular) Higgs branches for these theories are minimal nilpotent orbit closures in $\mathfrak{sl}_N^*$ and $A_{N-1}$ Kleinian singularities, respectively (see Section~\ref{subsec:fin-geom-AN}~and~\ref{subsec:fin-geom-min-nil}).

\subsection{\texorpdfstring{SQED[$N$]}{SQED[N]}}\label{sec:sqed3-example}

Our first set of theories are $\mathrm{U}(1)$ gauge theories with $N$ hypermultiplets of weight $1$. The charge matrix $\mathbf{Q}$ and its mirror charge matrix $\widetilde{\mathbf{Q}}$ were already reported in~\eqref{eq:charge-SQEDN}
\begin{equation}
    \mathbf{Q} = \left(
    \begin{array}{ccccc}
     1 & 1 & 1 & \cdots & 1 \\
     0 & 1 & 0 &  \cdots & 0 \\
     \vdots &  &  & \ddots &  \\
     0 & 0 & 0 & \cdots & 1 \\
    \end{array}
    \right)~,\qquad
    \widetilde{\mathbf{Q}} = \left(
    \begin{array}{ccccc}
     \ph{1} & 0 & 0 & \cdots & 0 \\
     -1 & 1 & 0 &  \cdots & 0 \\
     \vdots &  &  & \ddots &  \\
     -1 & 0 & 0 & \cdots & 1 \\
    \end{array}
    \right)~.
\end{equation}
The $A$-twisted VOA is the BRST reduction of $N$ symplectic boson pairs $(X_i,Y^i)$ and $N$ complex fermion pairs $(\chi_\alpha, \xi^\alpha)$ with respect to the zero modes of the BRST current
\begin{equation}
    J^{\mathrm{BRST}} =  c\left( X_a Y^a + \chi_\alpha \xi^\alpha \right)~,
\end{equation}
where $(b,c)$ is a ghost system of weight $(1,0)$.

The standard list of BRST-invariant operators (\cf\ \ref{subsec:general_embedding}) reduces to the following. For $i\neq j$ and $k\neq 1$, we have BRST-closed bi-linears composed of symplectic bosons,
\begin{equation}\label{eq:slN-1}
   X_iY^j~,\qquad-X_1Y^1 + X_k Y^{k}~,
\end{equation}
which generate a $V_{-1}(\slf_N)$ VOA. Similarly (for the same conditions $i\neq j$, $k\neq 1$), we have BRST-closed bi-linears composed of the complex fermions
\begin{equation}
    \xi^i\chi_j~,\qquad-\xi^1\chi_1 + \xi^k\chi_k~,
\end{equation}
which generate a $V_{1}(\slf (N))$ VOA. Finally, we have odd, fermionic current generators,
\begin{equation}\label{eq:slN-odd}
    Y^i \chi_j~,\qquad X_i\xi^j~.
\end{equation}
Together, these generate a quotient of the level one $V^1(\slf(N|N))$, which after taking cohomology with respect to the BRST operator gets reduced to a quotient (we expect the simple quotient) of $V^1(\pslf(N|N)))$.

We can now apply our general procedure to these examples. After briefly describing the change of lattice coordinates in Section~\ref{subsubsec:general_coordinate_charnge}, in Section~\ref{subsubsec:sln_quick_free} we give a short exposition of the free field realisation that arises by applying the methods of Section~\ref{sec:general_abelian}. We also compare our result with the free field realisation of $V_{-\frac32}(\slf(3))$ proposed in~\cite{Beem:2019tfp}, as well as an extensions of that result to $V_{-1}(\mathfrak{sl}_N)$. We comment on the geometric aspects of this comparison in Section~\ref{subsubsec:sln_more_geom}, and see a nice application of our Coulomb branch analysis above. The actual free field constructions are compared in Section~\ref{subsubsec:sln_slower_free}, where inspired by the finite-dimensional geometry we de-bosonise some combinations of our free fields.

\subsubsection{\label{subsubsec:general_coordinate_charnge}Change of coordinates}

We will work with the $\epsilon = (+--\cdots -)$ choice of bosonisation, which generalises the one used in our first analysis of SQED[$2$] in Section~\ref{subsec:TSU2-BRST}. From the point of view of the Higgs branch, this corresponds to localising on a $T^\ast \left(\mathbb{C}^\times\right)^N\subset \cM_H$ subset of the singular Higgs branch, as we described in Section~\ref{subsec:fin-geom-min-nil}. The modified charge matrices now take the form
\begin{equation}
    Q_\epsilon = \left(
    \begin{array}{ccccc}
     1 & -1 & -1 & \cdots & -1 
    \end{array}
    \right)~,\qquad
    \widetilde{Q}_\epsilon = \left(
    \begin{array}{ccccc}
     -1 & -1 & \ph{0} & \cdots & \ph{0} \\
     -1 & \ph{0} & -1 & \cdots & \ph{0} \\
     \vdots &  \ph{\vdots} &  & \ddots &  \\
     -1 & \ph{0} & \ph{0} & \cdots & -1 \\
    \end{array}
    \right)~,
\end{equation}
from which we compute the remaining players in the change of basis matrix \eqref{eq:gen-COB},
\begin{equation}
   A_\epsilon^{-1} =
  \frac{1}{N}\left(
    \begin{array}{ccccc}
   \ph{1} & -1 & -1 & \cdots & -1 \\
     -1 & -N+1 &\ph{1} & \cdots &\ph{1} \\
     \vdots & & \ph{\ddots} & & \ph{\vdots} \\
     -1 &\ph{1} & \ph{1} & -N+1 &\ph{1} \\
     -1 &\ph{1} & \ph{1} &1 & -N+1 \\
    \end{array}
    \right)
    ~,\quad
    F_\epsilon^T =\frac{1}{N} \left(
    \begin{array}{ccccc}
    \ph{1}      & -1            & -1        & \cdots &- 1 \\
    -1          & \ph{1}        & \ph{1}         & \cdots & \ph{1} \\
    \ph{\vdots} & \ph{\vdots} & \vdots    & \ddots & \ph{1} \\
     -1         & \ph{1}         & \ph{1}         & \cdots & \ph{1} \\
    \end{array}
    \right)~.
\end{equation}
For our free fields, we therefore take the following basis,
\begin{equation}
\begin{alignedat}{3}
    &\phi_i     &&= -\sigma_1 - \sigma_{i+1}~,\qquad &&i = 1,\ldots,N-1~, \\
    &\delta_i   &&= \rho_1 + \rho_{i+1}~,\qquad &&i = 1,\ldots,N-1~, \\
    &\omega_i   &&= \gamma_i + \frac{(-1)^{\delta_{1i}}}{N} \left(\sum_{j=1}^N (\sigma_j - \rho_j) - 2(\sigma_{1} - \rho_{1} ) \right)~,\qquad &&i=1,\ldots,N~.
\end{alignedat}
\end{equation}
The two point functions in this basis are not diagonal, but rather we have pairings
\begin{equation}
\begin{aligned}
    ( \delta_i,\delta_j ) &= -( \phi_i,\phi_j ) = 1 + \delta_{ij}~,\\
    ( \omega_i , \omega_j ) &=  \delta_{ij}
    ~.
\end{aligned}
\end{equation}
The lattice extension of this Heisenberg which defines $\wt{V}_{\mathrm{FFR}}$ is given by
\begin{equation}
    \bigoplus_{i=1}^{N-1} \mathbb{Z}\left(\delta_i + \phi_i - \sum_{j=1}^{N-1}\frac{\delta_j+\phi_j}{N}\right) \oplus \bigoplus_{i=1}^{N} \left( \omega_i \right)~.
\end{equation}
As we described below~\eqref{eq:refinement}, free fields with mutually vanishing OPEs could be obtained by a rational change of variables. Though such a refinement could be useful in interpreting the free fields in terms of CDOs on $T^\ast (\mathbb{C}^\times)^{N-1}$, for our purposes it will be simpler to continue with the present basis.\footnote{The relevant lattice refinement would correspond to an orthogonalisation of $\widetilde{Q}_\epsilon$, which reads,
\begin{equation}
        \widetilde{Q}^{\text{orth}}_\epsilon = \left(
    \begin{array}{ccccc}
     -1 & -1 & 0 & \cdots & 0 \\
     -\frac{1}{2} & \frac{1}{2} & 1 & \cdots & 0 \\
     \vdots &  \vdots &  & \ddots &  \\
     -\frac{1}{N-1} & \frac{1}{N-1} & \frac{1}{N-1} & \cdots & 1 \\
    \end{array}\right)~.
\end{equation}
}.

We will need a few notational ingredients to produce our free field realisation. First, we define our free fermions as follows (for later convenience, we are introducing a relative minus sign for $\omega_1$ compared to~\eqref{eq:gen_free_fermions})
\begin{equation}
    J_{\omega_i} = (-1)^{\mathbf{\delta}_{i1}} \psi^i \widetilde{\psi}_{i}~.
\end{equation}
We then introduce the following current, whose zero mode generates the $\mathrm{U}(1)_C$ outer automorphism,
\begin{equation}
    J_\psi \coloneqq \frac{1}{N}\left( \sum_{i= 1}^N \psi^i \widetilde{\psi}_i \right)~.
\end{equation}
The BRST representatives in $V_{\rm FFR}$ of the currents $J_{\rho_i}$ can now be written as\footnote{Here and below, we use the shorthand notation that fields with index $\leqslant 0$ vanish.}
\begin{equation}\label{eq:sqedN_rho_currents}
    \bar{J}_{\rho_i} = (-1)^{\delta_{i1}} \left( -\frac{1}{N} \sum_{j=1}^{N-1} J_{\delta_j} - J_\psi \right) + J_{\delta_{i-1}}~.
\end{equation}

\subsubsection{\label{subsubsec:sln_quick_free} Free field realisation}

Given the ingredients above, we can efficiently produce a free field realisation of the $\pslf(N|N)$ current algebra. Indeed, the computations are directly parallel to the $N=2$ case (\cf\ Tables~\ref{tab:TSU2-free-field}~and~\ref{tab:TSU2-free-field-2}). First, provided $i\neq j$, we have replacements,
\begin{equation}\label{eq:sln1-free-field}
    X_i Y^j ~~\mapsto~~
    \begin{cases}
        \phantom{-\bar{J}_{\rho_1} \bar{J}_{\rho_i}}e^{(\phi_{j-1} + \delta_{j-1})} & i=1~, \\
        \phantom{\bar{J}_{\rho_i}}\!\!-\bar{J}_{\rho_i}e^{-(\phi_{i-1} + \delta_{i-1})+(\phi_{j-1} + \delta_{j-1})} & i\neq 1,~j\neq 1~,\\
        -\bar{J}_{\rho_1} \bar{J}_{\rho_i} e^{-(\phi_{i-1} + \delta_{i-1})} & j=1~.
    \end{cases}
\end{equation}
The currents in~\eqref{eq:slN-1} are then simply identified with the $J_{\phi_i}$. One can check directly that these fields generate indeed a level $-1$ $\slf_N$ current algebra.

The $V_{1}(\slf_N)$ generators are given by the usual construction in terms of the free fermions, with
\begin{equation}
    \xi^i\chi_j ~\mapsto~ -\psi^{i}\widetilde{\psi}_j~,\qquad i\neq j~,
\end{equation}
and a basis for the Cartan subalgebra given by
\begin{equation}
    \psi^{1}\widetilde{\psi}_1 - \psi^{j}\widetilde{\psi}_j~,\qquad j=2,\ldots,N~.
\end{equation}
Finally, the odd generators read
\begin{equation}
\begin{aligned}
    Y^i\chi_j &~~\mapsto~~
    \begin{cases}
       \bar{J}_{\rho_1} e^{-\frac{1}{N}\sum_{l=1}^{N-1}( \phi_l + \delta_l )} \psi^j~, &i=1~, \\
      e^{-\frac{1}{N} \sum_{l=1}^{N-1}( \phi_l + \delta_l )+(\phi_{i-1}+\delta_{i-1})} \psi^j~,\quad &\text{else}~,
    \end{cases}\\
     X_i\xi^j &~~\mapsto~~
    \begin{cases}
       e^{\frac{1}{N} \sum_{l=1}^{N-1}( \phi_l + \delta_l )} \widetilde{\psi}_j &i=1~,\\
       -\bar{J}_{\rho_i}e^{\frac{1}{N} \sum_{l=1}^{N-1}( \phi_l + \delta_l )-(\phi_{i-1}+\delta_{i-1})} \widetilde{\psi}_j &\text{else}~.
    \end{cases}
\end{aligned}
\end{equation}
The full set of $\pslf(N|N)$ OPEs can be verified for all of these operators by direct computation.

As we mentioned above, the zero-mode of the trace current $J_\psi$ (which is not included in $V_{-1}(\slf_N) $) generates a $\mathfrak{u}(1)$ outer automorphism symmetry. This corresponds to the $\mathrm{U}(1)_C$ topological symmetry of SQED[$N$], which for $N>2$ is not enhanced. Indeed, from \eqref{eq:sqedN_rho_currents} one can see that the screening currents $e^{\bar{\rho}}_i$ will not enjoy any non-Abelian automorphism enhancement for $N=2$, as they each depend differently on the $\delta$ bosons. This is as it should be given the lack of a balanced subgroup of the gauge symmetry in this example. From the point of view of the outer automorphism symmetry, then, our choice of bosonisation is in this case was not particularly important.

\subsubsection{\label{subsubsec:sln_more_geom} More on finite-dimensional geometry}

In order to understand a little more concretely the inner workings of our free field realisation, it is instructive to compare our result here with a related construction in four dimensions~\cite{Beem:2019tfp}, where free field realisations of the simple affine Kac-Moody algebra $V_\kappa (\mathfrak{sl}_3)$ at level $\kappa=-\frac{3}{2}$ was produced. This VOA has been proven to have the minimal nilpotent orbit in $\mathfrak{sl}_3^\ast$ as associated variety, and the construction of~\cite{Beem:2019tfp}---much in the spirit of the \emph{ad hoc} free field realisation of $V_{-1}(\pslf(2|2))$ presented in Section~\ref{sec:tsu2}---exploits the geometry of the minimal nilpotent orbit to produce a free field realisation. Furthermore, it was announced in~\cite{Beem:2019tfp} that similar constructions can be performed for $V_{1}(\mathfrak{sl}_N)$, but at the cost of introducing an additional auxiliary $\mathfrak{gl}_1$ Heisenberg current. Intuitively, this is because the associated varieties to $V_{-1}(\mathfrak{sl}_N)$ \emph{are not} the respective minimal nilpotent orbits~\cite{ARAKAWA2017157}, but include some additional direction in $\slf_N^\ast$ related to the extra current.

Per Conjecture~\ref{conj:ass-var}, we expect the associated variety for $V_{1}(\pslf(N|N))$ to be the minimal nilpotent orbit, so we may expect some connection between the geometric construction of~\cite{Beem:2019tfp}---including the extra Heisenberg current---and our free field realisation, but the extra current should be accounted for by extra nilpotent degrees of freedom. We will indeed demonstrate that this is the case.

The geometric aspects of $\slf_N^\ast$ minimal nilpotent orbit closures used in~\cite{Beem:2019tfp} are based on the following decomposition of $\mathfrak{sl}_N$,
\begin{equation}\label{eq:slN-dec}
  \slf_N \cong \left( \slf_{N-2} \oplus \mathbb{C}h_{\perp} \oplus (\slf_2)_\theta \right) \oplus (\mathfrak{R},\mathbf{2})~,
\end{equation}
where $(\slf_2)_\theta$ is an $\mathfrak{sl}_2$-triple associated to a choice of highest root, and $\mathfrak{R}$ transforms in the representation $(\mathbf{N}-\mathbf{2})_+ \oplus \overline{(\mathbf{N}-\mathbf{2})}_-$ of $\slf_{N-2} \oplus \mathbb{C}$ (here $+$ and $-$ correspond to positive and negative weights with respect to the extra Cartan generator $h_\perp$). We will pick a basis $(e_\theta , f_\theta , h_\theta)$ for $(\slf_2)_\theta$ and decompose the $(\mathfrak{R},2)$ generators in terms of their $h_\theta$ weight,
\begin{equation}
    (\mathfrak{R},2) = \mathfrak{R}^+ \oplus \mathfrak{R}^-~.
\end{equation}
We write $e_a$, $\tilde{e}_a$ $a \in \{1,\cdots , N-2\}$ as bases for the two $\slf_{N-2}$ representations in $\mathfrak{R}^+$. A key observation in~\cite{Beem:2019tfp} is that upon inversion of $e_\theta$, the generators of the minimal nilpotent orbit can be expressed in terms of the functions $(e_\theta , e_{\theta}^{-1} ,  h_\theta )$ and $e^a$, $\widetilde{e}_a$ which have non-vanishing Poisson brackets
\begin{equation}\label{eq:sln-poisson}
     \{h_\theta , e_\theta\} = 2 e_\theta~,\quad
     \{e^a , \widetilde{e}_b\} = \delta_{b}^a e_\theta ~,\quad
     \{h_\theta , e^a \} = e^a~\quad
     \{h_\theta , \widetilde{e}_a \} = e_a~.
\end{equation}
Moreover, if one introduces $(\boldsymbol{d}_\theta , \boldsymbol{\beta}_a,  \boldsymbol{\gamma}_a)$ for $a\in \{1,\cdots , N-2\}$ according to,
\begin{equation}\label{eq:sln-beta-def}
    \boldsymbol{\beta}^a \coloneqq e^{-\frac12}_\theta e_a ~,\qquad
    \boldsymbol{\gamma}_a =  e^{-\frac12}_\theta \widetilde{e}_{a} ~,\qquad
    \boldsymbol{d}_\theta = e_\theta^{\frac12}~, 
\end{equation}
then one can realise all of the coordinate functions in terms of $\mathfrak{h}_\theta$, $\boldsymbol{d}_\theta$, $\boldsymbol{\beta}^a$, $\boldsymbol{\gamma}_a$ with non-vanishing Poisson brackets given by
\begin{equation}
    \{h_\theta , \boldsymbol{d}_\theta\} = \boldsymbol{d}_\theta~,\qquad \{\boldsymbol{\beta}^a , \boldsymbol{\gamma}_b\} = \delta^a_b~.
\end{equation}
These are interpreted as coordinate functions on 
\begin{equation}\label{eq:ad-hoc-free}
    T^\ast \mathbb{C}^\times \times T^\ast \mathbb{C}^{N-2}~,
\end{equation}
and so $h_\theta$, $e_\theta$, $e^a$, $\widetilde{e}_a$ can be interpreted as functions on
\begin{equation}
    \left( T^\ast  \mathbb{C}^\times \times T^\ast \mathbb{C}^{N-2} \right) /\mathbb{Z}_2~,
\end{equation}
with $\mathbb{Z}_2$ acting on $\mathbb{C}^\times$ by negation. The free fields introduced in~\cite{Beem:2019tfp} correspond to CDOs on the space~\eqref{eq:ad-hoc-free}.

Let us compare this picture with the open patches we discussed in Section~\ref{subsec:fin-geom-min-nil}, and in particular the restriction functor on Coulomb branches as presented in Section~\ref{sec:sqed3-mirror-geom}. We saw there that inversion of $e= e_\theta$ indeed leads to an openly embedded patch of the form
\begin{equation}
  \left( T^\ast  \mathbb{C}^\times \times T^\ast \mathbb{C}^{N-2} \right)~,
\end{equation}
where $T^\ast \mathbb{C}^{N-2}$ can be identified with $N-2$ copies of the Coulomb branch of SQED[$1$]. The way in which this patch is obtained, however, is slightly different from the one presented in~\cite{Beem:2019tfp}. In fact, instead of the redefinition~\eqref{eq:sln-beta-def}, the isomorphism~\eqref{eq:coulomb-iso} is such that factors of $e_\theta$ are mapped into the new variables in an asymmetric fashion, namely
\begin{equation}
w_+^a \coloneqq e_\theta e_a ~,\qquad w_-^a =   \widetilde{e}_{a}~.
\end{equation}
The relation between the two open patches can be understood in terms of an isomorphism
\begin{equation}
    \left( T^\ast  \mathbb{C}^\times \times T^\ast \mathbb{C}^{N-2} \right) /\mathbb{Z}_2 \cong  \left( T^\ast  \mathbb{C}^\times \times T^\ast \mathbb{C}^{N-2} \right)~,
\end{equation}
where the two sides are subject to different $\mathbb{C}^\times$ scaling. Finally, notice that the open patch $T^\ast  \left(\mathbb{C}^\times\right)^N$ that underlies the free field realisation of Section~\ref{subsubsec:sln_quick_free} can be obtained by further inverting the functions $e_a$.

\subsubsection{\label{subsubsec:sln_slower_free} More on free field realisations and de-bosonisation}

The discussion above raises the question of whether there is a chiral analogue of these relations. More precisely, starting from the free field realisation presented above, perhaps one can \emph{de-bosonise} some auxiliary $\beta \gamma$-systems to obtain the generators proposed in~\cite{Beem:2019tfp} for the $V_{-1}(\mathfrak{sl}_N)$ component of our current algebra, with the extra Heisenberg current provided by the free fermions. In this section, we will show that this is precisely the case.

By an abuse of notation, we use the same letters for the functions we introduce above and their chiral analogues. According to the discussion around~\eqref{eq:sln1-free-field}, we have for the relevant generators in our free field realisation
\begin{equation}\label{eq:sln-field-patch}
\begin{alignedat}{3}
    &h_\theta &&= -J_{\phi_1}~,\qquad e^a &&= e^{\phi_{a+1} + \delta_{a+1}}~,\\
    &e_\theta &&= e^{\phi_1+\delta_1}~,\qquad \wt{e}_a &&= -\bar{J}_{a+2} e^{-(\phi_{a+1} + \delta_{a1+} )+(\phi_1+\delta_1)}~,
\end{alignedat}
\end{equation}
with OPEs in accordance with~\eqref{eq:sln-poisson}
\begin{equation}
\begin{split}
    h_\theta (z) e_\theta (w) &~~\sim~~ \frac{2e_\theta(w)}{z-w}~, \\
    e^a (z) \widetilde{e}_b (w) &~~\sim~~ \frac{\delta_{b}^a e_\theta(w)}{z-w}~,\\
    h_\theta (z) e_a(w) &~~\sim~~ \frac{e_a(w)}{z-w}~.
\end{split}
\end{equation}
If we now define (for $a\in\{1,\cdots,N-2\}$),
\begin{equation}\label{eq:sln-field-patch-2}
\begin{split}
    \boldsymbol{d}_\theta &~~\coloneqq~~ e^{\frac{1}{2}\left(\phi_1+\delta_1\right)}~,\\
    \boldsymbol{\beta}^a &~~\coloneqq~~ e^{(\phi_{a+1} + \delta_{a+1})-\frac{1}{2}(\phi_1+\delta_1)}~,\\
    \boldsymbol{\gamma}_a &~~\coloneqq~~ -\bar{J}_{a+2} e^{-(\phi_{a+1} + \delta_{a+1} )+\frac12 (\phi_1+\delta_1)}~, 
\end{split}
\end{equation}
then these fields enjoy the OPEs of CDOs on $T^\ast \mathbb{C}^\times\times T^\ast \mathbb{C}^{N-2}$~,
\begin{equation}
\begin{split}
    \boldsymbol{\beta}^a (z) \boldsymbol{\gamma}_b (w) &~~\sim~~ \frac{\delta_a^b}{z-w}~,\\
    h_\theta (z) \boldsymbol{d}_\theta (w) &~~\sim~~ \frac{\boldsymbol{d}_\theta}{z-w}~.
\end{split}
\end{equation}
Indeed, we can see that the pairs $(\boldsymbol{\beta}^a,\boldsymbol{\gamma}_a)$ in~\eqref{eq:sln-field-patch-2} are in fact $(N-2)$ bosonised $\beta \gamma$-systems by making the following identifications,
\begin{equation}
    \boldsymbol{\beta}^a \coloneqq e^{\upsilon_a - \tau_a}~,\qquad\boldsymbol{\gamma}_a \coloneqq J_{\upsilon_a} e^{-\upsilon_a + \tau_a}~,
\end{equation}
where we define yet another basis of free bosons,
\begin{equation}
\begin{split}
    \upsilon_a &\coloneqq \delta_{a+1}-\nu ~,\\
    \tau_a &\coloneqq -\phi_{a+1}-\frac12(\phi_1+\delta_1)-\nu~.
\end{split}
\end{equation}
For convenience, we have introduced
\begin{equation}
    \nu \coloneqq \frac{1}{N} \left(\sum_{j=1}^{N-1} \delta_j + \sum_{i=1}^{N}(-1)^{\delta_{1i}}\omega_i \right)~.
\end{equation}

\begin{table}[t!]
    \centering
    \resizebox{\textwidth}{!}{
    \begin{tabular}{c|c|c|c|}
                & Gauge theory &  FFR Representative & De-bosonisation \\ 
        \hline
        $e_\theta$              & $X_1Y^2$               & $e^{ \phi_1+\delta_1 }$                                                                             & $e^{\phi_1+\delta_1}$\\
        $f_\theta$              & $X_2Y^1$               & $\left(-\frac{1}{4}J^2_{\delta_1}  + T \right)e^{-(\delta_1 + \phi_1)}$                             & $\left(-\frac{1}{4}J^2_{\delta_1}  + T \right)e^{-(\delta_1 + \phi_1) }$ \\
        $h_\theta$              & $(Y^1X_1-Y^2X_2)$      & $-J_{\phi_1}$                                                                                       & $-J_{\phi_1}$ \\
        $e^a$                   & $X_1Y^{a+2}$           & $e^{\phi_{a+1}+\delta_{a+1} }$                                                                      & $\boldsymbol{\beta}^a \boldsymbol{d}_\theta$\\
        $\widetilde{e}_a$       & $X_{a+2}Y^2$           & $\left(- J_{\delta_{a+1}} +J_\nu \right)e^{-(\phi_{a+1}+\delta_{a+1}) + (\phi_1 + \delta_1) }$      & $\boldsymbol{\gamma}_a \boldsymbol{d}_\theta$ \\
        $f^a$                   & $X_2Y^{a+2}$           & $\left(- J_{\delta_{1}} +J_\nu \right)e^{(\phi_{a+1}+\delta_{a+1}) - (\phi_1 + \delta_1) }$                                              & $\boldsymbol{\beta}^a \boldsymbol{d}_\theta^{-1}$             \\
        $\widetilde{f}_a$       & $X_{a+2}Y^{1}$         & $J_\nu   \left(- J_{\delta_{a+1}}  +J_\nu \right)e^{-\phi_{a+1}-\delta_{a+1} }$                     & $J_\nu\boldsymbol{\gamma}_a\boldsymbol{d}^{-1}_\theta$\\
        $c_a^b$                 & $X_{a+2}{Y}^{b+2}$     &  $\left(- J_{\delta_{a+1}}+J_\nu \right)e^{-(\phi_{a+1}+\delta_{a+1})+(\phi_{b+1}+\delta_{b+1})}$   & $\boldsymbol{\gamma}_a\boldsymbol{\beta}^b$  \\
        $t_{a}$                 & $(- \frac12 Y^1X_1 - \frac12 Y^2X_2 + Y^{a}X_a)$   &   $-J_{\phi_a} -\frac12J_{\phi_1}$                                        &    $ J_\nu + \boldsymbol{\beta}^a \boldsymbol{\gamma}_a +\frac12 J_{\delta_1}   $ \\[2pt] 
        \hdashline &&& \\[\dimexpr-\normalbaselineskip+2pt]
                                & $\chi_{j}{Y}^{1}$      &  $J_\nu e^{-\frac{1}{N} \sum_{l=1}^{N-1}( \phi_l + \delta_l )} \psi^j$                  & $J_\nu\boldsymbol{d}_\theta^{-1}\Psi^j$  \\
                                & $\chi_{j}{Y}^{2}$    &  $e^{-\frac{1}{N} \sum_{l=1}^{N-1}( \phi_l + \delta_l )+(\phi_{1}+\delta_{1})} \psi^j$  & $\boldsymbol{d}_\theta \Psi^j$  \\
                                & $\chi_{j}{Y}^{2+a}$    &  $e^{-\frac{1}{N} \sum_{l=1}^{N-1}( \phi_l + \delta_l )+(\phi_{a+1}+\delta_{a+1})} \psi^j$  & $\boldsymbol{\beta}_a\Psi^j$  \\
                                & $X_1{\xi}^{j}$         &  $e^{\frac{1}{N} \sum_{l=1}^{N-1}( \phi_l + \delta_l )} \widetilde{\psi}_j$                                     & $\boldsymbol{d_\theta}\widetilde{\Psi}_j$  \\
                                & $X_{2}\xi^j$         &  $\left(- J_{\delta_{1}} +J_\nu \right)e^{\frac{1}{N} \sum_{l=1}^{N-1}( \phi_l + \delta_l )-(\phi_{1}+\delta_{1})} \widetilde{\psi}_j$  & $\left(- J_{\delta_{1}} +J_\nu \right) \boldsymbol{d}^{-1}_\theta \widetilde{\Psi}_j$  \\
                                & $X_{2+a}\xi^j$         &  $\left(- J_{\delta_{a+1}} +J_\nu \right)e^{\frac{1}{N} \sum_{l=1}^{N-1}( \phi_l + \delta_l )-(\phi_{a}+\delta_{a})} \widetilde{\psi}_j$  & $\boldsymbol{\gamma}_a\widetilde{\Psi}_j$  \\
    \end{tabular}
    }
    \caption{\label{tab:SUN}Relation between free field realisations of $V_{1}(\mathfrak{psl}(N|N))$. Here $a,b \in \{1, \cdots N-2 \}$, $i,j \in \{1,\cdots N\}$ and we are assuming $a\neq b$. The top block corresponds to  $V_{-1}(\mathfrak{sl}_N)$. $f^a$, $\widetilde{f}_a$ are the generators corresponding to a basis of $\mathfrak{R}^-$, $c_a^b$, $t_a$ be generators of $\mathfrak{sl}_{(N-2)}\oplus\mathbb{C}h_{\perp}$. The second block corresponds to the odd generators. We omit the copy of $V_1(\slf_N)$ generated by free fermions, which are standard.}
\end{table}

By further redefining the free fermions $\psi^i$, $\wt{\psi}_j$, we can absorb much of the remaining dependence on the lattice bosons. Indeed, if we set
\begin{equation}
\begin{split}
    \Psi^j ~\coloneqq~ e^{-\frac{1}{N}\sum_{l=1}^{N-1}( \phi_l + \delta_l ) + \frac12 (\delta_1+\phi_1)} \psi^j~,\\
    \widetilde{\Psi}_j ~\coloneqq~ e^{\frac{1}{N}\sum_{l=1}^{N-1}( \phi_l + \delta_l ) - \frac12 (\delta_1+\phi_1)} \widetilde{\psi}_j~,
\end{split}
\end{equation}
then these still enjoy free fermion OPEs but also commute with $\boldsymbol{\beta}^a$, $\boldsymbol{\gamma}_a$, $\boldsymbol{d}_\theta$, $h_\theta$. Moreover, we now have
\begin{equation}
    J_\nu = -\frac{N}{2} J_{\Psi} + \frac12  \boldsymbol{\beta}^a\boldsymbol{\gamma}_a + \frac12 J_{\delta_1}~.
\end{equation}
where we have defined
\begin{equation}
    J_\Psi \coloneqq \frac{1}{N}\sum_{i=1}^N \Psi^i \widetilde{\Psi}_i~.
\end{equation}
Now let us consider the form of the lowering operator in $\slf(2)_\theta$,
\begin{equation}
    f_\theta = \left(-\frac{1}{4}J^2_{\delta_1}  + T \right)e^{-(\delta_1 + \phi_1)}
\end{equation}
where
\begin{equation}\label{eq:sln-stress-1}
    T = \left(J_\nu -\frac12J_{\delta_1}\right)\left( J_\nu-\frac12 J_{\delta_1}\right)~,
\end{equation}
which we can now rewrite in our new variables as
\begin{equation}
    T = \left(-\frac{N}{2} J_{\Psi} + \frac12  \boldsymbol{\beta}^a\boldsymbol{\gamma}_a\right)^2~.
\end{equation}
So indeed, we have $J_\Psi$ acting as the additional Heisenberg current compared to the free field realisation of \cite{Beem:2019tfp}.

More generally it can be verified that all generators of $V_{1}(\mathfrak{psl}(N|N))$ can in fact be expressed in terms of these de-bosonised free fields. We therefore obtain a free field realisation of $V_{1}(\mathfrak{psl}(N|N)$ modelled on a $T^\ast \mathbb{C}^\times \times T^\ast  \mathbb{C}^{N-2}$ patch on the associated variety. Details of this and the comparison between the different free field realisations is provided in Table~\ref{tab:SUN}.

\subsection{\texorpdfstring{$A_{N-1}$}{A[N-1]} quiver example}

Finally, we consider the mirror of SQED[$N$], for which we take the charge matrices defined in~\eqref{eq:charge-AN}
\begin{equation}
    \mathbf{Q} =  \left(
    \begin{array}{ccccc}
    -1 & 1 & 0 &  \cdots & 0 \\
     \vdots &  &  & \ddots &  \\
     -1 & 0 & 0 & \cdots & 1 \\
     -1 & 0 & 0 & \cdots & 0 \\
    \end{array}
    \right)~,\qquad \widetilde{\mathbf{Q}} =\left(
    \begin{array}{ccccc}
     0 & -1 & 0 &  \cdots & 0 \\
     \vdots &  &  & \ddots &  \\
     0 & 0 & 0 & \cdots & -1 \\
     -1 &- 1 & -1 & \cdots & -1 \\
    \end{array}
    \right)~.
\end{equation}
The charge matrix $Q$ encodes the $A_{N-1}$ quiver. The VOA for $N=3$ has recently been studied in~\cite{Yoshida:2023wyt} where the OPEs were constructed directly.

The theory has a $\mathfrak{su}(N)_C$ Coulomb branch symmetry (the symmetry of the $\slf_N$ minimal nilpotent orbit) which enhances the $\mathfrak{u}(1)_C^{N-2}$ topological symmetry visible in the UV. Per our discussion in Section~\ref{subsec:outer_out}, we expect two choices of bosonisation to make this symmetry manifest as outer automorphism symmetry of the free field realisation: $(++\cdots+)$ and $(--\cdots -)$. In terms of charts on the Higgs branch, as described in Section~\ref{subsec:fin-geom-min-nil}, these choices indeed correspond to localisation on an open $T^\ast \mathbb{C}^\times \subset \mathcal{M}_H $ of the singular Higgs branch.

Below, after reviewing the prescribed change of basis, we list the expected generators in terms of free fields in the $\epsilon = (++\cdots +)$ bosonisation and observe the manifest $\mathrm{SL}(N)_o$ outer automorphism symmetry. Finally, specialise to $N=3$ case and study the VOA in more detail in Section~\ref{subsubsec:A3-VOA}.

\subsubsection{\label{subsubsec:AN-COB}Change of coordinates}

We compute the relevant matrices for the change of basis \eqref{eq:gen-COB}:
\begin{equation}
   A_\epsilon^{-1} =
  \frac{1}{N}\left(
    \begin{array}{ccccc}
    -1 & -1 & \cdots & -1 & 1 \\
     N-1 & -1 & \cdots & -1 &1 \\
     \vdots & \ph{\ddots} &  & \vdots & \vdots \\
     -1 &  & N-1 & -1 &1 \\
     -1 & -1 & \cdots  & N-1 & 1 \\
    \end{array}
    \right)
    ~,~
    F_\epsilon =\frac{1}{N} \left(
    \begin{array}{ccccc}
     N-1 & -1 & \cdots & -1 &- 1 \\
     -1 & N-1 & \cdots & -1 & -1 \\
    \vdots & \vdots & \ddots  & \vdots & -1 \\
     -1 & -1 & -1 & \cdots & N-1 \\
    \end{array}
    \right)~.
\end{equation}
Thus, the new basis for the lattice vertex algebra will be
\begin{equation}
\begin{split}
    \phi            &= -\sum_{i=1}^N \sigma_i~,\\
    \delta          &= \sum_{i=1}^N \rho_i~,\\
    \omega_i        &= \gamma_i+(\sigma_i - \rho_i)+\frac{1}{N}\sum_{j=1}^N (\rho_j-\sigma_j)~.
\end{split}
\end{equation}
The fields have the normalisation
\begin{equation}
    ( \delta , \delta ) = - ( \phi , \phi ) = N~,\qquad ( \omega_i , \omega_j) = \delta_{ij}~,
\end{equation}
and the lattice extension of this algebra which defines $\wt{V}_{\rm FFR}$ is given by
\begin{equation}
    \mathbb Z \left(\frac{\delta+\phi}{N}\right) \oplus \bigoplus_{i=1}^N \mathbb{Z}\,\omega_i~.
\end{equation}
We define free fermions
\begin{equation}
    \psi^i \coloneqq e^{-\omega_i}~,\qquad~\widetilde{\psi}_i \coloneqq e^{\omega_i}~,
\end{equation}
so the free field space is the tensor product of a half lattice $\mathcal{D}^{(ch)}(\mathbb{C}^\times)$ with $N$ free fermions as expected. With an eye towards outer automorphisms, we adopt the notation $\psi^A$, $\wt\psi_A$ where $A,B,C\ldots = 1,\cdots N$ are now indices fundamental/anti-fundamental representations of $\mathrm{SL}(N)$. We introduce the $\mathrm{SL}(N)$ invariant (trace) current,
\begin{equation}\label{eq:a2-fcurr}
    J_\psi = \psi^A \widetilde{\psi}_{A}~.
\end{equation}
The BRST representatives $\bar{J}_{\rho_i}$ of the currents $J_{\rho_i}$ read as follows
\begin{equation}\label{eq:a2-drho}
    \bar{J}_{\rho_i} = \tfrac{1}{N}\left( J_\delta - N \psi^i \widetilde{\psi}_i + J_\psi \right)~.
\end{equation}
The corresponding screening currents are given by
\begin{equation}
    \tilde{\mathfrak{s}}_i = \exp\left(\tfrac{1}{N}\delta-\tfrac{1}{N}\sum_{j=1}^N\omega_j+\omega_i\right)~,
\end{equation}
and indeed these transform in the standard representation of horizontal $\slf_N$ subalgebra of the free fermionic $V_1(\slf_N)$ that commutes with $J_\psi$.

\subsubsection{\label{subsubsec:AN-gen}\texorpdfstring{$A_{N-1}$}{A[N-1]} generators and outer automorphism symmetry}

In this short section, we briefly outline the list of generators we expect in the $A_{N-1}$ case. From the flavour moment map and its fermion counterpart, we will have two Heisenberg currents,
\begin{equation}
\begin{split}
    \sum_{i=1}^N X_iY^i &~~~\mapsto~~ -J_\phi~,\\
    \sum_{i=1}^N \chi_i\xi^i &~~~\mapsto~~ \ph{J_\psi}~.
\end{split}
\end{equation}
The $N$ Grassmann-odd currents take the form
\begin{equation}\label{eq:an-free-currents-2}
\begin{alignedat}{3}
    &X_i\xi^i &&~~~\mapsto~~~ J^+_{A} &&~~\coloneqq~~   e^{\frac{1}{N}(\phi + \delta)} \widetilde{\psi}_A~,\\
    &\chi_iY^i &&~~~\mapsto~~~ J^{-,A} &&~~\coloneqq~~\left(\tfrac{1}{N}\left( - J_{\delta} - J_\psi \right)\psi^A+\partial {\psi^A}\right)e^{-\frac{1}{N}(\phi + \delta)} ~.
\end{alignedat}
\end{equation}
Finally, the additional generators are realised by starting with $X_1X_2\cdots X_N$ and $Y^1Y^2\cdots Y^N$ and performing the substitutions $X_i\leftrightarrow \chi_i$, $Y^i \leftrightarrow \xi^i$. We obtain operators of the form,\footnote{To avoid confusion, we remind the reader that here $A,B,\cdots$ denote fundamental indices of the $\mathrm{SL}(N)_o$ outer automorphism symmetry.}
\begin{equation}
    \begin{alignedat}{3}
        &W &&~~\coloneqq~~ &&e^{\phi+\delta}~,\\
        &W^A &&~~\coloneqq~~  &&\psi^A e^{\frac{N-1}{N}(\phi+\delta)}~,\\
        &W^{AB} &&~~\coloneqq~~ &&\psi^A\psi^B  e^{\frac{N-2}{N}(\phi+\delta)}~,\\
        &\cdots && &&\cdots~.
     \end{alignedat}
\end{equation}
as well as
\begin{equation}\label{eq:An-non-cov}
    \begin{alignedat}{3}
        &Z &&~~\coloneqq~~ &&e^{-\phi-\delta}~,\\
        &Z_A &&~~\coloneqq~~ &&\prod_{i\neq A}\left(-\frac{1}{N}(J_\delta + J_\psi) + \psi^i\widetilde{\psi}_i \right) \widetilde{\psi}_A   e^{-\frac{N-1}{N}(\phi+\delta)}~,\\
        &Z_{AB} &&~~\coloneqq~~ &&\prod_{i\not \in \{A,B\}} \left(-\frac{1}{N}(J_\delta + J_\psi) + \psi^i\widetilde{\psi}_i \right) \widetilde{\psi}_A\widetilde{\psi}_B   e^{-\frac{N-2}{N}(\phi+\delta)}~,\\
        &\cdots && &&\cdots~.
     \end{alignedat}
\end{equation}
To see that these operators enjoy an $\mathrm{SL}(N)_o$ outer automorphism symmetry, it is sufficient to note that the affine currents~\eqref{eq:an-free-currents-2} have been written in a manifestly covariant form, and the monomial operators given here appear in first-order poles when taking repeated OPEs of the singlets $W$ and $Z$ with $J^{-,A}$ and $J^+_A$, respectively. This means, in particular, that the expressions in~\eqref{eq:An-non-cov} that are not transparently covariant nevertheless always can be massaged into a covariant form.

\subsubsection{\label{subsubsec:A3-VOA}Free field realisation for the $A_{2}$ quiver gauge theory}

Finally, we include more detail in the special case $N=3$. In this case, we  can complete the list of $Z$ and $W$-type generators and put them in covariant form as expected. In particular, we have the very simple expressions for the $W$'s,
\begin{equation}
\begin{alignedat}{3}
        &W &&~~\coloneqq~~ &&e^{\phi+\delta}~,\\
        &W^A &&~~\coloneqq~~  &&\psi^A e^{\frac23(\phi+\delta)}~,\\
        &W^{AB} &&~~\coloneqq~~ &&\psi^A\psi^B  e^{\frac13(\phi+\delta)}~,\\
        &W^{ABC} &&~~\coloneqq~~ &&\psi^A\psi^B\psi^C~.
\end{alignedat}
\end{equation}
Moreover, we have the following for the $Z$'s,
\begin{equation}
\begin{alignedat}{3}
        &Z &&\coloneqq~ &&e^{-\phi-\delta}~,\\
        &Z_A &&\coloneqq~ &&\left( \tfrac{1}{9} (J_\delta + J_\psi)(J_\delta -2 J_\psi)\widetilde{\psi}_A -\tfrac{1}{3}(J_\delta +J_\psi)\widetilde{\psi}'_A -\tfrac{1}{12}\epsilon_{ABC} \epsilon^{DEF} \psi^B\psi^C \widetilde{\psi}_D\widetilde{\psi}_E\widetilde{\psi}_F   \right) e^{-\frac{2}{3}(\phi+\delta)}~,\\
        &Z_{AB} &&\coloneqq~ &&\left(-\tfrac{1}{3}\left(J_\delta + J_\psi \right) \widetilde{\psi}_A \widetilde{\psi}_B + \tfrac{1}{6} \epsilon_{ABC}\epsilon^{DEF} \psi^C \widetilde{\psi}_D \widetilde{\psi}_E \widetilde{\psi}_F\right) e^{-\frac13(\phi+\delta)}~,\\
        &Z_{ABC} &&\coloneqq~ &&\widetilde{\psi}_A \widetilde{\psi}_B \widetilde{\psi}_C~.
\end{alignedat}
\end{equation}
The $\mathrm{SL}(3)_o$ outer automorphism symmetry allows us to express the OPEs amongst these somewhat compactly. We report here some of the non-vanishing ones 
\begin{equation}
\begin{split}
    J_{\phi}(z) W^{AB\cdots}(w) &~\sim~ \frac{-{(3-|AB\cdots|)}W^{AB\cdots}}{z-w}~,\\
    J_{\phi}(z) Z_{AB\cdots}(w) &~\sim~ \frac{+(3-|AB\cdots|) Z_{AB\cdots}}{z-w}~.
\end{split}
\end{equation}
Here $|AB\cdots|$ simply denotes the number of indices. Then, we have
\begin{equation}
\begin{split}
    J^{+,A} (z)W^{BC\cdots} (w) &~\sim~ \frac{-W^{ABC\cdots}}{z-w}~,\\
    J^{-}_A (z)W^{BC\cdots} (w) &~\sim~ \frac{\delta_A^B W^{C \cdots} + \delta_{A}^C W^{B \cdots}+\cdots}{z-w}~,
\end{split}
\end{equation}
as well as
\begin{equation}
\begin{split}
    J_A^- (z) Z_{BC\cdots} (w) &~\sim~ \frac{Z_{ABC\cdots}}{z-w}~,\\
    J^{+,A} (z)Z_{BC\cdots} (w) &~\sim~ \frac{\delta^A_B Z_{C \cdots} + \delta^{A}_C Z_{B \cdots}+\cdots}{z-w}~.
\end{split}
\end{equation}
The OPEs amongst cubic operators are slightly more complicated. For example, we have
\begin{equation}
\begin{split}
    W(z) Z(w) &~\sim~ \frac{1}{(z-w)^3} + \frac{J_{\phi}}{(z-w)^2} +\frac{1}{(z-w)} \left(\frac{1}{2}\partial J_\phi
    +\frac{1}{3}J_\phi J_\phi
    +\frac{1}{6}J_\phi J_\psi 
    +\frac{1}{6}J_\psi J_\phi\right)\\
    &
    +\frac{1}{(z-w)}\left(-\frac{1}{6}J_\psi J_\psi +\frac{1}{2} \sum_{A=1}^3 \left( J^{+,A}J^-_{A} - J_{A}^-J^{A,+}\right)\right)~.
\end{split}
\end{equation}
One may checked that the OPEs then close amongst the selected generators.

\subsubsection{\label{subsubsec:stress_tensor}Stress tensors}

In particular, the above VOA does not require a separate stress tensor as a generator. Indeed, there is a Sugawara-type expression for the stress tensor in terms of the various currents. First, we have the stress tensors associated to the free fields,
\begin{equation}
\begin{split}
    T_\phi &~\coloneqq~ -\frac{1}{6}J_\phi J_\phi~,\\
    T_\delta &~\coloneqq~ \frac{1}{6}\left( J_\delta J_\delta -3 J_\delta\right)~,\\
    T_{\psi^i} &~\coloneqq~ \frac{1}{2}\left( {\partial \psi^i}\widetilde{\psi}_i -\psi^i \partial \widetilde{\psi}_i\ \right)~,
\end{split}
\end{equation}
where the background charge for $J_\delta$ ensures that the vertex operators have the correct dimension. It can then be checked explicitly that in terms of the free fields, we have the identity
\begin{equation}
    T_\phi + T_\delta + \sum_{i=1}^3 T_{\psi^i} = T_{\text{Sug}}~,
\end{equation}
where
\begin{equation}
    T_{\text{Sug}} \coloneqq -\frac{1}{6}\left( J_\phi J_\psi+J_\psi J_\phi\right)
    +\frac{1}{3} J_\psi J_\psi
    -\frac{1}{2} \sum_{i=A}^3 \left( J^{+,A} J_{A}^- - J_{A}^- J^{+,A} \right)~.
\end{equation}
This provides the required Sugawara-type construction in the $A_2$ case.

\acknowledgments

The authors would like to thank Dylan Butson, Davide Gaiotto, Andr\'e Henriques, Justin Hilburn, Sujay Nair, and Ben Webster for helpful discussions about this work and related topics. They would also like to thank Andrew Ballin, Thomas Creutzig, Tudor Dimofte and Wenjun Niu for coordinating submission of a paper that studies the same $A$-twisted vertex algebras discussed here. Christopher Beem's work is supported in part by grant \#494786 from the Simons Foundation, by ERC Consolidator Grant \#864828 ``Algebraic Foundations of Supersymmetric Quantum Field Theory'' (SCFTAlg), and by the STFC consolidated grant ST/T000864/1. Andrea Ferrari's work is supported by the EPSRC Grant EP/T004746/1 ``Supersymmetric Gauge Theory and Enumerative Geometry".

\appendix

\section{\label{app:psl22-conv}Conventions for the structure of (affine) \texorpdfstring{$\pslf(2|2)$}{psl(2,2)}}

In this appendix we establish our notation for the Lie superalgebra $\pslf(2|2)$ and its $\mathrm{SL}(2)_o$ group of outer automorphisms. 

\subsection{\texorpdfstring{$\pslf(2|2)$}{psl(2-2)}}

Starting with the larger superalgebra $\glf(2|2)$, we denote by $\ell_a^{\phantom{a}b}$ and $\scriptr_{\alpha}^{\phantom{\alpha}\beta}$ a basis for the $\glf(2)\times \glf(2)$ (Grassmann-)even subalgebra. In reference to the realisation of these subalgebras as symmetries in the $T[\mathrm{SU}(2)]$ theory, we will refer to the subalgebra spanned by the $\ell_a^{\phantom{a}b}$ as the \emph{bosonic subalgebra}, and that spanned by $\scriptr_{\alpha}^{\phantom{\alpha}\beta}$ as the \emph{fermionic subalgebra}. As a basis for the (Grassmann-)odd part of the algebra we take $\vartheta_{a}^{\phantom{a}\alpha}$, $(\widetilde{\vartheta})^b_{\phantom{b}\beta}$. In matrix form, these are assembled as follows,
\begin{equation}\label{eq:psl22matrix}
    M=\left(
    \begin{array}{cc|cc}
        ~\ell_1^{\phantom{1}1}~ & ~\ell_1^{\phantom{1}2}~ & ~\vartheta_1^{\phantom{1}1}~ & ~\vartheta_1^{\phantom{1}2}~ \\
        ~\ell_2^{\phantom{1}1}~ & ~\ell_2^{\phantom{1}2}~ & ~\vartheta_2^{\phantom{1}1}~ & ~\vartheta_2^{\phantom{1}2}~ \\[1pt]
        \hline
        &&&\\[\dimexpr-\normalbaselineskip+2pt]
        ~\tilde{\vartheta}^1_{\phantom{1}1}~ & ~\tilde{\vartheta}^2_{\phantom{2}1}~ & ~\scriptr_1^{\phantom{1}1}~ & ~\scriptr_1^{\phantom{1}2}~ \\
        ~\tilde{\vartheta}^1_{\phantom{1}2}~ & ~\tilde{\vartheta}^2_{\phantom{2}2}~ & ~\scriptr_2^{\phantom{1}1}~ & ~\scriptr_2^{\phantom{1}2}~ \\
    \end{array}
    \right)~.
\end{equation}
The subalgebra which includes only super-traceless combinations of diagonal elements is identified with the simple Lie superalgebra $\slf(2|2)$, and one can define the super-traceless version of the matrix elements of \eqref{eq:psl22matrix} according to
\begin{equation}
    \widehat{M}_a^{\phantom{a}b} = M_{a}^{\phantom{a}b}-\frac{1}{4}\delta_a^{\phantom{a}b}\,{\rm Str}\,M~.
\end{equation}
Further taking the quotient by the ideal generated by the \emph{ordinary} trace defines $\pslf(2|2)$. In this quotient, one has that the $\ell_a^{\phantom{a}b}$ and $\scriptr_\alpha^{\phantom{a}\beta}$ generate $\slf_2$  even subalgebras (in particular, $\ell_1^{\phantom{a}1}+\ell_2^{\phantom{a}2}=\scriptr_1^{\phantom{a}1}+\scriptr_2^{\phantom{a}2}=0$). Adopting identical notation as before for our matrix elements, but \emph{now denoting their avatars in $\pslf(2|2)$}, we have the following non-vanishing brackets,
\begin{equation}\label{eq:psl2_commutators}
    \begin{split}
        [\ell_a^{\phantom{a}b},X_c] &= +\delta_c^b X_a - \frac{1}{2}\delta_a^b X_c~,\\
        [\ell_a^{\phantom{a}b},Y^c] &= -\delta_a^c Y^b + \frac{1}{2}\delta_{a}^b Y^c~,\\
        [\scriptr_\alpha^{\phantom{a}\beta}, X_\gamma] &= +\delta_\gamma^\beta Y^\alpha - \frac{1}{2}\delta_\alpha^\beta X_\gamma~,\\
        [\scriptr_\alpha^{\phantom{a}\beta}, Y^\gamma] &= -\delta^\gamma_\alpha Y^\beta + \frac{1}{2}\delta_\alpha^\beta X_\gamma~,\\
        \{\vartheta_{a}^{\phantom{a}\alpha}, \tilde{\vartheta}^b_{\phantom{b}\beta} \} &= +\delta_{a}^b \scriptr_{\alpha}^{\phantom{a}\beta} + \delta_{\alpha}^\beta \ell_a^{\phantom{a}b}~, 
    \end{split}
\end{equation}
with the other anti-commutators involving $\vartheta_{\alpha}^a$ and $\tilde{\vartheta}_b^\beta$ being zero. Here $X$ is a placeholder for any elements of the algebra with the appropriate index structure.

This superalgebra enjoys the action of an $\mathrm{SL}(2)_o$ outer automorphism group. This acts only on the Grassmann odd part of the algebra, and mixes the upper-right-hand and lower-left-hand blocks of \eqref{eq:psl22matrix}. In particular, if we define the following alternative notation for the odd entries in the lower-left-hand block of the matrix in \ref{eq:psl22matrix},
\begin{equation}
    (\vartheta^{(1)})_{a}^{\phantom{a}\alpha} = \vartheta_{a}^{\phantom{a}\alpha}~,\qquad
    (\vartheta^{(2)})_{a}^{\phantom{a}\alpha} = \epsilon_{ab}(\tilde{\vartheta})^b_{\phantom{b}\beta}\epsilon^{\beta\alpha}~,
\end{equation}
then the $(\vartheta^{(A)})_a^{\phantom{a}\alpha}$ transform in the standard representation of $\mathrm{SL}(2)_o$ with $A$ being a fundamental index. The last (anti)commutator in \eqref{eq:psl2_commutators} then becomes
\begin{equation}
    \{(\vartheta^{(A)})_a^{\phantom{a}\alpha},(\vartheta^{(B)})_b^{\phantom{a}\beta}\}=\epsilon^{AB}\left(\epsilon^{\alpha\beta}\ell_{ab}+\epsilon_{ab}\scriptr^{\alpha\beta}\right)~.
\end{equation}
It is a simple exercise to check that the defining commutation relations are invariant under these outer automorphism rotations.

\subsection{\texorpdfstring{$V^k(\pslf(2|2))$}{V(k)(psl(2|2))}}

For the current algebra corresponding to the $\pslf(2|2)$ Lie superalgebra, we promote $\ell$, $\scriptr$, and $\vartheta$ to affine currents $L$, $J$, and $\Theta$, with non-vanishing singular OPEs (at level $k$) given by
\begin{equation}\label{eq:psl22_opes_all}
\resizebox{\textwidth}{!}{$%
\begin{aligned}
    L_{ab}(z)L_{cd}(w)&~\sim~\frac{\frac12 k(\epsilon_{bc}\epsilon_{ad}+\epsilon_{ac}\epsilon_{bd})}{(z-w)^2} - \frac{\frac12\left(\epsilon_{ac}L_{bd}(w)+\epsilon_{bc}L_{ad}(w)+\epsilon_{ad}L_{bc}(w)+\epsilon_{bd}L_{ac}(w)\right)}{z-w}~,\\
    J^{\alpha\beta}(z)J^{\gamma\delta}(w)&~\sim~\frac{-\frac12 k(\epsilon^{\beta\gamma}\epsilon^{\alpha\delta}+\epsilon^{\alpha\gamma}\epsilon^{\beta\delta})}{(z-w)^2} - \frac{\frac12\left(\epsilon^{\alpha\gamma}J^{\beta\delta}(w)+\epsilon^{\beta\gamma}J^{\alpha\delta}(w)+\epsilon^{\alpha\delta}J^{\beta\gamma}(w)+\epsilon^{\beta\delta}J^{\alpha\gamma}(w)\right)}{z-w}~,\\
    \Theta^{\alpha}_{Aa} \Theta^{\beta}_{Bb} &~\sim~ -\frac{k\epsilon_{AB}\epsilon^{\alpha\beta}\epsilon_{ab}}{(z-w)^2}+\frac{\epsilon_{AB}\left(\epsilon_{ab}J^{\alpha\beta}(w)+\epsilon^{\alpha\beta}L_{ab}(w)\right)}{z-w}~,\\ 
    L_{ab}(z)\Theta^{\gamma}_{Ac}&~\sim~-\frac{\epsilon_{bc}\Theta^{\gamma}_{Aa}(w)+\epsilon_{ac}\Theta^{\gamma}_{Ab}(w)}{z-w}~,\\
    J^{\alpha\beta}(z)\Theta^{\gamma}_{Aa}&~\sim~-\frac{\epsilon^{\beta\gamma}\Theta^{\alpha}_{Aa}(w)+\epsilon^{\alpha\gamma}\Theta^{\beta}_{Aa}(w)}{z-w}~.
\end{aligned}
$}
\end{equation}
Here the trace parts of the two $\slf_2$ subalgebras are absent, so $L_{12}=L_{21}$ and $J^{12}=J^{21}$, and we recall our conventions that $\epsilon^{12}=\epsilon_{21}=1$.

\section{\label{app:lattice_bosons}Lattice vertex algebras and bosonisations of elementary fields}

In this appendix, we establish notation and recall standard results used in the discussion of bosonisation and lattice vertex algebras used in the analysis of VOAs for Abelian gauge theories in three dimensions.

\subsection{\label{subapp:lattice_VA}Lattice vertex algebras}

Given an Abelian Lie algebra $\mathfrak{h}$ equipped with a bilinear form $(,)$, let $\widetilde{\hf}\cong\mathfrak{h}[t,t^{-1}]$ denote the loop algebra for $\mathfrak{h}$, and let $\widehat{\mathfrak{h}}$ denote its central extension according to the given bilinear form, with underlying vector space $\widetilde{\mathfrak{h}}\oplus\mathbb{C}K$ and Lie bracket
\begin{equation}
    [at^m,bt^n]= mK\delta_{m+n,0}(a,b)~.
\end{equation}
We denote the graded components (with loop grading) of $\widehat{\mathfrak{h}}$ by $\mathfrak{h}_n$, $n\in\mathbb{Z}$,
\begin{equation}
    \mathfrak{h}_{n\neq0}={\rm span}\{at^n~|~a\in\mathfrak{h}\}~,\qquad \mathfrak{h}_0={\rm span}\{K; a~|~a\in\mathfrak{h}\}~.
\end{equation}
and we have the grade decomposition
\begin{equation}
\widehat{\mathfrak{h}}=\bigoplus_{i\in\mathbb{Z}}\mathfrak{h}_n~.
\end{equation}
It proves useful to introduce various (Abelian) subalgebras associated with the semi-infinite structure on $\widehat{\mathfrak{h}}$ (\cf\ \cite{voronov1994semi}),
\begin{equation}\label{eq:semi-infinite}
\begin{split}
    \mathfrak{b}_-&=\bigoplus_{i\leqslant0}\mathfrak{h}_i~,\qquad \mathfrak{b}_+=\bigoplus_{i\geqslant0}\mathfrak{h}_i~,\\
    \mathfrak{n}_-&=\bigoplus_{i<0}\mathfrak{h}_i~,\qquad \mathfrak{n}_+=\bigoplus_{i>0}\mathfrak{h}_i~.
\end{split}
\end{equation}
From $\widehat{\mathfrak{h}}$ one constructs the corresponding Heisenberg vertex operator algebra/Abelian affine current algebra with underlying vector space (vacuum module)
\begin{equation}
    V_{\mathfrak{h}} = \cU(\mathfrak{h})\otimes_{\cU(\mathfrak{b}_+)}\mathbb{C}_{1,0}~,
\end{equation}
where $\mathbb{C}_{1,0}$ is the one-dimensional $\mathfrak{b}_+$-module on which $K$ acts as the identity and all other basis elements act as zero. This vertex algebra is strongly generated by affine currents $\{J_{a_i}(z)\}$, where the $a_i$ form a basis for $\mathfrak{h}$, with singular OPEs taking the form
\begin{equation}
    J_a(z)J_b(w)\sim \frac{(a,b)}{(z-w)^2}~.
\end{equation}
This vacuum module is the simplest example of a Fock module for $\hhat$.
\begin{definition}[Fock module]
    For any $\lambda\in\mathfrak{h}^\ast$, we denote by $\mathcal{F}_\lambda$ the \emph{Fock module of charge $\lambda$}, which is the induced module
    \begin{equation}
        \mathcal{F}_\lambda=\cU(\mathfrak{h})\otimes_{\cU(\mathfrak{b}_+)}\mathbb{C}_{1,\lambda}~,
    \end{equation}
    where the action of $\hf_0$ in the one-dimensional representation $\mathbb{C}_{1,\lambda}$ is given by $\lambda$. We will usually abuse notation and write Fock modules with $\lambda\in\mathfrak{h}$ with the corresponding element of $\mathfrak{h}^\ast$ arising via the bilinear form on $\mathfrak{h}$.
\end{definition}

Certain collections of Fock modules can be given the structure of a vertex operator algebra or a vertex operator superalgebra, extending the vertex algebra structure on $V_\mathfrak{h}$.\footnote{This sketch of a definition is primarily intended to establish notation; for complete details on the construction of the vertex algebra structure on $V_L$, see, \emph{e.g.}, \cite{kac1998vertex}.}

\begin{definition}[Lattice vertex (super)algebra]
    For any integral lattice (with respect to the bilinear form $(,)$) $L\subset V$, define the $V_\mathfrak{h}$ module
    \begin{equation}
        M_L = \bigoplus_{\lambda\in L}\mathcal{F}_\lambda=\bigoplus_{\lambda\in L}V_{\mathfrak{h}}\cdot e^{\lambda}~.
    \end{equation}
    Then $M_L$ can be given the structure of a vertex operator superalgebra, the \emph{lattice vertex (super)algebra} $V_L$. In $V_L$, vectors lying in Fock modules corresponding to points in $L$ with even/odd squared length are assigned even/odd Grassmann parity.
\end{definition}

\subsection{\label{subapp:symp_bosonisation}Bosonisation of the symplectic boson}

The symplectic bosons and free fermions appearing in the gauge theory description of $A$-twisted boundary vertex algebras can both be realised as (subalgebras of) lattice vertex algebras of the type just reviewed. We start with the case of the symplectic bosons vertex algebra following \cite{friedan1986conformal} (see also the recent comprehensive discussion in \cite{allen2022bosonic}).

Consider the integral lattice $L_{\texttt{Sb}} = \left(\mathbb{Z}\rho\oplus\mathbb{Z}\sigma\right)$, where $(\rho,\rho)=1$ and $(\sigma,\sigma)=-1$. Let $\mathfrak{h}_{\texttt{Sb}}=L_{\texttt{Sb}}\otimes_{\mathbb{Z}}\mathbb{C}$; the corresponding Heisenberg vertex algebra $V_{\mathfrak{h}_{\texttt{Sb}}}$ is generated by currents $\{J_\rho, J_\sigma\}$ obeying
\begin{equation}
    J_{\rho}(z)J_{\rho}(w) \sim \frac{1}{(z-w)^2}~,\qquad J_{\sigma}(z)J_{\sigma}(w) \sim \frac{-1}{(z-w)^2}~.
\end{equation}
The corresponding lattice vertex superalgebra has underlying vector space
\begin{equation}\label{eq:lattice_from_Fock}
    V_{L_{\texttt{Sb}}}\coloneqq\bigoplus_{\lambda\in L_{\texttt{Sb}}}\mathcal{F}_{\lambda}~,
\end{equation}
where the vacuum state of a given Fock module is represented by the vertex operator $e^\lambda$.

Next consider the isotropic sublattice of $L_{\texttt{Sb}}$ given by $L_{\texttt{Sb}}^!=\mathbb{Z}(\rho_i-\sigma_i)$. Restricting the Fock modules in \eqref{eq:lattice_from_Fock} to lie in $L_{\texttt{Sb}}^!$ then gives the canonical example of a \emph{half-lattice vertex algebra} \cite{berman2002representations}, which we denote by
\begin{equation}\label{eq:L-ISO}
    V_{L_{\texttt{Sb}}}^!=\bigoplus_{\lambda\in L_{\texttt{Sb}}^!}V_{\mathfrak{h}}\cdot e^{\lambda}~.
\end{equation}

The symplectic boson VOA can be realised within $V_{L_{\texttt{Sb}}}^!$ according to either of two embeddings (\cf\ \cite{allen2022bosonic}),
\begin{equation}\label{eq:bosonisation}
    \begin{pmatrix}
    X\\Y
    \end{pmatrix} ~\substack{\longmapsto\\\mathfrak{i}^+}~
    \begin{pmatrix}
    e^{\rho-\sigma}\\ J_\rho e^{-\rho+\sigma}
    \end{pmatrix}~,\qquad\qquad 
    \begin{pmatrix}
    X\\Y
    \end{pmatrix} ~\substack{\longmapsto\\\mathfrak{i}^-}~
    \begin{pmatrix}
    -J_\rho e^{-\rho + \sigma}\\
    e^{\rho-\sigma}
    \end{pmatrix}~,
\end{equation}
with the internal Heisenberg subalgebra generated by the current
\begin{equation}
    (XY) = \pm J_\sigma~.
\end{equation}
Here the sign on the right hand side is correlated with the embedding choice.

It is instructive to note that this half-lattice vertex algebra is isomorphic to the algebra of global sections of the sheaf of chiral differential operators \cite{MelnikovCDO} on $\mathbb{C}^\times$,
\begin{equation}
    V_{L^!_{\texttt{Sb}}}\cong \mathcal{D}^{ch}(\mathbb{C}^\times)~.
\end{equation}
On the other hand, symplectic bosons themselves are the global sections of CDOs on the affine line,
\begin{equation}
    \texttt{Sb}\cong\mathcal{D}^{ch}(\mathbb{C})~.
\end{equation}
The half lattice vertex algebra can thus be thought of as arising by (chiral) localisation from $\mathbb{C}$ to $\mathbb{C}^\times$, and the free field realisation/bosonisation of the symplectic bosons amounts to corresponding restriction of sections.

For both of the embeddings in \eqref{eq:bosonisation}, we can define the \emph{screening charge} $\mathfrak{s}_{\pm}$ acting on $V_{L_{\texttt{Sb}}^!}$,
\begin{equation}
\begin{split}
    \mathfrak{s}_{\pm}:V_{L^!_{\texttt{Sb}}} &\longrightarrow \bigoplus_{\lambda\in \rho+L_{\texttt{Sb}}^!}\mathcal{F}_\lambda\\
    v&\longmapsto \mathrm{Res}\,e^{\rho(z)}v~.
\end{split}
\end{equation}
Then the image of the inclusions $\mathfrak{i}^\pm$ can be characterised as precisely the kernel of the screening charge \cite{allen2022bosonic},
\begin{equation}
    \texttt{Sb}\cong{\rm Im}\,\mathfrak{i}_\pm = {\rm Ker}\,\mathfrak{s}_\pm~.
\end{equation}

\subsection{\label{subapp:ff_bosonisation}Bosonisation of the free fermion}

Complex free fermions can also famously be bosonised. Take the unit lattice $L_{\texttt{Ff}}=\mathbb{Z}\gamma$ with $(\gamma,\gamma)=1$, the corresponding Heisenberg current $J_\gamma$ obeys
\begin{equation}\label{eq:sfermion}
    J_\gamma(z) J_\gamma(w) \sim \frac{1}{(z-w)^2}~.
\end{equation}
Then the lattice vertex superalgebra $V_{L_{\texttt{Ff}}}$ is isomorphic to the free fermion vertex algebra $\texttt{Ff}$,
\begin{equation}
    V_{L_{\texttt{Ff}}}\cong\texttt{Ff}~,
\end{equation}
with the identification given by determined by either of the following (equivalent) identifications of strong generators,
\begin{equation}\label{eq:bosonisation_of_fermions}
    \begin{pmatrix}
    \chi\\\xi
    \end{pmatrix} ~\substack{\longleftrightarrow\\+}~
    \begin{pmatrix}
    e^{\gamma}\\ e^{-\gamma}
    \end{pmatrix}~,\qquad \mathrm{or}\qquad 
    \begin{pmatrix}
    \chi\\\xi
    \end{pmatrix} ~\substack{\longleftrightarrow\\-}~
    \begin{pmatrix}
    e^{-\gamma}\\
    e^{\gamma}
    \end{pmatrix}~.
\end{equation}

As in the symplectic boson case, the internal Heisenberg current of the free fermion vertex algebra is realised with a relative sign depending on the choice of identification convention,
\begin{equation}
    (\xi\chi) = \pm J_\gamma~.
\end{equation}

\section{\label{app:cob}Change of basis and its inverse}

We derive the inverse of the change of basis first introduced in~\eqref{eq:TSU2-COB} in the special case of $T[\mathrm{SU}(2)]$ and then more generally in~\eqref{eq:gen-COB}. This is given by a matrix
\begin{equation}
     P
    = \left(
        \begin{array}{c:c:c}
        ~~Q~~  &  ~~0~~ & ~~Q~~ \\ \hdashline
        &&\\[\dimexpr-\normalbaselineskip+2pt]
        ~~Q~~  &  ~~-Q~~ & ~~0~~ \\ \hdashline
        &&\\[\dimexpr-\normalbaselineskip+2pt]
        ~~\widetilde{Q}~~  & ~~0~~  & ~~0~~ \\ \hdashline
        &&\\[\dimexpr-\normalbaselineskip+2pt]
        ~~0~~  &  ~~-\widetilde{Q}~~ \ & ~~0~~ \\ \hdashline
        &&\\[\dimexpr-\normalbaselineskip+2pt]
        ~~F~~  &  ~~-F~~ & ~~\mathbf{1}_{N\times N}~~
        \end{array} 
    \right)
\end{equation}
where
\begin{equation}\label{eq:Fdef_app}
    F^T = \begin{pmatrix}
        Q \\ 
        \widetilde{Q}
    \end{pmatrix}^{-1}
    \begin{pmatrix}
        Q \\ 
        0
    \end{pmatrix} 
   ~.
\end{equation}
$Q$ and $\widetilde{Q}$ are orthogonal. Let us introduce
\begin{equation}
    A = \begin{pmatrix}
        Q \\ \widetilde{Q}
    \end{pmatrix}~,~B=\begin{pmatrix}
        Q \\ 
        0
    \end{pmatrix}~,~C=\begin{pmatrix}
        0 \\ 
       \widetilde{Q}
    \end{pmatrix}~,
\end{equation}
so that 
\begin{equation}
    F = A^{-1}B~.
\end{equation}
Using the identities 
\begin{equation}
    BC^T=0~,~A^{-1}B=\mathbf{1}_{N\times N}-A^{-1}C
\end{equation}
it is easy to verify that $F$ is idempotent (in fact, $F^T$ is by definition a projector onto the span of $Q$) and that moreover
\begin{equation}
    FF^T=F~.
\end{equation}
Now swapping the second and third row-blocks in $P$ leads to the matrix
\begin{equation}
    \widetilde{P} =
     \left(
        \begin{array}{c:c:c}
        ~~A~~  &  ~~0~~ & ~~B~~ \\ 
        \hdashline
        &&\\[\dimexpr-\normalbaselineskip+2pt]
        ~~B~~  &  ~~-A~~ \ & ~~0~~ \\ \hdashline
        &&\\[\dimexpr-\normalbaselineskip+2pt]
        ~~(A^{-1}B)^T~~  &  ~~-(A^{-1}B)^T~~ & ~~\mathbf{1}_{N\times N}~~
        \end{array} 
    \right)~.
\end{equation}
This can be thought of as a $2\times 2$ block matrix with the bottom-right block being the identity, and the upper-left block (which is itself a $2 \times 2$ block) being invertible. The inverse can then be computed by repeated use of standard formul\ae for $2\times 2$ block-matrices. The result simplifies thanks to the identities
\begin{equation}
   F^T A^{-1} = FA^{-1} = \left(A^{-1}\right)^{(k)}
\end{equation}
where $\left(A^{-1}\right)^{(k)}$ is a matrix obtained from $A^{-1}$ by setting to zero the last $(N-k)$ columns. We find
\begin{equation}
    \widetilde{P}^{-1} =  \left(
        \begin{array}{c:c:c}
        ~~A^{-1}~~  &  ~~\left(A^{-1}\right)^{(k)}~~ & ~~-F~~ \\ 
        \hdashline
        &&\\[\dimexpr-\normalbaselineskip+2pt]
        ~~\left(A^{-1}\right)^{(k)}~~  &  ~~-A^{-1}+\left(A^{-1}\right)^{(k)}~~ \ & ~~-F~~ \\ \hdashline
        &&\\[\dimexpr-\normalbaselineskip+2pt]
        ~~0~~  &  ~~-\left(A^{-1}\right)^{(k)}~~ & ~~\mathbf{1}_{N\times N}~~
        \end{array} 
    \right)~.
\end{equation}
From here, the inverse $P^{-1}$ is recovered by exchanging columns $\{k,k+1,\ldots,N\}$ with columns $\{N+1,N+2,\ldots,2N-k\}$.

\bibliographystyle{ytphys}
\bibliography{3d_N4_VOA}

\end{document}